\documentclass[10pt, oneside]{article}

\usepackage[margin=0.8in]{geometry}
\usepackage[parfill]{parskip}
\usepackage{amsmath, amssymb, amsthm}
\usepackage{mathtools}
\usepackage{enumitem}
\usepackage{array}
\usepackage{pgfplots}
\usepackage{amssymb}
\pgfplotsset{compat=1.18}
\usepgfplotslibrary{groupplots}
\usepackage{caption}
\usepackage{hyperref}
\usepackage{setspace}
\usepackage{titlesec}
\hypersetup{
  pdftitle={Multitable: A Perfect and Efficient Hash Table from Multicollisions},
  pdfkeywords={hash table, hash map, hash set, multicollision, dictionary}
}

\titlespacing*{\paragraph}{0pt}{0ex plus 1ex minus 0.2ex}{1em}

\newcommand{\appproof}[1]{proof in \hyperref[#1]{Appendix~\ref*{app:proofs}}}

\newtheorem{theorem}{Theorem}
\newtheorem{lemma}{Lemma}
\newtheorem{corollary}{Corollary}
\newtheorem{proposition}{Proposition}
\theoremstyle{definition}
\newtheorem*{model}{Model}
\newtheorem{definition}{Definition}
\theoremstyle{remark}
\newtheorem{remark}{Remark}

\DeclareMathOperator{\Bin}{Bin}
\DeclareMathOperator{\Pois}{Poisson}
\newcommand{\E}{\mathbb{E}}
\newcommand{\PR}{\mathbb{P}}

\newcommand{\1}{\mathbf{1}}
\newcommand{\ind}[1]{\1_{\{#1\}}}
\newcommand{\pos}[1]{\left(#1\right)^{+}}

\title{MultiTable:\\ \large{A Faster Hash Table at any Physical Load Factor up to and Including One}}
\author{Maksym Petkus\footnote{me@petkus.info, $\mathbb{X}$ @maksympetkus}}
\date{}

\begin{document}

\maketitle

\begin{abstract}
We present \emph{multitable} and its Rust reference implementation: a
stable hash table both materially faster at equal physical memory and more flexible
than the SwissTable in its Rust's hashbrown implementation.
As an
arithmetic mean over 84 configurations it delivers $\mathbf{2.1\times}$
hashbrown's throughput when both hash the same raw bytes and
$\mathbf{1.9\times}$ when hashbrown is keyed on native integers, its best
case; on negative lookups alone, $3.2\times$ and $2.9\times$.

Multitable reaches \textbf{any physical
load factor} up to and \textbf{including one} ($0.9999$ demonstrated), exactly for the requested capacity, compared to hashbrown which doubles at $0.777$ for 4-byte keys and values. At $75\%$ saturation of hashbrown (assumed average case of its rigid ladder) and multitable sized to $0.97$ physical load factor, hashbrown takes $66\%$ more space. 
The lookup probe count has no cliff as the load factor approaches one. 
Bucket size, physical load factor, and failure budget are parameters, and the multitable can be grown without rehashing.

We implement two variants of multitable: plain and filtered. 
At equal physical memory on an Apple M2 Pro the filtered multitable leads hashbrown in all $84$ insert, hit, and miss configurations. 
Multitable is more \textbf{memory-efficient}, at equal mixed-lookup throughput on the map of $4$-byte keys and values the filtered multitable needs up to $12\%$ fewer bytes than hashbrown, and the plain multitable is $18\%$ smaller, holding $\mathbf{22\%}$ more keys in the same memory.
\end{abstract}


\section{Introduction}\label{sec:intro}

A hash table is often the most used data structure in a program
(hashbrown~\cite{hashbrown2026crate}, behind Rust's
standard \texttt{HashMap} and \texttt{HashSet}, is the most downloaded
crate on crates.io at over two billion downloads), and its
footprint is often the largest. Whether it fits in cache or
spills to RAM, fits in RAM or spills to SSD, matters more than any
constant-factor tuning of the probe sequence, because each storage-tier
boundary is a performance cliff, and the table's true size in bytes
decides on which side we land. That true size is not what the
advertised load factor measures. Hashbrown, the Rust port of Google's
SwissTable~\cite{abseil2018swisstables,kulukundis2017cppcon}, has a
maximum load
factor of $7/8=0.875$, but that counts occupied \emph{slots}. Each
slot carries one control byte, so with 4-byte keys and
4-byte values only $8$ of every $9$ allocated bytes can hold payload.
At peak slot load the \emph{physical} load factor, payload bytes over
allocated bytes, is
\begin{equation}\label{eq:physical-load}
\frac{8}{9}\cdot\frac{7}{8}=\frac{7}{9}\approx0.777
\quad\text{(maps)},\qquad
\frac{4}{5}\cdot\frac{7}{8}=0.70
\quad\text{(4-byte sets)}.
\end{equation}
Hence hashbrown's footprint is about $29\%$ (maps) and $43\%$ (sets)
larger than the payload it stores, and that is at \emph{peak}
utilization, before the control array's trailing padding group and
allocator rounding widen the gap. Worse, the bucket count is
a power of two and doubles the moment slot load reaches $7/8$, so
right after a resize the physical load factor is about $0.39$ (maps)
or $0.35$ (sets), for example a 20GB table would need to step up to 40GB. The resize itself rehashes every key, a latency and memory
spike proportional to the table size. And a table serving as a cache
needs timestamps or auxiliary bookkeeping to evict the oldest entries
in bulk.

What this overhead costs depends on the setting:
\begin{itemize}[leftmargin=3em,nosep]
\item \textbf{Storage tiers.} A cache hit costs nanoseconds, a DRAM
access on the order of $100$ nanoseconds, an SSD read tens of
microseconds, so a few percent of footprint that moves the working
set across a boundary is worth orders of magnitude in latency. On the
fast side the trade tilts further: the multitable can run
plain, without per-slot metadata, at physical load factor up to
$1$,
(Figure~\ref{fig:bench-lfsweep}, Remark~\ref{rem:perfect-example}).

\item \textbf{Fixed-memory devices.} Microcontrollers and network
silicon have no next tier: flow, session, and routing tables live in
SRAM sized at design time. The load factor sets how many entries the
die supports, and the sizing must hold with a stated failure
probability, not on average, because rehashing on chip is not an
option.
\item \textbf{Nested tables.} Slack multiplies when tables nest inline,
as in a graph stored as a table of vertices each holding a
table of neighbors, or in pointer-free serialized indexes. With
physical load factor $\varphi$ at every level and the nested tables
dominating the payload, a depth-$k$ nesting stores about $\varphi^k$
payload bytes per allocated byte. Raising $\varphi$ from hashbrown's
$0.777$ to a plain multitable's $0.97$ shrinks the structure by
$1.25\times$ at depth $1$, $1.6\times$ at depth $2$, $1.9\times$ at
depth $3$, and $2.4\times$ at depth $4$, assuming an unlikely 100\% saturation of hashbrown's capacity; assuming $75\%$ saturation multitable is $7.7\times$ smaller at depth 4.
\item \textbf{Garbage-collected heaps.} A chained map puts one small
heap object per entry, and marking work scales with the live-object
count. In runtimes that store entries unboxed, a flat table at
high load holds the same payload in a few large allocations, and the
saved bytes become headroom that lowers collection frequency.
\item \textbf{RAM-resident fleets.} In-memory key-value stores buy DRAM
by the server. When memory is the binding resource, bytes per key set
the machine count, and a $20\%$ smaller table cuts the fleet by up to
$20\%$, in dollars and in watts. Moreover, multitable's agility allows to utilize all of the available memory with no performance penalty.
\end{itemize}

Multitable closes the gap between advertised and physical load. 
Keys hash into fixed-size buckets sized
to the hardware's optimal access size, and the few keys that overflow
a bucket are rehashed into a second, smaller level of identical
structure, whose overflow feeds a third, and so on until the
remainder fits in a single bucket. Because every level is plain
multicollision distribution, binomial tail bounds size each
level at a chosen failure probability. Multitable is \textbf{stable}, i.e., inserted keys don't move and safe to be referenced.

Our contributions are:
\begin{enumerate}[leftmargin=3em,nosep]
\item \textbf{A perfect table} (Section~\ref{sec:perfect}): capacity
exactly $q$ (a multiple of $s$); with probability at least
$1-(L-1)\delta$ all $q$ keys
stored and \emph{every slot occupied} (Theorem~\ref{thm:perfect});
exact expected lookup time
formula (Proposition~\ref{prop:lookup}).
\item \textbf{An efficient table} (Section~\ref{sec:efficient}): load
factor and bucket size are configuration parameters, any target
$\alpha<1$ at any $s$ (Theorem~\ref{thm:efficient}), while
SwissTable/hashbrown is pinned at $7/8$ slot load, $0.777$ physical
for maps \eqref{eq:physical-load}; depth $O(\log q)$
(Proposition~\ref{prop:depth}) with $O(1)$ expected positive-lookup
probes (Corollary~\ref{cor:probes}, Appendix~\ref{app:asymptotics});
depth reduction by exactly
quantified truncation, Chernoff tail bounds, and a near-exact overflow
approximation.
\item \textbf{No cliff near full load} (Section~\ref{sec:quantiles},
Figure~\ref{fig:lookup}): a lookup probes at most one bucket per
level, and the level count grows slowly with the load, whereas under
linear probing the expected cost of a miss grows without bound, as
$1/(1-\alpha)^2$~\cite{knuth1998taocp3}, and SwissTable and
hashbrown stay clear of that region by resizing at slot load $7/8$. The plain
multitable's measured throughput declines gradually up to physical
load $0.977$ (Section~\ref{sec:performance},
Figure~\ref{fig:bench-lfsweep}).
\item \textbf{An improved multicollision bound}
(Appendix~\ref{app:multicollision}), never weaker than the classical
$s$-subset bound and stronger by orders of magnitude once samples
outnumber bins (a factor of $2000$ at $q=10n$, $s=24$).
\item \textbf{Resizing \& cache tables} (Section~\ref{sec:resizing}): incremental
resizing and wrap-around cache eviction.
\item \textbf{Engineering}
(Sections~\ref{sec:implementation}): fixed-size,
SIMD-scanned buckets, byte filters using all 8 bits, and a Rust reference
implementation. At equal physical memory on an Apple M2 Pro the
filtered multitable beats hashbrown in all $84$ insert, hit, and miss
configurations (maps and sets, $4$- and $16$-byte keys, seven
sizes per shape from $3.5\cdot10^3$ to $2.3\cdot10^8$ keys):
$2.1\times$ hashbrown's throughput
with both tables keyed on the same raw bytes and $1.9\times$ with
hashbrown keyed on native integers, its best case, arithmetic means
over the $84$ cells, with negative lookups at $3.2\times$ and
$2.9\times$. At equal mixed-lookup throughput on the $4{+}4$ map the
filtered multitable needs up to $12\%$ fewer bytes and, in cache, the
plain multitable $18\%$ fewer ($22\%$ more keys); against a hashbrown
that has just doubled, a multitable at physical load $0.77$ to
$0.977$ holds $2$ to $2.5\times$ the keys in the same bytes
(Section~\ref{sec:performance}).
\item \textbf{A framed variant} (Section~\ref{sec:variations}): keeping each
lookup within one region of locality.
\item \textbf{Deletion churn strategies} (Section~\ref{sec:churn}): deletes are
supported, and for churn at capacity we give three strategies, suited
to light, heavier and continuous churn: 1) hoisting the deep levels
back into level~$1$ in place, which gives up stability only during
the pass, 2) growing by a level as the tail fills, and 3) sizing for a
raised load factor; the last two are \emph{stable}.
\end{enumerate}

\paragraph{Related work.}
Knuth~\cite{knuth1998taocp3} is the classical reference. Cuckoo hashing
\cite{pagh2004cuckoo} reaches high slot loads with constant worst-case
lookups at the cost of moving keys on insertion. Iceberg hashing
\cite{bender2023iceberg} achieves load $1-o(1)$ while keeping entries
\emph{stable}, a property multitable shares. 
Our
engineering baseline throughout is SwissTable with its Rust port
hashbrown.

The closest relative is funnel hashing, the greedy scheme of
Farach-Colton, Krapivin and Kuszmaul~\cite{farachcolton2025funnel}:
geometrically shrinking arrays of fixed-size buckets, each key
probing one bucket per array and moving on when it is full. The two
designs were developed independently, ours from the earlier
companion work \cite{manifold} that prompted search for efficient storage solution, and are
surprisingly alike, differing in purpose and in
what is left free. Funnel hashing is primarily a theoretical construction whose bucket
size, level count and shrink ratio are fixed by its proof in terms
of the vacancy $1-\alpha$ (their $\delta$), and whose cascade ends
in a special array; the multitable takes bucket size, load factor
and failure budget as parameters matched to hardware and
application, and ends in a single bucket; at load $\alpha=0.9$ the
funnel formulas, with the proof's constants, give $24$ levels plus
the special array where the multitable's $10^5$-key example runs
$5$ levels. Funnel hashing is analyzed for asymptotic probe complexity,
on which it is optimal among greedy schemes, and is presented
without implementation, measurement or treatment of deletions; the multitable is sized for
a load factor at a stated failure probability, judged by physical
load, implemented, and measured ahead of state of the art (hashbrown) at equal memory.
Appendix~\ref{app:funnel} compares the two in detail.

\section{Intuition}\label{sec:intuition}

This work grew out of a companion paper \cite{manifold} on
membership structures built on multicollision hardness, where storage
efficiency was left open (the present paper is self-contained). A
\emph{multicollision} generalizes the birthday collision: an
$s$-collision is $s$ samples landing on one outcome. If we hash $q$
keys uniformly into $n$ buckets, the count in any fixed bucket is a
sum of $q$ independent indicators of probability $1/n$ each, hence
$\Bin(q,1/n)$, and the expected number of buckets holding exactly $k$
keys is
$n\,\PR[\Bin(q,1/n)=k]$ (Proposition~\ref{prop:occupancy} and
linearity of expectation).
Figure~\ref{fig:occupancy} (top left) plots this for a running example,
$n=10$ and $q=100$: occupancy clusters around the mean $\lambda=q/n=10$.

\begin{figure}[t]
\centering
\begin{minipage}[c]{0.48\textwidth}
\centering
\begin{tikzpicture}
\begin{axis}[
    width=\linewidth,
    height=3.4cm,
    ybar,
    bar width=4pt,
    xlabel={$k$ (keys in bucket)},
    ylabel={$10\cdot P(\mathrm{Bin}(100,1/10)=k)$},
    xmin=0, xmax=21,
    xtick={0,5,...,20},
    ymin=0,
    axis lines=left,
    enlarge x limits=0.02,
    tick label style={font=\footnotesize},
    label style={font=\footnotesize},
]
\addplot+[draw=blue!70!black, fill=blue!40, mark=none] coordinates {
(0,0.000266) (1,0.002951) (2,0.016232) (3,0.058916) (4,0.158746) (5,0.338658) (6,0.595787) (7,0.888952) (8,1.148230) (9,1.304163) (10,1.318653) (11,1.198776) (12,0.987880) (13,0.743021) (14,0.513038) (15,0.326824) (16,0.192917) (17,0.105915) (18,0.054265) (19,0.026022) (20,0.011710) (21,0.004957)
};
\end{axis}
\end{tikzpicture}\\[4pt]
\begin{tikzpicture}[y=1cm, x=1cm]
\providecommand{\cascadebucket}{}
\renewcommand{\cascadebucket}[3]{%
  \ifnum#3>0
    \fill[black!25] (#1,#2) rectangle ({#1+0.2*#3},{#2+0.42});
  \fi
  \foreach \k in {1,2,3}
    \draw[black!45, very thin] ({#1+0.2*\k},#2) -- ({#1+0.2*\k},{#2+0.42});
  \draw[black!75, semithick] (#1,#2) rectangle ({#1+0.8},{#2+0.42});
}
\footnotesize
% level 1: 6 buckets
\foreach \x/\occ in {0/4, 0.92/2, 1.84/3, 2.76/4, 3.68/1, 4.6/4}
  \cascadebucket{\x}{0}{\occ}
\node[anchor=east, black!75] at (-0.15,0.21) {level $1$};
\draw[<->, black!60, thin] (0,0.62) -- (0.8,0.62)
  node[midway, above=-1pt, black!75] {$s$ slots};
% level 2: 3 buckets
\foreach \x/\occ in {1.38/3, 2.3/4, 3.22/2}
  \cascadebucket{\x}{-1.05}{\occ}
\node[anchor=east, black!75] at (-0.15,-0.84) {level $2$};
% level 3: 2 buckets
\foreach \x/\occ in {1.84/4, 2.76/1}
  \cascadebucket{\x}{-2.1}{\occ}
\node[anchor=east, black!75] at (-0.15,-1.89) {level $3$};
% one key's journey: in from outside, then full bucket to full bucket
\node[black!75, anchor=south east, inner sep=1pt] at (4.44,0.90) {key};
\draw[->, black!55, semithick] (4.42,0.88) -- (4.95,0.30);
\draw[->, black!55, semithick] (4.70,0.0) -- (2.62,-0.79);
\draw[->, black!55, semithick] (2.60,-1.05) -- (2.26,-1.84);
\draw[->, black!55, semithick] (2.30,-2.12) -- (2.46,-2.62);
\node[black!75] at (2.7,-2.78) {$\vdots$};
% level L: a single bucket
\cascadebucket{2.3}{-3.5}{2}
\node[anchor=east, black!75] at (-0.15,-3.29) {level $L$};
% past the elided levels the key settles in the last bucket
\draw[->, black!55, semithick] (2.50,-2.80) -- (2.82,-3.24);
\end{tikzpicture}
\end{minipage}\hfill
\begin{minipage}[c]{0.48\textwidth}
\centering
\usepgfplotslibrary{groupplots}
\begin{tikzpicture}
\begin{groupplot}[
    group style={group size=1 by 2, vertical sep=0.6cm},
    width=\linewidth,
    height=3.3cm,
    axis lines=left,
    tick label style={font=\footnotesize},
    label style={font=\footnotesize},
]
\nextgroupplot[
    ylabel={min $s$},
    xlabel={$n$},
    ymin=0,
]
\addplot+[mark=*, mark size=1pt, blue] coordinates {
(1,100) (2,69) (3,52) (4,43) (5,37) (6,33) (7,30) (8,28) (9,26) (10,24) (11,23) (12,22) (13,21) (14,20) (15,20) (16,19) (17,18) (18,18) (19,17) (20,17) (21,16) (22,16) (23,16) (24,15) (25,15) (26,15) (27,14) (28,14) (29,14) (30,14)
};
\nextgroupplot[
    ylabel={load factor $q/(ns)$},
    xlabel={$n$},
    ymin=0,
]
\addplot+[mark=*, mark size=1pt, red!70!black] coordinates {
(1,1.000000) (2,0.724638) (3,0.641026) (4,0.581395) (5,0.540541) (6,0.505051) (7,0.476190) (8,0.446429) (9,0.427350) (10,0.416667) (11,0.395257) (12,0.378788) (13,0.366300) (14,0.357143) (15,0.333333) (16,0.328947) (17,0.326797) (18,0.308642) (19,0.309598) (20,0.294118) (21,0.297619) (22,0.284091) (23,0.271739) (24,0.277778) (25,0.266667) (26,0.256410) (27,0.264550) (28,0.255102) (29,0.246305) (30,0.238095)
};
\end{groupplot}
\end{tikzpicture}
\end{minipage}
\caption{Top left: expected number of buckets holding exactly $k$ keys,
$n=10$, $q=100$. Bottom left: the cascade of
Section~\ref{sec:construction}, levels of fixed bucket size $s$ and
shrinking width down to a single bucket at level $L$, filled cells
marking occupied slots; the arrows follow one key from its full
level-$1$ bucket down to a free slot at level $L$. Right: smallest
$s$ with overflow probability below $\delta_0=1.31\cdot10^{-4}$, the
budget of the $s=24$ example, versus $n$, and the resulting load
factor $q/(ns)$.}
\label{fig:occupancy}
\end{figure}
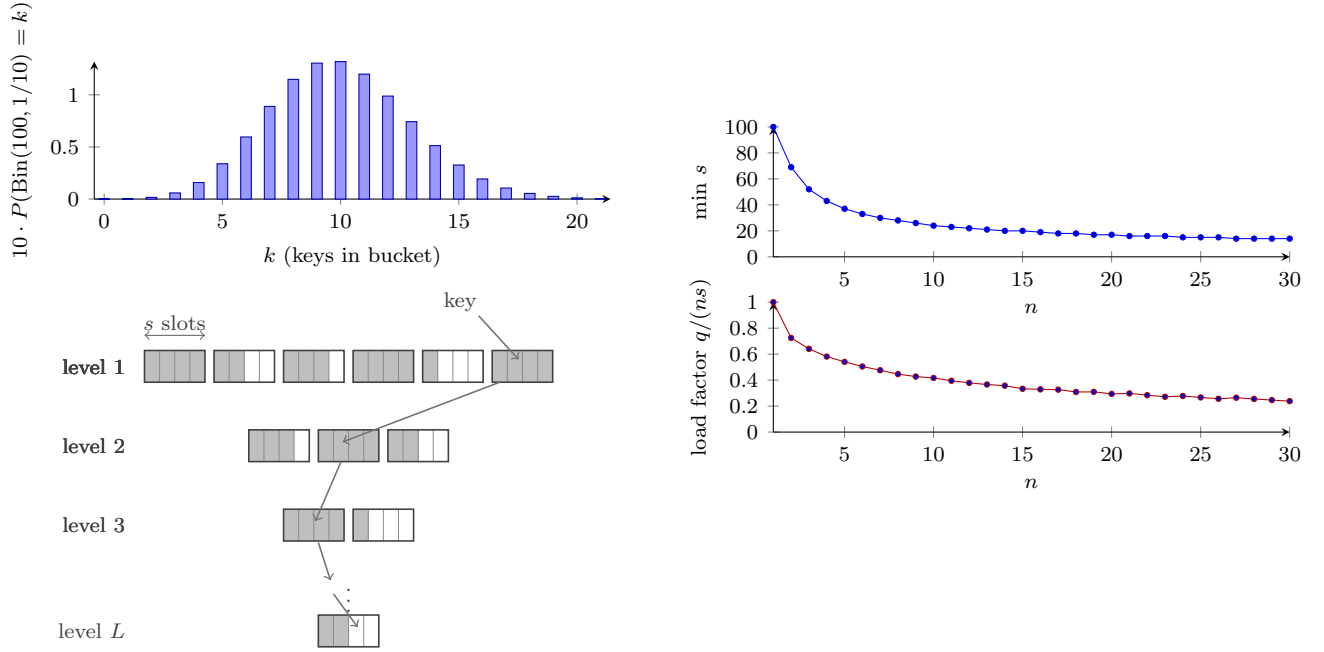

That's the first clue; the simplest table built on it has one
level and no overflow handling. We pick a
bucket size the occupancy is
unlikely to exceed, e.g., $s=24$: the probability that any bucket
overflows is at most
$n\,\PR[\Bin(100,\tfrac1{10})\ge25]\approx1.3\cdot10^{-4}$, about
$0.013\%$. The table is then a flat array of $n\cdot s$ slots: insert
hashes the key to a bucket and writes it into a free slot, lookup
and delete check that bucket, so every operation touches one
bucket: $O(s)$ work, roughly one memory access for hardware-sized $s$.
What does hardware-sized mean?

Every storage medium has a minimum unit
it serves at once: cache lines, DRAM columns, SSD pages, disk
sectors, tape blocks. In DRAM, once a row is activated, reading a few
adjacent columns of it costs about the same as reading one
(a bucket straddling a row boundary, or hitting a closed row, still
pays the activation). So we pick $s$ such that a bucket fills about one
such unit, and a SIMD register compares the whole bucket in about the
time of one comparison. The same logic climbs the hierarchy, from
prefetched cache lines to SSD pages to disk sectors whose positioning latency
dwarfs their reading.

Our trivial construction is fast but wasteful: the load factor is
$q/(ns)=100/240\approx0.417$, because most buckets hold about
$\lambda$ keys while we pay for the outlier everywhere. Relative
variance shrinks as buckets grow: one bucket of size $s = q = 100$ never
overflows and has load factor one, and at two buckets we need
$s=69$ for the same overflow probability as for $s=24$, load factor $0.72$.
Figure~\ref{fig:occupancy} (right) traces this trade across $n$. Why
not big buckets, then? They don't suit the hardware: even where
an open row makes wide reads nearly free, the channel to the CPU and
the SIMD width bound how much we compare per cycle, so smaller reads
stay preferable.

We can do better: shrink the buckets and catch the excess. If we keep
$n=10$ but cut $s$ to $12$, on average only $4.73$ keys fail to fit,
so one spare bucket of the same size absorbs them, and the load
factor becomes $100/(11\cdot12)\approx0.758$, an $82\%$ improvement,
at a cost of a second probe when the home bucket doesn't settle a lookup.
But pushing load or capacity further swells the overflow, and
we're back to one oversized bucket, just relocated. Variable-size
overflow buckets don't help either: we can't know in advance which
home buckets overflow by how much, so we'd need pointers out of
home buckets or entry moves between overflow slabs as
they fill. The first level maps cleanly onto the hardware's atomic
units; the overflow area doesn't.

But why treat overflowing keys differently from the
original ones? They're a smaller instance of the problem we
solved (the spare bucket's few keys in the running example): we
hash them again,
with different seed, into a second
level with the same bucket size and fewer buckets. Rehashing strips
the first-level structure, so the second level is again plain
balls-into-bins; its overflow
cascades to a third level, and so on
until the remainder fits in one final bucket, all levels in one flat
array (Figure~\ref{fig:occupancy}, bottom left). 
Figure~\ref{fig:intro-cascade} traces the cascade on a $10^5$-key
example at load $0.9$, sized as in Section~\ref{sec:quantiles}.

\begin{figure}[t]
\centering
\usetikzlibrary{decorations.pathreplacing}
\begin{tikzpicture}
% ---------- level 1: main panel, plot box 5.6cm x 3.3cm ----------
\begin{axis}[
    at={(0cm,0cm)}, anchor=origin,
    scale only axis, width=5.6cm, height=3.3cm,
    xmin=0, xmax=28, ymin=0, ymax=1,
    axis lines=left,
    xtick={0,8,16,24}, ytick={0,0.5,1},
    yticklabels={$0$,$0.5$,$1$},
    ylabel={fraction of buckets with $>k$ keys},
    tick label style={font=\footnotesize},
    label style={font=\footnotesize},
    clip=false,
]
% empty capacity: background rectangle up to y=1 left of the cutoff;
% the part not painted over by the occupied region below remains visible
\fill[gray!12] (axis cs:0,0) rectangle (axis cs:8,1);
\draw[black!40, very thin] (axis cs:0,1) -- (axis cs:8,1);
% occupied region: area under the staircase left of the cutoff
\addplot[const plot, draw=none, fill=blue!20] coordinates {
(0,0.999999) (1,0.999982) (2,0.999871) (3,0.999365) (4,0.997641) (5,0.992939) (6,0.982252) (7,0.961434) (8,0.961434)
} \closedcycle;
% overflow region: area under the staircase right of the cutoff
\addplot[const plot, draw=none, fill=red!20] coordinates {
(8,0.925949) (9,0.872186) (10,0.798877) (11,0.708002) (12,0.604742) (13,0.496435) (14,0.390950) (15,0.295064) (16,0.213351) (17,0.147814) (18,0.098170) (19,0.062546) (20,0.038260) (21,0.022492) (22,0.012721) (23,0.006928) (24,0.003637) (25,0.001843) (26,0.000902) (27,0.000427) (28,0.000427)
} \closedcycle;
% staircase outline (occupied ++ overflow lists; doubled x=8 = cutoff drop)
\addplot[const plot, draw=black!60, thin] coordinates {
(0,0.999999) (1,0.999982) (2,0.999871) (3,0.999365) (4,0.997641) (5,0.992939) (6,0.982252) (7,0.961434) (8,0.961434)
(8,0.925949) (9,0.872186) (10,0.798877) (11,0.708002) (12,0.604742) (13,0.496435) (14,0.390950) (15,0.295064) (16,0.213351) (17,0.147814) (18,0.098170) (19,0.062546) (20,0.038260) (21,0.022492) (22,0.012721) (23,0.006928) (24,0.003637) (25,0.001843) (26,0.000902) (27,0.000427) (28,0.000427)
};
\end{axis}
% cutoff at x = s = 8 (unit 0.2cm/key: x = 1.6cm), meeting the brace row
\draw[black!75, thin, densely dashed] (1.6,0) -- (1.6,3.42);
% stays / overflows braces over the two x-segments at the cutoff, each
% labeled with its share of the level's input keys
\draw[decorate, decoration={brace, amplitude=3pt}, black!75, thin]
  (0.02,3.42) -- (1.55,3.42)
  node[midway, above=4pt, font=\footnotesize, align=center, black] {stays\\$58.2\%$};
\draw[decorate, decoration={brace, amplitude=3pt}, black!75, thin]
  (1.65,3.42) -- (5.58,3.42)
  node[midway, above=4pt, font=\footnotesize, align=center, black] {overflows\\$41.8\%$};
\node[font=\footnotesize, anchor=south, inner sep=1pt] at (1.6,3.56) {$s$};
% direct region labels: the stays side splits into occupied vs empty
% (the occupancy), both as fractions of capacity
\node[font=\footnotesize, align=center] at (0.8,1.35) {occupied\\$99.2\%$};
\node[font=\footnotesize, anchor=west, inner sep=1pt] at (2.05,3.08) {empty $0.8\%$};
\draw[black!50, very thin] (2.02,3.10) -- (1.45,3.26);
% bottom stack under the panel, mirroring the two-line mini labels:
% level label on the aligned level row, axis label beneath it
\node[font=\footnotesize, anchor=north] at (2.8,-0.42) {level $1$};
\node[font=\footnotesize, anchor=north] at (2.8,-0.78) {keys per bucket $k$};
% ---------- level 2 mini: 41.82% linear scale, box 2.34cm x 1.38cm ----------
\begin{axis}[
    at={(6.75cm,0cm)}, anchor=origin,
    scale only axis, width=2.34cm, height=1.38cm,
    xmin=0, xmax=28, ymin=0, ymax=1,
    axis lines=left, ticks=none,
    axis line style={very thin, -},
]
\fill[gray!12] (axis cs:0,0) rectangle (axis cs:8,1);
\draw[black!40, ultra thin] (axis cs:0,1) -- (axis cs:8,1);
\addplot[const plot, draw=none, fill=blue!20] coordinates {
(0,0.999781) (1,0.997938) (2,0.990172) (3,0.968352) (4,0.922374) (5,0.844874) (6,0.736012) (7,0.604946) (8,0.604946)
} \closedcycle;
\addplot[const plot, draw=none, fill=red!20] coordinates {
(8,0.466876) (9,0.337592) (10,0.228642) (11,0.145177) (12,0.086566) (13,0.048575) (14,0.025708) (15,0.012863) (16,0.006099) (17,0.002746) (18,0.001177) (19,0.000481) (20,0.000188) (21,0.000070) (22,0.000025) (23,0.000009) (24,0.000003) (25,0.000001) (26,0.000000) (27,0.000000) (28,0.000000)
} \closedcycle;
\addplot[const plot, draw=black!60, very thin] coordinates {
(0,0.999781) (1,0.997938) (2,0.990172) (3,0.968352) (4,0.922374) (5,0.844874) (6,0.736012) (7,0.604946) (8,0.604946)
(8,0.466876) (9,0.337592) (10,0.228642) (11,0.145177) (12,0.086566) (13,0.048575) (14,0.025708) (15,0.012863) (16,0.006099) (17,0.002746) (18,0.001177) (19,0.000481) (20,0.000188) (21,0.000070) (22,0.000025) (23,0.000009) (24,0.000003) (25,0.000001) (26,0.000000) (27,0.000000) (28,0.000000)
};
\end{axis}
\draw[black!75, very thin, densely dashed] (7.419,0) -- (7.419,1.38);
\node[font=\footnotesize, align=center, anchor=north] at (7.92,-0.42) {level $2$\\$88.3\%$ full};
% ---------- level 3 mini: 6.76% linear scale ----------
\begin{axis}[
    at={(10.4cm,0cm)}, anchor=origin,
    scale only axis, width=0.758cm, height=0.446cm,  %
    xmin=0, xmax=28, ymin=0, ymax=1,
    axis lines=left, ticks=none,
    axis line style={ultra thin, -},
]
\fill[gray!12] (axis cs:0,0) rectangle (axis cs:8,1);
\addplot[const plot, draw=none, fill=blue!20] coordinates {
(0,0.992000) (1,0.953358) (2,0.860053) (3,0.709876) (4,0.528617) (5,0.353625) (6,0.212861) (7,0.115819) (8,0.115819)
} \closedcycle;
\addplot[const plot, draw=none, fill=red!20] coordinates {
(8,0.057291) (9,0.025918) (10,0.010785) (11,0.004151) (12,0.001484) (13,0.000496) (14,0.000155) (15,0.000046) (16,0.000013) (17,0.000003) (18,0.000001) (19,0.000000) (20,0.000000) (21,0.000000) (22,0.000000) (23,0.000000) (24,0.000000) (25,0.000000) (26,0.000000) (27,0.000000) (28,0.000000)
} \closedcycle;
\addplot[const plot, draw=black!60, ultra thin] coordinates {
(0,0.992000) (1,0.953358) (2,0.860053) (3,0.709876) (4,0.528617) (5,0.353625) (6,0.212861) (7,0.115819) (8,0.115819)
(8,0.057291) (9,0.025918) (10,0.010785) (11,0.004151) (12,0.001484) (13,0.000496) (14,0.000155) (15,0.000046) (16,0.000013) (17,0.000003) (18,0.000001) (19,0.000000) (20,0.000000) (21,0.000000) (22,0.000000) (23,0.000000) (24,0.000000) (25,0.000000) (26,0.000000) (27,0.000000) (28,0.000000)
};
\end{axis}
% solid, not dashed: at this panel size a dash pattern reads as noise
\draw[black!75, ultra thin] (10.618,0) -- (10.618,0.446);
\node[font=\footnotesize, align=center, anchor=north] at (10.59,-0.42) {level $3$\\$59.1\%$ full};
% ---------- level 4 glyph (enlarged; true scale 0.14%): 27.1% full ----------
\fill[gray!12] (11.9,0) rectangle (12.7,0.3);
\fill[blue!20] (11.9,0) rectangle (12.117,0.3);
\draw[black!60, thin] (11.9,0) rectangle (12.7,0.3);
\node[font=\footnotesize, align=center, anchor=north] at (12.3,-0.42) {level $4$\\$27.1\%$ full};
% ---------- flow arrows: expected keys between levels ----------
\draw[->, black!55, semithick] (5.72,0.69) -- (6.65,0.69)
  node[midway, above=1pt, font=\footnotesize, black] {$41816$};
\draw[->, black!55, semithick] (9.2,0.11) -- (10.3,0.11)
  node[midway, above=1pt, font=\footnotesize, black] {$6762$};
\draw[->, black!55, semithick] (10.9,0.15) -- (11.8,0.15)
  node[midway, above=1pt, font=\footnotesize, black] {$141$};
\end{tikzpicture}
\captionsetup{font=footnotesize}
\caption{The cascade of the Section~\ref{sec:quantiles} example,
$q=10^5$ keys, $s=8$, target load $0.9$, overfill $r=1.1$,
$\delta=2^{-10}$, $5$ levels with $(n_i)=(7334,4962,1401,65,1)$
buckets, drawn through level $4$. Each panel plots the fraction of
buckets holding more than $k$ keys in unit-width steps, so each
shaded area is an expected number of keys per bucket: left of the
cutoff at $s$ the keys stay in the level, blue being occupied
capacity and gray empty; beyond it, in red, they overflow into the
next level.
Arrow labels are expected key counts; stays and overflows are
fractions of the level's input, and the occupied, empty and fill
figures are fractions of capacity, all in expectation. Panels $2$
and $3$ are drawn at $41.8\%$ and $6.8\%$ of the first panel's
linear size, their shares of the keys; level $4$, with $0.14\%$, is
an enlarged plain-fill glyph, and level $5$ is a single bucket sized
to receive the remaining keys, at most $8$ with probability
$1-\delta$.}
\label{fig:intro-cascade}
\end{figure}
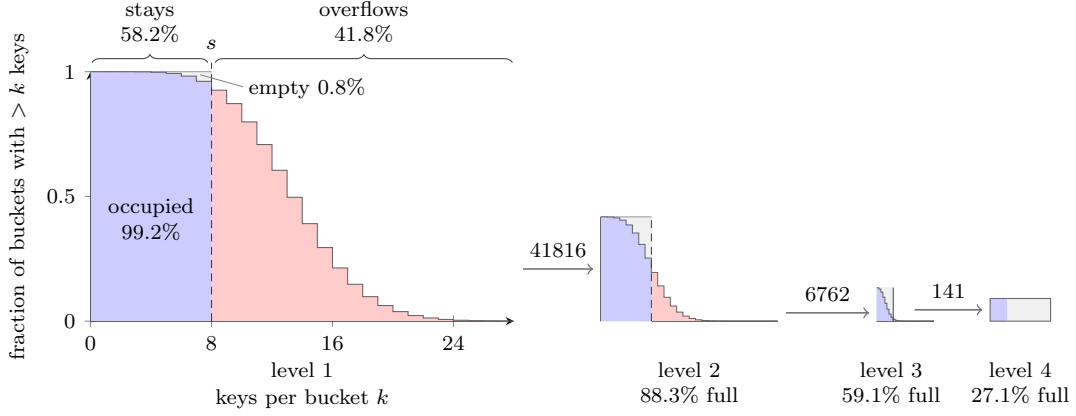

A last knob reaches load factor one: a level wastes space only through \emph{underfull}
buckets, and concentration controls that slack as well. At $n=10$, $q=100$
over $99\%$ of buckets are expected to hold at least $4$ keys, and
choosing $n$ small relative to the incoming keys makes even one
underfull bucket unlikely, so every slot fills. The cost is more
overflow, hence a deeper cascade. This gives us, in expectation:
1) structurally identical levels, 2) a load factor that is a
parameter, adjustable up to one, 3) positive lookups concentrated at
the top (levels shrink geometrically), 4) negative lookups paying one
bucket per level of a logarithmic cascade. Expectations aren't
guarantees, though: next we replace them with explicit tail bounds,
so the table survives the worst case at a probability we choose.

\section{Construction and analysis}\label{sec:construction}

Section~\ref{sec:model} fixes the probabilistic model,
Section~\ref{sec:perfect} builds the table with load factor exactly
one, Section~\ref{sec:efficient} the table for any target load below
one, and Sections~\ref{sec:depth} and~\ref{sec:quantiles} shorten the
resulting cascade.

\subsection{Model}\label{sec:model}

Throughout the
analysis we work in the random-oracle model: the hash values of distinct keys
are independent, and the values assigned to the same key at
different levels are independent as well. For time analysis we assume constant lookup time in a bucket of size $s$ (due to SIMD). Proofs not given in the body are
collected in Appendix~\ref{app:proofs}. Throughout, $\ind{A}$ is the
indicator of an event or predicate $A$, equal to $1$ when it holds
and $0$ otherwise, and $\pos{x}:=\max(x,0)$. Formally:

\begin{model}[Multilevel table]
A table for $q$ keys consists of levels $1,\dots,L$; level $i$ has
$n_i$ buckets of $s$ slots each. Each key is assigned a bucket at each
level, and all these assignments (across keys and across levels) are
independent and uniform. Keys are inserted level by level: if $X_b$
keys map to bucket $b$ of level $i$, the bucket stores
$\min(X_b,s)$ of them and the remaining $\pos{X_b-s}$ cascade to
level $i+1$. We write
\begin{equation}\label{eq:overflow-def}
O \;=\; \sum_{b=1}^{n}\pos{X_b-s}
\end{equation}
for the overflow of a level with $n$ buckets, and we call a level
\emph{saturated} if every one of its buckets receives at least $s$
keys. Keys that cascade past the last level are lost; we call this
event \emph{terminal overflow}.
\end{model}

Beyond that measurability requirement, which keys are cascaded is
irrelevant: the next level hashes them with independent randomness.
Everything below reduces to two elementary facts,
Proposition~\ref{prop:occupancy} and Lemma~\ref{lem:counts}.

The companion manuscript
\cite{manifold} counts the buckets receiving $s$ or more
keys: when $q$ keys are hashed independently and uniformly into $n$
buckets and $X_b$ denotes the number of keys in bucket $b$,
\begin{equation}\label{eq:companion-count}
\E\Bigl[\sum_b \ind{X_b\ge s}\Bigr]
\;=\; n-\sum_{j=0}^{s-1} n\Bigl(\tfrac1n\Bigr)^{j}\binom{q}{j}
\Bigl(1 - \tfrac{1}{n}\Bigr)^{q-j}.
\end{equation}
Dividing by $n$ gives the probability of any given bucket holding
$s$ or more keys: by symmetry every bucket receives $s$ or more keys
with the same probability,
$1-\sum_{j=0}^{s-1}\binom{q}{j}(\tfrac1n)^{j}(\tfrac{n-1}{n})^{q-j}$,
and that sum is the distribution function of a binomial law. That
form is efficiently computable, CDF and survival function alike, and is what
the tail machinery of this section controls.

\begin{proposition}[Bucket occupancy]\label{prop:occupancy}
Hash $q$ keys independently and uniformly into $n$ buckets, as when
the keys reaching a level with $n$ buckets in the model of
Section~\ref{sec:model} are placed,
and let
$X_b$ be the number of keys in bucket $b$. Then
$(X_1,\dots,X_n)$ is multinomial with $q$ trials and uniform cell
probabilities, and marginally
\begin{equation}\label{eq:binomial-occupancy}
X_b \sim \Bin\!\left(q,\tfrac1n\right)
\qquad\text{for every } b .
\end{equation}
\end{proposition}

\begin{proof}
$X_b=\sum_{j=1}^{q}\ind{\text{key } j \text{ maps to } b}$ is a sum of
$q$ independent Bernoulli$(1/n)$ indicators; the joint law is
multinomial by definition of independent uniform placement.
\end{proof}

For the binomial occupancy law we write the distribution function
\begin{equation}\label{eq:cdf-def}
F\!\left(k;q,\tfrac1n\right) \;:=\;
\PR\!\left[\Bin\!\left(q,\tfrac1n\right)\le k\right]
\;=\; \sum_{j=0}^{k}\binom{q}{j}n^{-j}\bigl(1-\tfrac1n\bigr)^{q-j},
\end{equation}
so that $\PR[\Bin(q,\tfrac1n)\ge s]=1-F(s-1;q,\tfrac1n)$.

In this notation the subtracted terms in \eqref{eq:companion-count}
are $n\,\PR[\Bin(q,\tfrac1n)=j]$
(Proposition~\ref{prop:occupancy}), and the per-bucket probability
above is exactly the binomial tail $1-F(s-1;q,\tfrac1n)$, the form
we use from here on (later levels substitute their own key count
for $q$).

\begin{lemma}[Collision counts]\label{lem:counts}
In the setting of Proposition~\ref{prop:occupancy}, for every
$s\ge 1$,
\begin{align}
\E\Bigl[\sum_b \ind{X_b\ge s}\Bigr]
  &= n\,\PR\!\left[\Bin\!\left(q,\tfrac1n\right)\ge s\right]
   = n\bigl(1-F(s-1;q,\tfrac1n)\bigr), \label{eq:count-ge}\\
\E\Bigl[\sum_b \ind{X_b\le s-1}\Bigr]
  &= n\,F\!\left(s-1;q,\tfrac1n\right). \label{eq:count-le}
\end{align}
\end{lemma}

\begin{proof}
Apply linearity of expectation together with
\eqref{eq:binomial-occupancy}: the $n$ marginals are identical (the
$X_b$ are dependent, but linearity requires no independence). The
second identity is the complementary count.
\end{proof}

Beyond sizing levels, \eqref{eq:count-ge} improves on the classical
$s$-subset multicollision bound, a side result developed in
Appendix~\ref{app:multicollision}.

\begin{remark}\label{rem:survival-fix}
In survival-function notation, $\mathrm{SF}(k):=\PR[X>k]$ for
$X\sim\Bin(q,\tfrac1n)$, the count \eqref{eq:count-ge} reads
$n\cdot\mathrm{SF}(s-1)$, since $X\ge s$ is the event $X>s-1$.
\end{remark}

\subsection{A perfect table}\label{sec:perfect}

The last knob of the previous section is the simplest case, and we
make it precise first: load factor exactly $1$, every slot occupied
and no terminal overflow, with probability as close to $1$ as
desired. The
construction makes each level as large as possible subject to the
constraint that the expected number of underfull buckets is at most
$\delta$, while always leaving at least a bucket's worth of keys for
the next level.

\begin{definition}[Perfect construction]\label{def:perfect}
Fix a bucket size $s\ge1$, a per-level budget $\delta\in(0,1)$, and a
number of keys $q$ that is a multiple of $s$, $q\ge s$. The quantity
$o_i$ will count the keys reaching level $i$, which levels
$i,\dots,L$ are allocated to hold. Set $o_1:=q$ and repeat for
$i=1,2,\dots$: if $o_i=s$, allocate
a final level with
$n_i:=1$ bucket and stop, setting $L:=i$; otherwise ($o_i\ge 2s$)
allocate
\begin{equation}\label{eq:perfect-ni}
n_i \;:=\; \max\Bigl\{\,1\le n\le \tfrac{o_i}{s}-1 \;:\;
   n\,F\!\left(s-1;\,o_i,\tfrac1n\right)\le\delta \,\Bigr\},
\qquad
o_{i+1} \;:=\; o_i - s\,n_i .
\end{equation}
\end{definition}

The maximum in \eqref{eq:perfect-ni} is over a nonempty finite set:
$o_i\ge2s$ in this branch, so $n=1$ is in range, and it satisfies
the constraint because a single bucket receives all $o_i$ keys
deterministically, giving $F(s-1;o_i,1)=0$. The cap
$n\le o_i/s-1$ keeps $o_{i+1}=o_i-sn_i$ between $s$ and $o_i-s$, so
every $o_i$ is a positive multiple of $s$, the construction
terminates after at most $q/s$ levels, and it always ends in the
single-bucket final level, which receives at least $s$ keys and
cannot be underfull. Candidates $n>o_i/s$ fail the constraint on
their own, some bucket being underfull by pigeonhole, so the
expected count \eqref{eq:count-le} is at least $1>\delta$; the cap
hence removes only the exact fit $n=o_i/s$, which would end the
cascade with $o_{i+1}=0$ and no final level. By telescoping, the
total capacity is
\begin{equation}\label{eq:perfect-capacity}
s\sum_{i=1}^{L} n_i \;=\; \sum_{i=1}^{L-1}(o_i-o_{i+1}) + o_L \;=\; q :
\end{equation}
not a single slot more than there are keys.

\begin{theorem}[Perfect table]\label{thm:perfect}
Run Definition~\ref{def:perfect} and insert the $q$ keys under
the model of Section~\ref{sec:model}. With probability at least
$(1-\delta)^{L-1}\ge1-(L-1)\delta$ every level is saturated; on this
event every key is stored, every slot of the table is occupied, and
the slot load factor is exactly~$1$.
\end{theorem}

\begin{proof}
Let $A_i$ denote the event that level $i$ is saturated, i.e., every
one of its $n_i$ buckets receives at least $s$ keys. We claim that on
$A_1\cap\dots\cap A_{i-1}$ the number of keys reaching level $i$ is
exactly $o_i$. This holds by induction: level $1$ receives $o_1=q$
keys, and if level $j$ receives $o_j$ keys and is saturated, it stores
exactly $s$ keys per bucket, hence exactly $s\,n_j$ in total, and
passes on $o_j-s\,n_j=o_{j+1}$.

Fix $i\le L-1$ and condition on any realization of the randomness of
levels $1,\dots,i-1$ lying in $A_1\cap\dots\cap A_{i-1}$. The $o_i$
keys reaching level $i$ are placed by fresh uniform randomness,
independent of the conditioning, so the loads of level $i$ follow
Proposition~\ref{prop:occupancy} with $q=o_i$, $n=n_i$. Let
$U_i:=\sum_{b\le n_i}\ind{X_b\le s-1}$ be the number of underfull buckets.
By Markov's inequality, Lemma~\ref{lem:counts}, and the defining
constraint in \eqref{eq:perfect-ni},
\begin{equation}\label{eq:perfect-markov}
\PR\bigl[\lnot A_i \,\big|\, A_1\cap\dots\cap A_{i-1}\bigr]
= \PR\bigl[U_i\ge 1 \,\big|\, \cdots\bigr]
\;\le\; \E\bigl[U_i \,\big|\, \cdots\bigr]
= n_i\,F\!\left(s-1;o_i,\tfrac1{n_i}\right)
\;\le\;\delta .
\end{equation}
For the final level, on $A_1\cap\dots\cap A_{L-1}$ exactly $o_L=s$
keys reach the single bucket of level $L$, which is hence
saturated with conditional probability $1$. Multiplying the
conditional bounds along the chain,
\begin{equation}\label{eq:perfect-union}
\PR\Bigl[\bigcap_{i=1}^{L}A_i\Bigr]
= \prod_{i=1}^{L}\PR\bigl[A_i\,\big|\,A_1\cap\dots\cap A_{i-1}\bigr]
\;\ge\;(1-\delta)^{L-1}\;\ge\;1-(L-1)\,\delta ,
\end{equation}
the last step by Bernoulli's inequality.
On $\bigcap_i A_i$ every bucket of every level holds exactly $s$ keys,
so all $q=s\sum_i n_i$ keys are stored by \eqref{eq:perfect-capacity},
no key cascades past level $L$, and all slots are full.
\end{proof}

Randomness decides only whether the build succeeds, so when one
fails the obvious repair is to redraw everything and rebuild, at an
expected cost of $1/\PR[\text{success}]$ full builds, at most
$(1-\delta)^{-(L-1)}$ by Theorem~\ref{thm:perfect}.
Alternatively, $\delta$ can be pushed cryptographically low
($2^{-40}$, $2^{-64}$): capacity stays exactly $q$ by
\eqref{eq:perfect-capacity}, only the cascade deepens, and terminal
overflow becomes an event we plan never to observe, as is usual in
probabilistic cryptography; since slot load exactly $1$ holds on the
no-terminal-overflow event, load factor $1$ is then for every
practical purpose unconditional. 

\begin{proposition}[Lookup cost]\label{prop:lookup}
Condition on all levels being saturated. A lookup probes buckets at
levels $1,2,\dots$ (one bucket per level) until the key is found. For
a key chosen uniformly among the $q$ stored keys, the expected number
of probed buckets is
\begin{equation}\label{eq:lookup}
\sum_{i=1}^{L} i\,\frac{s\,n_i}{q},
\end{equation}
while a lookup of a key not in the table probes all $L$ levels.
\end{proposition}

\begin{proof}
On the saturation event level $i$ stores exactly $s\,n_i$ keys, so a
uniformly random stored key resides at level $i$ with probability
$s\,n_i/q$, and finding it costs exactly $i$ probes; take
expectations. Since every bucket is full, a miss at a level carries no
information, and a negative lookup can stop only after all $L$ levels.
\end{proof}

\begin{remark}[Worked example]\label{rem:perfect-example}
For $q=7680$, $s=16$, $\delta=2^{-10}$ the construction
\eqref{eq:perfect-ni} yields $L=12$ levels,
$(n_i)=(191,117,71,43,25,14,8,5,3,1,1,1)$, terminal-overflow
probability at most $11\cdot2^{-10}\approx0.011$, and expected
positive-lookup cost \eqref{eq:lookup} $\approx2.45$ buckets. In
simulation ($2\times10^5$ placements per level) the overflow
frequency is
$0.0045$, about $2.4\times$ inside the bound. Larger buckets compress
the cascade:
$s=192$ at the same $q$ and $\delta$ yields $L=4$,
$(n_i)=(30,7,2,1)$, lookup cost $1.35$, and overflow frequency
$0.00033$ against the bound $0.0029$. Where a shallow table of
exactly $q$ slots is preferred and a failed build costs little to
redraw (e.g., static table), we can push $\delta$ up: ${\delta=0.9}$ at $s=16$
gives only $L=7$ and expected lookup time $1.74$. The bound of
Theorem~\ref{thm:perfect} is then $(1-\delta)^{L-1}=10^{-6}$, four
orders of magnitude below what simulation shows: roughly $1$
one-shot build in $66$ succeeds.
The construction is realizable as it stands; the crate of
Section~\ref{sec:implementation} ships it as
\texttt{examples/perfect.rs}: the plain table with $32$-byte keys and
$128$-byte values at $q=512000$, $s=64$, $\delta=2^{-5}$, with the
bucket counts $(n_i)=(4872, 1947, 752, 280, 100, 34, 11, 1, 1, 1, 1)$, $L=11$, supplied
literally. All $512000$ insertions succeed, so the slot load is
exactly the $1$ of Theorem~\ref{thm:perfect}, and \eqref{eq:lookup}
puts the expected positive-lookup cost at $1.62$ bucket probes
against $11$ for a miss. The measured physical load factor is $0.9999001$; the shortfall is one length byte per bucket plus the
per-level counts, offsets and metadata. Down to the sixth level the counts are
those of Definition~\ref{def:perfect}; the tail $(1, 1,1,1)$ means
only
the seven upper levels carry failure probability.
\end{remark}

\subsection{An efficient table}\label{sec:efficient}

Load factor one costs depth, a dozen levels in the worked example;
relaxing the target buys depth back: we now make the load factor a
parameter and build a table meeting any target $\alpha<1$.
Such a table no longer needs saturated levels, only levels holding
an $\alpha$ fraction of their capacity with high probability, which
is the same as bounding overflow from above.

We proceed in two steps: first the exact mean overflow, the average
case; then a bound on how large the overflow can be at a failure
probability $\delta$ of our choosing. The second step yields $4$
bounds in place of one, because no single argument stays tight
across the whole range of loads: they rest on different mechanisms,
from Markov's inequality through bounded differences to exponential
tilting, and Figure~\ref{fig:bounds} maps which is smallest at each
load; since each holds on its own, the construction uses their
pointwise minimum, again a valid bound.

\paragraph{Overflow of one level.}
The exact mean of the overflow \eqref{eq:overflow-def} follows from a
truncated-mean identity for the occupancy law
\eqref{eq:binomial-occupancy}.

\begin{lemma}[Truncated mean; \appproof{pf:truncated}]\label{lem:truncated}
Let $X\sim\Bin(q,\tfrac1n)$ be the occupancy
\eqref{eq:binomial-occupancy} of a bucket with $s$ slots. The
expected number of keys in the bucket, counted on the event that the
bucket overflows, is
\begin{equation}\label{eq:truncated}
\E\bigl[X\,\ind{X\ge s+1}\bigr]
= \frac{q}{n}\,\PR\!\left[\Bin\!\left(q-1,\tfrac1n\right)\ge s\right]:
\end{equation}
each key lands in the bucket with probability $1/n$, and the bucket
then overflows precisely when at least $s$ of the other $q-1$ keys
join it.
\end{lemma}

\begin{theorem}[Exact mean overflow; \appproof{pf:meanoverflow}]\label{thm:meanoverflow}
Hash $q$ keys uniformly into $n$ buckets of $s\ge1$ slots and let $O$
be the overflow \eqref{eq:overflow-def}. Then
\begin{equation}\label{eq:mean-overflow}
\E[O]
= q\,\PR\!\left[\Bin\!\left(q-1,\tfrac1n\right)\ge s\right]
 \;-\; s\,n\,\PR\!\left[\Bin\!\left(q,\tfrac1n\right)\ge s+1\right].
\end{equation}
\end{theorem}

Both terms count only overflowing buckets, those holding $s+1$ or
more keys. The first is the expected number of keys landing in them,
since a key's bucket overflows precisely when at least $s$ of the
other $q-1$ keys share it, the event of the first factor; the second
is the expected number of slots there, $s$ for each of the expected
$n\,\PR[\Bin(q,\tfrac1n)\ge s+1]$ overflowing buckets by
Lemma~\ref{lem:counts}. The overflow is the difference; buckets
holding exactly $s$ keys enter neither term.

Two deterministic facts recur.

\begin{lemma}[Cap and monotonicity; \appproof{pf:capmono}]\label{lem:capmono}
In the setting of Theorem~\ref{thm:meanoverflow}:
\begin{enumerate}[label=(\roman*),nosep]
\item $O\le\pos{q-s}$ always, with equality for $n=1$;
\item if $O^{(q,n)}$ denotes the overflow with $q$ keys in the $n$
buckets (same $s$), then $O^{(q,n)}$ is stochastically nondecreasing
in $q$: for $q'\le q$ and every $t$,
$\PR[O^{(q',n)}>t]\le\PR[O^{(q,n)}>t]$.
\end{enumerate}
\end{lemma}

We now bound the upper tail of $O$ two ways, with complementary
regimes of strength. These two arguments yield $3$ of the $4$
bounds, the tilting one giving a pair. Markov's inequality
applied to $O$ supplies the fourth.

\begin{lemma}[Bounded differences]\label{lem:mcdiarmid}
In the setting of Theorem~\ref{thm:meanoverflow}, for every $t\ge0$,
\begin{equation}\label{eq:mcdiarmid}
\PR\bigl[O\ge\E[O]+t\bigr]\le \exp\!\left(-\tfrac{2t^2}{q}\right).
\end{equation}
\end{lemma}

\begin{proof}
This is McDiarmid's inequality~\cite{mcdiarmid1989bounded}: $O$ is a
function of the $q$ independent bucket choices, and rerouting one
key moves one ball between two buckets, changing $O$ by at most $1$.
\end{proof}

\begin{lemma}[Negative association and stored keys; \appproof{pf:na}]\label{lem:na}
In the setting of Theorem~\ref{thm:meanoverflow}, the occupancy
vector $(X_1,\dots,X_n)$ is negatively associated, and monotone
functions acting on disjoint sets of coordinates preserve that
property. Moreover, pathwise
\begin{equation}\label{eq:stored-keys}
O \;=\; q-\sum_{b=1}^{n}\min(X_b,s).
\end{equation}
\end{lemma}

The remaining pair tilts this negative association exponentially.
Both are Chernoff bounds in quantile form, fixing the budget $\delta$
and solving for the threshold that meets it rather than reporting a
probability for a fixed deviation; negative association bounds the
exponential moment by the product that independence would make an
identity.

\begin{lemma}[Tilted quantiles; \appproof{pf:tilt}]\label{lem:tilt}
In the setting of Theorem~\ref{thm:meanoverflow}, let
$W:=\sum_b\pos{s-X_b}$ be the unused slots, so that
\eqref{eq:stored-keys} reads $O=q-sn+W$. For
$X\sim\Bin(q,\tfrac1n)$ and $\delta\in(0,1)$ define
\begin{equation}\label{eq:chernoff-cert}
o^*(n):=\inf_{\theta>0}\,
\frac{n\ln\E\bigl[e^{\theta\pos{X-s}}\bigr]+\ln(1/\delta)}{\theta},
\qquad
w^*(n):=\inf_{\theta>0}\,
\frac{n\ln\E\bigl[e^{\theta\pos{s-X}}\bigr]+\ln(1/\delta)}{\theta}.
\end{equation}
Then $\PR\bigl[O>o^*(n)\bigr]\le\delta$ and
$\PR\bigl[W>w^*(n)\bigr]\le\delta$.
\end{lemma}

Combining Lemmas \ref{lem:mcdiarmid} and~\ref{lem:tilt} with
Markov's inequality, we define the per-level overflow bound used by
the construction. Fix $\delta\in(0,1)$ and set, with $\E[O]$ from
\eqref{eq:mean-overflow} and $o^*,w^*$ from
\eqref{eq:chernoff-cert},
\begin{align}
B_1(q,n,\delta) &:= \E[O]+\sqrt{\tfrac{q}{2}\ln\tfrac1\delta},
   \label{eq:B1}\\
B_2(q,n,\delta) &:= \E[O]/\delta \quad\text{(Markov)},
   \label{eq:B2}\\
B_3(q,n,\delta) &:= o^*(n),
\qquad
B_4(q,n,\delta) := q-sn+w^*(n),
   \label{eq:B34}\\
B(q,n,\delta) &:= \min\bigl\{B_1,\,B_2,\,B_3,\,B_4,\,\pos{q-s}\bigr\}.
   \label{eq:Bdef}
\end{align}

\begin{proposition}[Per-level bound; \appproof{pf:levelbound}]\label{prop:levelbound}
$\PR\bigl[O> B(q,n,\delta)\bigr]\le\delta$, and
$B(q,1,\delta)=\pos{q-s}$, which the overflow of a single bucket
attains deterministically.
\end{proposition}

\begin{remark}[Which bound wins]\label{rem:b1-vs-b2}
With $\lambda=q/n$ the mean load, the tilts cover most of the range:
$B_4$ wins on overcommitted levels, where $\E[W]$ is small, and
$B_3$ on undercommitted residues. The closed form $B_1$ is minimal
only on a narrow band near $\lambda=s$, and below
$\lambda\approx1.3$, where overflow is already rare, the mean bound
$B_2$ is the minimum. 
The deviation terms of
$B_1$ and $B_4$ are of order $\sqrt{q}$: the term of $B_1$ exceeds
$s$ exactly when $q>2s^{2}/\ln(1/\delta)$, and the factored
exponential moment behind $w^*$ keeps a deviation of the same order
at every $n\ge2$. Past that threshold neither bound can guarantee
$O\le s$, apart from the single-bucket case, where the tilted bounds
$B_3$ and $B_4$ degenerate to the deterministic cap $\pos{q-s}$. At $n\ge2$, only
$B_2$ (by forcing a tiny mean) and $B_3$ (by bounding the tail
directly) reach values that small, which is what the truncation of
Section~\ref{sec:depth} needs. Figure~\ref{fig:bounds} maps the
regimes at the reduced budget $\delta=2^{-20}$ against the largest
overflow in $2^{20}$ simulated runs, an ensemble sized so that each
bound is exceeded at most once in expectation.
\end{remark}

\begin{figure}[t]
\centering
\input{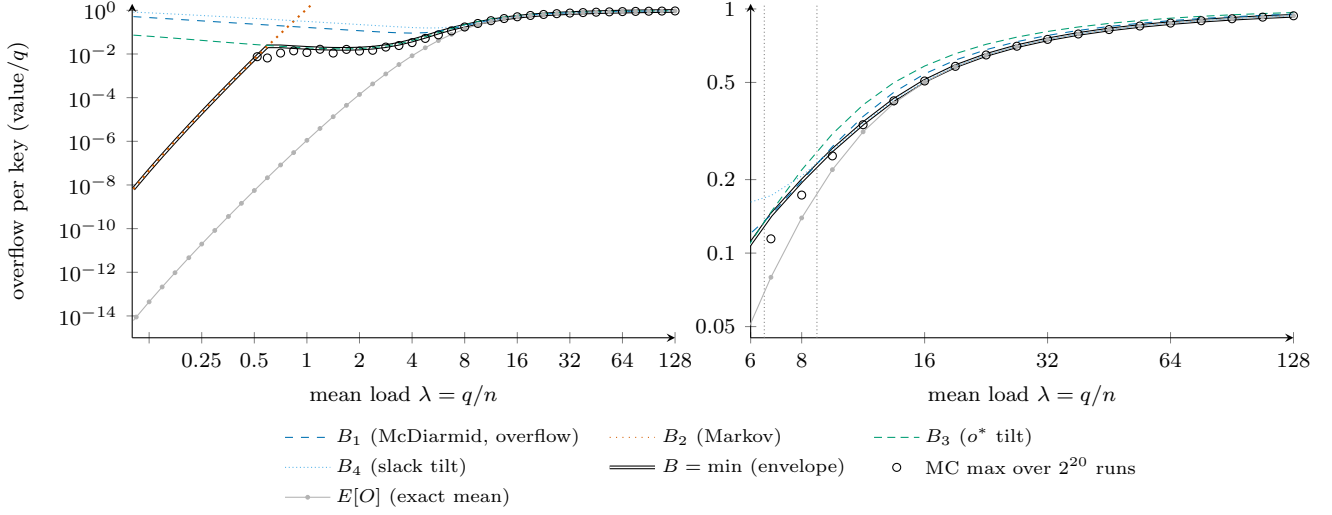}
\caption{The per-level bounds $B_1$--$B_4$ of \eqref{eq:Bdef}, their
minimum $B$, the exact mean $\E[O]$, and the largest overflow in
$2^{20}$ simulated runs (circles), all normalized per key, versus
the mean load $\lambda=q/n$ at $s=8$, $n=256$, $\delta=2^{-20}$, and
no observed maximum crosses a bound. The left panel spans the full
load range; the right repeats the data over $6\le\lambda\le128$
with the vertical axis fitted to that window, where $B_2$ runs off
scale, the panel's dotted verticals marking the two handoffs inside
the window. The circles
start at $\lambda\approx0.52$, the exact onset of
$\PR[O\ge1]>\delta$, below which the quantile is $0$ and no run
overflows. $B_2$ is the minimum for $\lambda\lesssim0.6$, $B_3$ on
$0.6\lesssim\lambda\lesssim6.5$, $B_1$ on
$6.5\lesssim\lambda\lesssim8.7$, and $B_4$ beyond; the deterministic
cap is nowhere active.}
\label{fig:bounds}
\end{figure}

\paragraph{The construction and its guarantee.}
The constraint now runs the other way: a level's input less the
overflow bound must fill the target fraction of its slots.

\begin{definition}[Efficient construction]\label{def:efficient}
Fix $s\ge1$, a target load factor $\alpha\in(0,1)$, a per-level budget
$\delta\in(0,1)$, and $q>s$ keys. Set $o_1:=q$ and repeat for
$i=1,2,\dots$: if $o_i\le s$, allocate a final level with $n_i:=1$
bucket and stop, setting $L:=i$; otherwise allocate
\begin{equation}\label{eq:efficient-ni}
n_i := \max\Bigl\{n\ge1 :\;
  o_i - B(o_i,n,\delta)\;\ge\;\alpha\,s\,n\Bigr\},
\qquad
o_{i+1} := \bigl\lfloor B(o_i,n_i,\delta)\bigr\rfloor .
\end{equation}
\end{definition}

The maximum in \eqref{eq:efficient-ni} is well defined: $n=1$
qualifies, since $B(o_i,1,\delta)=o_i-s$ by
Proposition~\ref{prop:levelbound} and $o_i-(o_i-s)=s\ge\alpha s$, and
$B\ge0$ forces $n\le o_i/(\alpha s)$. Moreover
$o_{i+1}\le B(o_i,n_i,\delta)\le o_i-s$ by the cap in \eqref{eq:Bdef},
so the construction terminates after at most $\lceil q/s\rceil$ levels;
flooring is
harmless because the overflow itself is an integer.

\begin{theorem}[Efficient table; \appproof{pf:efficient}]\label{thm:efficient}
Run Definition~\ref{def:efficient} and insert the $q$ keys under
the model of Section~\ref{sec:model}. Let $C:=s\sum_{i=1}^{L}n_i$, with $n_L=1$,
be the total capacity. Then:
\begin{enumerate}[label=(\alph*),nosep]
\item deterministically,
\begin{equation}\label{eq:capacity-bound}
C \;\le\; \frac{q-o_L}{\alpha}+s ;
\end{equation}
\item with probability at least $(1-\delta)^{L-1}\ge1-(L-1)\delta$,
the overflow of every level $i<L$ is at most $o_{i+1}$ and no
terminal overflow occurs, so all $q$ keys are stored;
\item on the event of (b) the slot load factor $q/C$ satisfies
\begin{equation}\label{eq:load-bound}
\frac{q}{C}\;\ge\;\frac{q\,\alpha}{q-o_L+\alpha s}
\;\ge\;\frac{\alpha}{1+\alpha s/q}\;\ge\;\alpha-\frac{s}{q}.
\end{equation}
\end{enumerate}
\end{theorem}

Every quantity in the sizing is worst case: level $i$ is sized
against the quantile bound $B$ at its planned input $o_i$, and on
the event of (b) the realized input never exceeds $o_i$. The
capacity $C$ is fixed at construction time, so the load factor bound
of (c) holds on every non-failing run, not merely on average. The
table is also \emph{stable} (an entry, once written, never moves
until an explicit resize, so pointers into it stay valid) and
iterable by a linear scan of the flat array.

That conservatism has a cost. Level $i$ is loaded to $\alpha$ at
its planned input $o_i$, with up to $o_{i+1}$ of those keys budgeted
to move on, but the mean overflow sits below the quantile bound
wherever the chosen bound is not the deterministic cap, thinly on
the loaded top levels and by an order of magnitude or more on
undercommitted residues. A typical run therefore moves fewer
keys down than budgeted, and the shortfall compounds down the
cascade: the top levels run at essentially a little over their planned load
$\alpha$, and the deep levels, provisioned for overflow that rarely
arrives, stay nearly empty. Keys shift toward the top levels,
positive lookups shorten, and neither $C$ nor the guarantee changes;
a negative lookup, however, still visits the nearly empty tail,
which the truncation of Section~\ref{sec:depth} closes early.

\subsection{Depth reduction}\label{sec:depth}

The efficient table meets its load target but ends in a tail of
nearly empty levels, and a negative lookup pays one probe per level,
so depth is what we cut next. The tilted bounds $B_3$ and
$B_4$ already cut it at every level through the minimum
\eqref{eq:Bdef}. This subsection adds two structural cuts: a
progressive overfill of the top levels, and an exactly quantified
truncation of the tail, which the overfill makes affordable. Both
keep the per-level failure budget $\delta$, and the logarithmic
depth guarantee of Appendix~\ref{app:asymptotics} survives the
truncation. Section~\ref{sec:quantiles} then cuts most of the bounds'
slack.

\paragraph{Progressive marginal overfill.}\label{par:overfill}
The first cut changes how the bucket budget is spent, not how tails
are bounded. We fix the budget $\lceil q/(\alpha s)\rceil$ up front
and overfill the top: for a configurable overfill factor $r>1$, level $1$
takes the largest $n_1$ whose planned load is at least
$\min(0.99,\,r\alpha)$, and each later level $i$ targets
$\alpha_i:=\min\bigl(\alpha,\ r\,o_i/(sN_i)\bigr)$, with $N_i$ the
unspent part of the bucket budget: the load at which the residue
$o_i$ exactly fills the remaining buckets, boosted by $r$ and capped
at $\alpha$. Level sizes decay geometrically, so a few points of
extra load on the large top levels progressively lighten every
level below, and the cascade typically ends sooner. The mechanism
does not depend on how a level's residue is estimated: driven by
the worst-case bound $B$ of Definition~\ref{def:efficient} it keeps
the per-level failure budget $\delta$ intact and composes with the
truncation that follows. The factor $r$ is itself a tuning knob:
positive lookups pay the expected depth of their key, negative
lookups pay one probe per level, so a workload mix (e.g., $30\%$
positive, $70\%$ negative) fixes a weighted expected lookup cost as
a function of $r$, and we pick the $r$ minimizing it.

\paragraph{Truncating the tail.}
The second cut closes the cascade early: once the residue is
small, two levels finish the table, a wrap-up level sized so that
its overflow is at most $s$ with probability $1-\delta$, and a
final single bucket that absorbs that overflow. At every level the
construction tests whether this wrap-up pair fits the unspent part
of the bucket budget $\lceil q/(\alpha s)\rceil$ and truncates at
the first level where it does, and the overfill above frees the
budget the pair spends.
Our implementation lets the test tolerate a small budget overrun
(we allow $0.5\%$), which can close the cascade a level earlier at
a capacity cost bounded by the same $0.5\%$.

\begin{corollary}[Depth truncation]\label{cor:shortcircuit}
In the construction of Definition~\ref{def:efficient}, with levels
$1,\dots,i-1$ sized by \eqref{eq:efficient-ni}, fix a level $i$ with
planned input $o_i>s$ and choose
\begin{equation}\label{eq:shortcircuit-n}
n_i:=\min\bigl\{n\ge1:\ B(o_i,n,\delta)<s+1\bigr\},
\end{equation}
which exists since $B\le B_2=\E[O^{(o_i,n)}]/\delta$, with
$O^{(o_i,n)}$ the overflow of $o_i$ keys in $n$ buckets, and the
mean \eqref{eq:mean-overflow} vanishes as $n\to\infty$. Allocate
level $i$ with $n_i$ buckets and one final level $i+1$ with a single
bucket, and stop the construction there. Then with probability at
least $(1-\delta)^{i}\ge1-i\delta$ the overflow $O_i$ of level $i$ is
at most $s$ and
all $q$ keys are stored, and the total capacity satisfies
\begin{equation}\label{eq:shortcircuit-capacity}
C\;\le\;\frac{q-o_i}{\alpha}+s\,(n_i+1).
\end{equation}
\end{corollary}

\begin{proof}
Levels $1,\dots,i-1$ are handled exactly as in
Theorem~\ref{thm:efficient}. Level $i$ is fed at most $o_i$ keys,
and its overflow is an integer, so $O_i\ge s+1$ implies
$O_i>B(o_i,n_i,\delta)$; by stochastic monotonicity
(Lemma~\ref{lem:capmono}(ii)) and Proposition~\ref{prop:levelbound}
that event has conditional probability at most $\delta$, and the
chain multiplies as in \eqref{eq:perfect-union}. On the complement
the last bucket receives at most $s$ keys. The capacity bound is
\eqref{eq:capacity-bound} truncated at level $i$.
\end{proof}

\begin{remark}[The tail, quantified]\label{rem:free-tail}
In the proof of Theorem~\ref{thm:efficient}(b) the per-level failure
bound is $0$, not $\delta$, for every level whose chosen bound is the
deterministic cap $o_{i+1}=o_i-s$, so the failure probability is
really $L^*\delta$ with $L^*:=\sum_{i<L}\ind{o_{i+1}<o_i-s}$ counting
the levels whose chosen bound sits strictly below the deterministic
cap. 
For $q=10^5$, $s=8$, $\alpha=0.9$, $\delta=2^{-10}$, 
\eqref{eq:Bdef} runs the cascade
as $(n_i)=(11200,2125,426,95,24,8,2,1,\dots)$: $14$ levels 
at load $0.900$ the truncation test then closes the
table at $L=13$, spending exactly the bucket budget
$\lceil q/(\alpha s)\rceil$.

Truncating the baseline at $i=4$ isolates what the tilt does for
the test \eqref{eq:shortcircuit-n}. 
Demanding \emph{no} overflow at the wrap-up level instead
($n\,\PR[\Bin(o_4,1/n)\ge s{+}1]\le\delta$) needs $1128$ buckets:
tolerating an overflow of at most $s$ and adding one bucket is
cheaper, which is why we keep the final single-bucket level.

The progressive overfill above compounds this: at $r=1.025$ the
extra load on the top levels lightens every level below and frees
the budget for the wrap-up pair sooner, so the same bound-driven
cascade with the truncation test closes at $L=7$ in place of $13$
at load $0.900$, keeping the per-level failure budget $\delta$
throughout. The two
structural cuts together take the proven-bound sizing from $14$
levels to $7$. The remaining gap to the $5$ levels of the near-exact
sizing of Section~\ref{sec:quantiles} is the price of keeping a
theorem behind every level.
\end{remark}

\paragraph{Asymptotics.}
The asymptotic guarantees are stated and proved in
Appendix~\ref{app:asymptotics}:
\begin{center}
\fbox{\parbox{0.92\linewidth}{For every fixed configuration
$(s,\alpha,\delta)$ with $\alpha<1$, the cascade has $O(\log q)$
levels, and on a successful build a positive lookup costs $O(1)$
expected probes, the constants depending only on $(s,\alpha,\delta)$
(Proposition~\ref{prop:depth}, Corollary~\ref{cor:probes}).}}
\end{center}

\subsection{Near-exact overflow approximation}\label{sec:quantiles}

Every sizing so far consumes a proven tail bound, and a bound
carries slack. We can give up the closed
form and compute an approximate $(1-\delta)$-quantile
$\widehat{Q}(o,n)$ of the overflow $O$ itself: the true quantile
sits below any valid bound on it, and $\widehat{Q}$ tracks it
closely, so levels shrink further. Figure~\ref{fig:approx}
quantifies both statements at $n=256$: $\widehat{Q}$ stays within
$2$ keys of the simulated quantile, and the envelope $B$ reaches
more than twice $\widehat{Q}$ at low load and comes within $0.2\%$
of it at high load. The concession is that the configuration is backed by
a computation checked against simulation rather than by a theorem.
The computation is a fast four-way dispatch: 1) exact enumeration
for tiny inputs, 2) the deterministic $o-sn$ once a union bound
makes every bucket full with probability $1-\delta$, 3) a
Cornish--Fisher quantile built from de-Poissonized cumulants of $O$
when $\E[O]$ is large, which matches Monte Carlo within 3 keys at
quantiles near $4\times10^4$, and 4) a compound-Poisson tail
otherwise. Our sizing recipe applies this dispatch under the
overfill policy above: each level takes the largest $n_i$ whose
quantile-sized load meets its target and hands down its overflow
estimate, $o_{i+1}=\widehat{Q}(o_i,n_i)$, except the first, which
hands down the residue its target leaves,
$o_2=\lceil q-\alpha_1 s n_1\rceil$; the cascade truncates as in
Corollary~\ref{cor:shortcircuit}. Once the residue is under $2\%$ of
$q$, the recipe tests at every level whether the wrap-up pair, here
the smallest $n$ satisfying $\widehat{Q}\le s$ plus one final
bucket, fits the unspent bucket budget within the configured
overrun margin, closing the table where it first fits. Computed as for the figures, at
$(q,s,\alpha,\delta)=(10^5,8,0.9,2^{-10})$ and $r=1.1$ the recipe
yields $5$ levels, $(n_i)=(7334,4962,1401,65,1)$, and realized load
$0.908$, with no terminal overflow in $2\times10^4$ simulated
builds.
Figure~\ref{fig:framed} in Section~\ref{sec:variations} is computed
by the same dispatch.
Figure~\ref{fig:lookup} compares the two sizings across load
factors, the full bound \eqref{eq:Bdef} with its truncation test
against the approximation recipe, both under the overfill policy at
$r=1.025$. The approximation's advantage is concentrated in depth:
the constructions agree below $\alpha\approx0.75$, and bounds then
cost $1$ to $3$ extra levels ($7$ against $5$ at $\alpha=0.9$, the
overfill trimming the $13$ levels of Remark~\ref{rem:free-tail} to
$7$), pure negative-lookup overhead, while the expected
positive-lookup cost never differs by more than about $1\%$ of a
probe. The bound's conservatism adds extra top-level buckets that
lighten every level below, so the positive-lookup cost barely
moves, and at $\alpha=0.96$ the comparison even flips in the
bound's favor.

The curves of Figure~\ref{fig:lookup} also show that the cascade has
no cliff near full load, where linear probing has one. Under linear probing, Knuth's
analysis~\cite{knuth1998taocp3} puts the expected number of probes
at about $\frac12\bigl(1+1/(1-\alpha)^2\bigr)$ for an unsuccessful
search and $\frac12\bigl(1+1/(1-\alpha)\bigr)$ for a successful one,
unbounded as $\alpha\to1$: $50.5$ and $5.5$ at $\alpha=0.9$, $313$
and $13$ at $\alpha=0.96$. A multitable lookup probes at most one
bucket per level, counting one bucket of $s$ slots, scanned with
SIMD, as one probe where linear probing counts single slots. With
the near-exact approximation the level count rises from $3$ at
$\alpha=0.5$ to $6$ at $\alpha=0.96$, about linearly in
$\log(1/(1-\alpha))$, an observed rate rather than a
proven one, while a positive lookup costs $1.01$ to about $1.5$
bucket probes in expectation.

\begin{figure}[t]
\centering
\definecolor{faOIblue}{HTML}{0072B2}
\definecolor{faOIgreen}{HTML}{009E73}
\definecolor{faRefGray}{HTML}{808080}
\begin{tikzpicture}[baseline=(current bounding box.north)]
\begin{groupplot}[
    group style={group size=2 by 1, horizontal sep=1.35cm},
    width=0.5\textwidth,
    height=6cm,
    xmode=log, ymode=log,
    log basis x=2,
    xmin=1, xmax=145,
    xtick={1,2,4,8,16,32,64,128},
    xticklabels={1,2,4,8,16,32,64,128},
    xlabel={mean load $\lambda = q/n$},
    legend cell align=left,
    tick label style={font=\footnotesize},
    label style={font=\footnotesize},
    axis lines=left,
]
% ---------------------------------------------------------- left panel
\nextgroupplot[
    ymin=0.0001, ymax=2,
    ytick={1e-4,1e-3,1e-2,1e-1,1},
    ylabel={overflow per entry (value$/q$)},
    ylabel style={yshift=-0.15cm},
    legend style={font=\scriptsize, at={(0.985,0.015)}, anchor=south east,
        draw=none, fill=none},
]
% proven envelope B
\addplot[forget plot, black, double, thin] coordinates {
(1,0.0011378) (1.25,0.0057365) (1.5,0.0090637) (1.75,0.0092593) (2,0.0097756)
(2.25,0.010621) (2.5,0.011824) (3,0.015474) (3.5,0.021144) (4,0.029279)
(5,0.05432) (6,0.091672) (8,0.18045) (10,0.27218) (12,0.36107)
(16,0.50495) (20,0.60147) (24,0.66737) (32,0.7503) (48,0.83345)
(64,0.87506) (96,0.91669) (128,0.93751)
};
% near-exact approximation Qhat (four-way dispatch); no point where Qhat = 0
\addplot[forget plot, faOIblue, solid, thin, mark=*, mark size=0.9pt,
    mark options={fill=faOIblue, draw=faOIblue}] coordinates {
(1.25,0.003125) (1.5,0.0052083) (1.75,0.0044643) (2,0.0058594) (2.25,0.0069444)
(2.5,0.0078125) (3,0.011719) (3.5,0.016741) (4,0.024414) (5,0.04375)
(6,0.07487) (8,0.16113) (10,0.26172) (12,0.35677) (16,0.50415)
(20,0.60117) (24,0.66732) (32,0.75) (48,0.83333) (64,0.875)
(96,0.91667) (128,0.9375)
};
% exact (1-delta)-quantile from the conditional-binomial DP, over the
% sub-range where that DP is tractable at n = 256 -- ground truth for Qhat
\addplot[forget plot, faOIgreen, solid, thin, mark=square*, mark size=0.5pt,
    mark options={fill=faOIgreen, draw=faOIgreen}] coordinates {
(1.25,0.003125) (1.5,0.0026042) (1.75,0.0044643) (2,0.0058594) (2.25,0.0069444)
(2.5,0.0078125) (3,0.011719) (3.5,0.015625) (4,0.022461) (5,0.042969)
(6,0.074219) (8,0.16064)
};
% Monte Carlo (1-delta)-quantile: the 1025th largest of 2^20 runs,
% drawn last, and omitted where it is 0
\addplot[forget plot, black, only marks, mark=o, mark size=1.5pt] coordinates {
(1.25,0.003125) (1.5,0.0026042) (1.75,0.0044643) (2,0.0058594) (2.25,0.0069444)
(2.5,0.0078125) (3,0.011719) (3.5,0.015625) (4,0.022461) (5,0.042969)
(6,0.074219) (8,0.16064) (10,0.26133) (12,0.35645) (16,0.50391)
(20,0.60098) (24,0.66699) (32,0.75) (48,0.83333) (64,0.875)
(96,0.91667) (128,0.9375)
};
\addlegendimage{black, thin, double}
\addlegendentry{$B$ (proven envelope)}
\addlegendimage{faOIblue, solid, thin, mark=*, mark size=0.7pt}
\addlegendentry{$\widehat{Q}$ (approximation)}
\addlegendimage{faOIgreen, solid, thin, mark=square*, mark size=0.6pt}
\addlegendentry{exact $(1-\delta)$-quantile}
\addlegendimage{black, only marks, mark=o, mark size=1.5pt}
\addlegendentry{MC quantile, $2^{20}$ runs}
% --------------------------------------------------------- right panel
\nextgroupplot[
    ymin=0.4, ymax=2.8,
    ytick={0.5,0.7,1,1.5,2},
    yticklabels={0.5,0.7,1,1.5,2},
    minor ytick={0.6,0.8,0.9,1.2,2.5},
    ylabel={ratio to $\widehat{Q}$},
    ylabel style={yshift=-0.1cm},
    legend style={font=\scriptsize, at={(0.97,0.97)}, anchor=north east,
        draw=none, fill=none},
]
% unit reference line
\addplot[forget plot, faRefGray, very thin, densely dashed, no marks,
    samples=2, domain=0.9:145] {1};
% B / Qhat -- the slack the approximation removes
\addplot[forget plot, black, thin, double, mark=*, mark size=1.1pt,
    mark options={fill=black, fill opacity=0.4, draw opacity=0.4}] coordinates {
(1.25,1.8357) (1.5,1.7402) (1.75,2.0741) (2,1.6684) (2.25,1.5294)
(2.5,1.5134) (3,1.3205) (3.5,1.263) (4,1.1993) (5,1.2416)
(6,1.2244) (8,1.1199) (10,1.04) (12,1.012) (16,1.0016)
(20,1.0005) (24,1.0001) (32,1.0004) (48,1.0001) (64,1.0001)
(96,1) (128,1)
};
% MC / Qhat -- fidelity of the approximation
\addplot[forget plot, black, only marks, mark=o, mark size=1.5pt] coordinates {
(1.25,1) (1.5,0.5) (1.75,1) (2,1) (2.25,1)
(2.5,1) (3,1) (3.5,0.93333) (4,0.92) (5,0.98214)
(6,0.9913) (8,0.99697) (10,0.99851) (12,0.99909) (16,0.99952)
(20,0.99968) (24,0.99951) (32,1) (48,1) (64,1)
(96,1) (128,1)
};
\addlegendimage{black, thin, double, mark=*, mark size=1.1pt, mark options={fill=black, fill opacity=0.4, draw opacity=0.4}}
\addlegendentry{$B/\widehat{Q}$}
\addlegendimage{black, only marks, mark=o, mark size=1.5pt}
\addlegendentry{MC$/\widehat{Q}$}
\addlegendimage{faRefGray, very thin, densely dashed}
\addlegendentry{ratio $=1$}
\end{groupplot}
\end{tikzpicture}
\caption{The proven envelope $B$ of \eqref{eq:Bdef}, the
approximation $\widehat{Q}$, the exact $(1-\delta)$-quantile where
tractable ($\lambda\le8$), and the empirical $(1-\delta)$-quantile
of $2^{20}$ simulated runs (circles), all normalized per key, versus
the mean load $\lambda=q/n$ at $s=8$, $n=256$, $\delta=2^{-10}$
(left); $B$ and the simulated quantile as ratios to $\widehat{Q}$
(right). The exact and empirical quantiles coincide at every
tractable point, and $\widehat{Q}$ is within $2$ keys of both. $B$
exceeds $2\,\widehat{Q}$ at $\lambda=1.75$, runs $12$ to $32$
percent above it on $3\le\lambda\le8$, and agrees within $0.2\%$
from $\lambda=16$ on. The quantile curves start just above the onset
$\lambda\approx1.18$ of $\PR[O\ge1]>\delta$, below which only the
bound is positive.}
\label{fig:approx}
\end{figure}
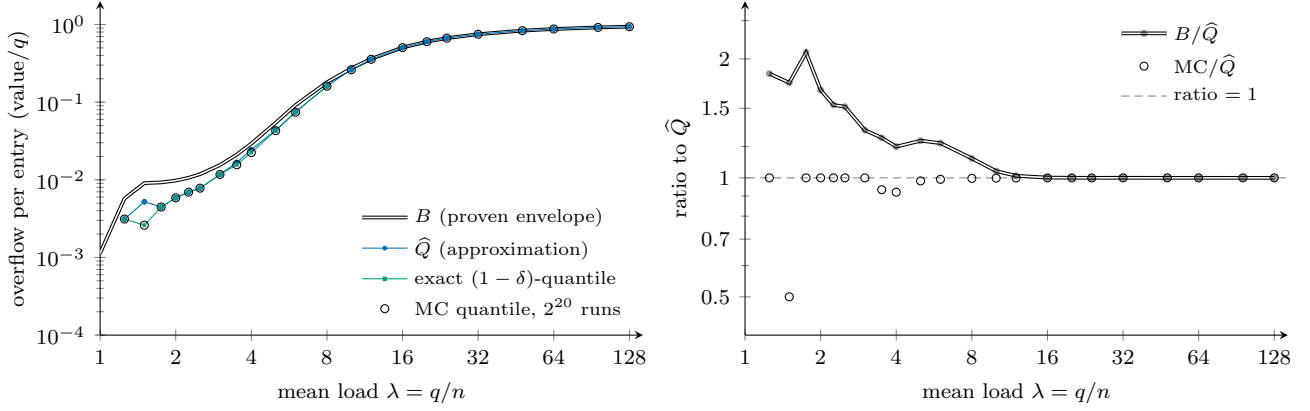

\begin{figure}[t]
\centering
% baseline=north top-aligns the two panels side by side
\begin{tikzpicture}[baseline=(current bounding box.north)]
\begin{axis}[
    width=0.48\textwidth,
    height=4.8cm,
    xmin=0.48, xmax=0.98,
    xlabel={target load factor $\alpha$},
    ylabel={expected positive-lookup probes},
    ylabel style={yshift=-0.15cm},
    legend style={font=\footnotesize, at={(0.02,0.98)}, anchor=north west, draw=none, fill=none},
    tick label style={font=\footnotesize},
    label style={font=\footnotesize},
    axis lines=left,
]
\addplot+[mark=*, mark size=1.2pt, blue, thick] coordinates {
(0.50,1.01196) (0.55,1.01902) (0.60,1.02909) (0.65,1.04296) (0.70,1.06204) (0.75,1.08854) (0.80,1.12516) (0.85,1.17929) (0.90,1.26395) (0.92,1.31605) (0.94,1.39024) (0.95,1.44358) (0.96,1.49785)
};
\addlegendentry{bounded}
\addplot+[mark=square*, mark size=1.2pt, red!70!black, thick] coordinates {
(0.50,1.01153) (0.55,1.01835) (0.60,1.02795) (0.65,1.04123) (0.70,1.05935) (0.75,1.08418) (0.80,1.11893) (0.85,1.17024) (0.90,1.25305) (0.92,1.30397) (0.94,1.37784) (0.95,1.43032) (0.96,1.50510)
};
\addlegendentry{near-exact approximation}
\end{axis}
\end{tikzpicture}\hfill
\begin{tikzpicture}[baseline=(current bounding box.north)]
\begin{axis}[
    width=0.48\textwidth,
    height=4.8cm,
    xmin=0.48, xmax=0.98,
    ymin=2, ymax=10,
    ytick={2,3,4,5,6,7,8,9,10},
    xlabel={target load factor $\alpha$},
    ylabel={levels $L$ (negative-lookup probes)},
    ylabel style={yshift=-0.15cm},
    legend style={font=\footnotesize, at={(0.02,0.98)}, anchor=north west, draw=none, fill=none},
    tick label style={font=\footnotesize},
    label style={font=\footnotesize},
    axis lines=left,
]
\addplot+[mark=*, mark size=1.2pt, blue, thick] coordinates {
(0.50,3) (0.55,3) (0.60,4) (0.65,4) (0.70,4) (0.75,4) (0.80,5) (0.85,6) (0.90,7) (0.92,7) (0.94,8) (0.95,8) (0.96,9)
};
\addlegendentry{bounded}
\addplot+[mark=square*, mark size=1.2pt, red!70!black, thick] coordinates {
(0.50,3) (0.55,3) (0.60,4) (0.65,4) (0.70,4) (0.75,4) (0.80,4) (0.85,5) (0.90,5) (0.92,5) (0.94,6) (0.95,6) (0.96,6)
};
\addlegendentry{near-exact approximation}
\end{axis}
\end{tikzpicture}
\caption{Expected positive-lookup probes (left) and levels $L$, the
negative-lookup probe count (right), versus the target load factor
$\alpha$ at $q=10^5$, $s=8$, $\delta=2^{-10}$: sizing by the
full bound \eqref{eq:Bdef} with the
truncation test against sizing by the near-exact approximation
$\widehat{Q}$, both under the overfill policy with $r=1.025$.}
\label{fig:lookup}
\end{figure}
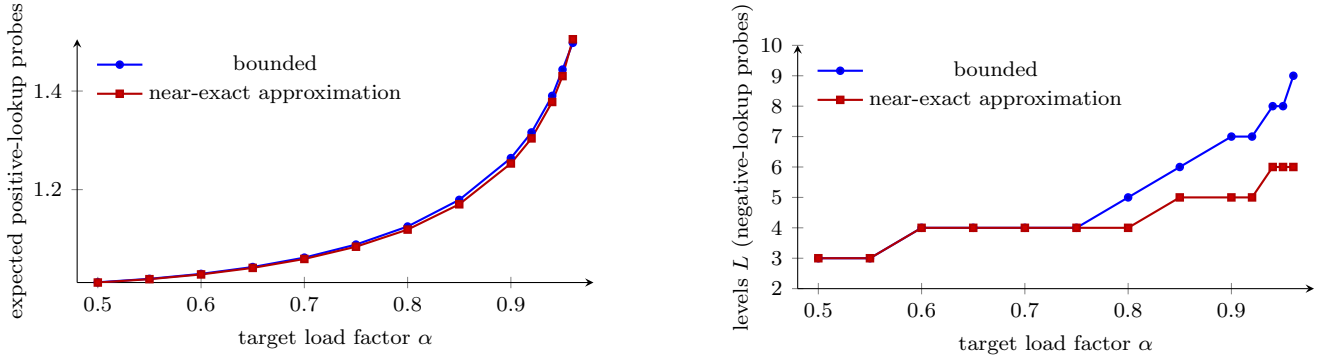

\section{Resizing and cache eviction}\label{sec:resizing}

\subsection{Resizing}\label{sec:grow}

The constructions of Section~\ref{sec:construction} assume the
capacity is known up front, but growth needs no rehash, because
nothing ties a level to the levels built before it: a level only
sees the keys that cascade into it, hashed by its own function, so
the table can resize without touching data at rest. Shortly before it reaches capacity we \emph{prepend} a fresh level
in front of level $1$, sized to keep the load factor above a chosen
floor, e.g., $0.8$. The trigger should precede the first terminal overflow, since by then the tail is full and the overflow of the new level would run straight into it; a simple one is a deep level, such as level $L-1$, reaching the key count the sizing planned for it. New keys hash into the prepended level first,
only its overflow trickles into the old structure, and no existing
key moves or is rehashed. When it fills in turn, we can either
prepend another (cheapest, though each adds a cost to lookups) 
or grow the prepended level in place and remap only its
own elements: a rehash of one small level, never of the whole table. 
The exact strategy will depend on the target load factor to ensure that new layer overflow doesn't overflow the table, e.g., when it is high we can redistribute entries from the small tail levels to the first one, clearing space for expected overflows at a marginal cost.
Prepending as the tail fills is also one of the three ways to
absorb deletion churn (Section~\ref{sec:churn}): under churn, level
$1$ turns away more inserts than the one-shot sizing planned for,
that flow runs into a tail sized for less, and a fresh level in
front takes it without moving a stored key.

A fixed-capacity table computes its levels largest-first
(Definitions~\ref{def:perfect} and~\ref{def:efficient}); a grown
table runs the picture in reverse, each prepended level larger
than everything behind it, for the cost of allocation alone. Hence
capacity moves in increments we choose, not in factors of two
($101$\,MB of payload can be held within a few percent of that, not
at the next power of two), and when deletions bring the load down
the table may shed its smallest levels, re-inserting their few keys.

\subsection{Cache tables}\label{sec:cache}

A cache needs its oldest entries evicted in bulk, and the level
structure gives us that almost for free: we build the table as a
ring of equal-sized levels (with an optional tail). Insertion
cascades from the logical first level, so long-lived entries sit
high while recent entries, finding the top levels full, land
deeper: on average, deeper levels hold newer data. When the table
crosses its load ceiling or an insertion terminally overflows, we
erase the logically first level (the oldest data) and advance a
wrapping counter: level $2$ becomes the new first level and the
erased level re-enters the ring empty at the back, catching the
deepest cascades. Eviction is one bulk clear, with no scans,
rehashing or timestamps, and to keep an entry alive we re-insert
it.

\section{Implementation}\label{sec:implementation}

We provide a Rust reference implementation \cite{manifoldGit}
with configurable bucket size, key and value sizes, and load factor,
in two variants: a \emph{plain} variant with almost no metadata, and a
\emph{filtered} variant that trades one byte per slot for fewer key
comparisons.

\paragraph{Memory layout.}
The plain variant stores, besides the buckets themselves, only the
per-level bucket counts and offsets and one length byte per bucket,
so that an all-zero key can never false-positive against an empty
slot; one byte per \emph{bucket}, against SwissTable's control byte
per \emph{slot}, is why the physical load factor tracks the slot
load factor so closely. Within a bucket all keys are stored
contiguously, then all values, and the key block is scanned linearly
with SIMD equality checks while the CPU fetches the adjacent value
block.

\paragraph{Bucket indexing.}
Levels have arbitrary bucket counts $n_i$, not powers of two, so a
hash is mapped to $[0,n_i)$ by Lemire's multiply-shift
method~\cite{lemire2019fast} without a division and, saving a
branch, without its rejection step, which skews bucket probabilities
by at most $2^{-32}$ and matters only at cryptographically low
failure budgets.

\paragraph{Byte filters.}
The filtered variant adds a one-byte filter per slot, checked with
SIMD before any key is compared, in the spirit of SwissTable's
control bytes~\cite{abseil2018swisstables,hashbrown2026crate}, with
one difference: where SwissTable spends a bit on occupancy and keeps
$7$ bits of hash fragment, we use all $8$ bits as fragment and
reserve two byte values instead, $0$ for empty and $1$ for a deleted
last slot, raising key fragments $0$ and $1$ to $2$. A mismatched
key then passes a slot's filter with probability
$262/2^{16}\approx1/250$, roughly half the $1/128$ of a $7$-bit
fragment, hence half the useless key comparisons. Because the filter
also identifies empty slots, it replaces the length byte, so at full
occupancy the filtered multitable stores $\ell/(\ell+1)$ payload bytes
per allocated byte for entries of $\ell$ bytes, $8/9\approx0.89$ for
$4$-byte keys and values against SwissTable's $0.777$: the same byte
per slot, but our slot load reaches one where theirs stops at $7/8$.
A delete clears the slot's filter byte to $0$ except in the last
slot, where it writes $1$, so a bucket's last byte is nonzero
exactly when the bucket has ever been full, the only way it can have
overflowed, and the insert probe uses that byte to end its duplicate
check early. The benchmark of Section~\ref{sec:performance} runs the
filtered multitable at $s=16$ and the plain one, which scans keys, at
$s=8$.

\paragraph{One hash per lookup.}
An operation computes a single $64$-bit hash, and each level past
the first reads its own $32$-bit window of it, rotated by an amount
fixed by the level index so that the window slides one byte per
level, and remaps the window as above. Nothing is stored per level,
and the cascade runs as deep as needed, $11$ levels in the plain
arms of Figure~\ref{fig:bench-lfsweep}. The reuse is a heuristic
outside the independence model of Section~\ref{sec:model}, and we
compensate by running the construction with a smaller per-level
$\delta$ than the failure rate we need: by \eqref{eq:B1} the
deviation term grows only like $\sqrt{\ln(1/\delta)}$, so tightening
$\delta$ by orders of magnitude costs a few percent of capacity at
$\alpha<1$ and, for the perfect table, slightly deepens the cascade
instead. The measured terminal overflow stays within the compensated
budget.

\paragraph{Choosing the failure budget.}
The bounds behind our guarantees are conservative: a level that
stays under plan starves the levels below, and the slack compounds.
The sizing recipe of Section~\ref{sec:quantiles} turns $10^4$ keys
at $s=8$, $\alpha=0.9$, $r=1.025$, $\delta=2^{-10}$ into
$(n_i)=(1035,285,57,15,1)$, terminal overflow bounded by about
$(L-1)\,\delta\approx0.004$, and in $10^7$ simulated builds it never
occurred. We hence pick $\delta$ so that the bound is merely
reasonable, e.g., $2^{-10}$, and let the compounding slack deliver
the reliability. See the implementation for remaining details \cite{manifoldGit}.

\subsection{Performance}\label{sec:performance}

The construction's promise is keys per byte, so we compare at equal
\emph{physical} memory and ask what that density costs in speed.
The baseline is hashbrown~\cite{hashbrown2026crate}, the strongest
we know of, on an Apple M2 Pro: multitable (filtered and plain) and
hashbrown are sized to the same footprint in bytes at every point,
matched within $0.5\%$ at every anchor (except the smallest $4{+}4$
table, where the two bucket sizes quantize $0.93\%$ apart),
multitable overfill factor $r=1.01$. We use FxHasher for both, since it's very fast and allows to minimize the role of hash function in evaluation; moreover we use the stable toolchain. The reference load is
hashbrown's own (for fair comparison): every size anchor holds $99\%$ of its capacity at a
power-of-two bucket count, leaving it just under its $7/8$
resize threshold, at slot load $0.866$, as full as it gets (we also benchmarked at 75\% of capacity with similar results, except the smaller multitable advantage on misses).
	Its one
control byte per slot turns that into a physical load
factor of $0.866\times8/9\approx0.770$ for the $4{+}4$ map's
$8$-byte elements; the same arithmetic gives $0.840$, $0.693$ and
$0.815$ for the other three shapes, and the multitable realizes
the same figures.

Four workloads: insert, positive and negative lookups, and a
half-hit half-miss mixed stream. Seven sizes per shape from
$3.5\cdot10^3$ to $2.3\cdot10^8$ keys step the working set from
cache into RAM (the $16{+}16$ shape's ladder tops out at
$1.16\cdot10^8$), and four shapes cover maps and sets: $4{+}4$ and
$16{+}16$ byte maps (Figures~\ref{fig:bench-kv4},
\ref{fig:bench-kv4-ratio} and~\ref{fig:bench-kv16}) and $4$- and
$16$-byte sets (Figures~\ref{fig:bench-set4}
and~\ref{fig:bench-set16}), the last three in
Appendix~\ref{app:benchmarks} with the load-factor sweep and the
sizing variants. The filtered multitable runs at $s=16$ and the plain one
at $s=8$, discussed below.

The multitable always hashes raw bytes, so the like-for-like arm
keys hashbrown on the same byte arrays (\emph{byte-keyed});
hashbrown's best case is a native integer key, so we also key it on
\texttt{u32}/\texttt{u128} for $4$- and $16$-byte keys (the
\emph{integer-keyed} arm, the conservative comparison), against which the byte-keyed
arm is about $11\%$ slower overall and $18\%$ on hits. All figures
report throughput, Criterion center estimates in Mops/s ($10^6$
operations per second); aggregates are arithmetic means over cells.
The full benchmark
methodology accompanies the implementation's repository
\cite{manifoldGit}.

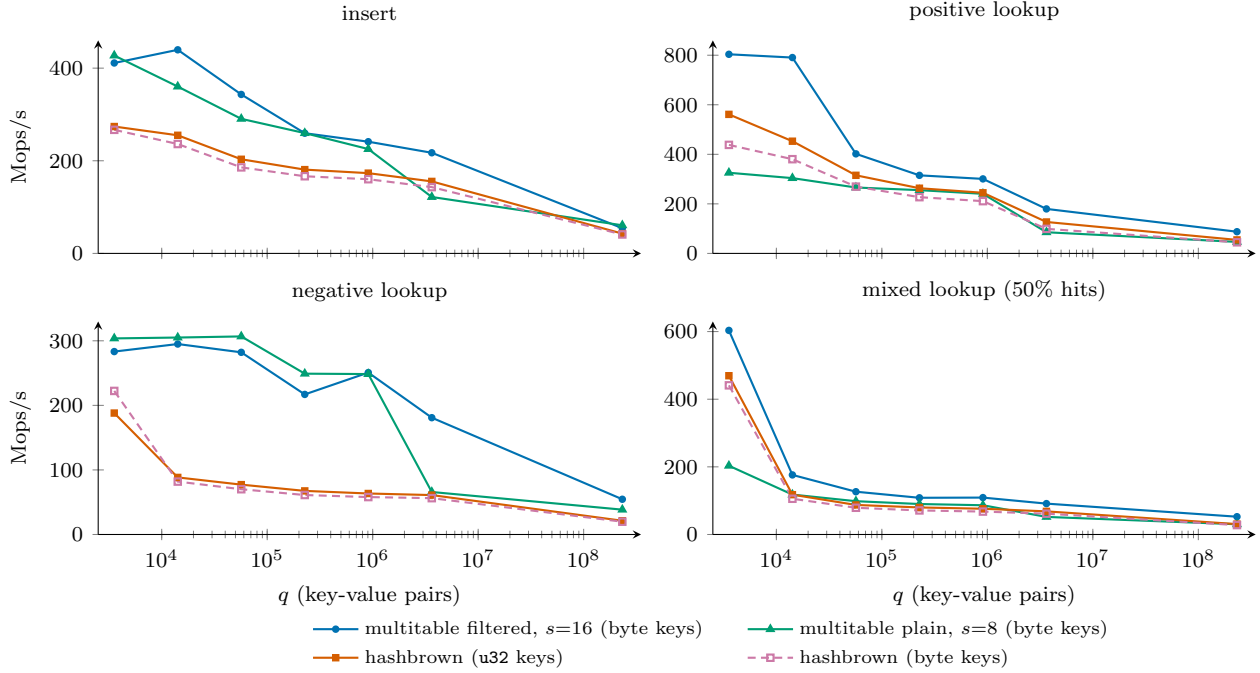
\begin{figure}[t]
\centering
\definecolor{fkOIblue}{HTML}{0072B2}
\definecolor{fkOIverm}{HTML}{D55E00}
\definecolor{fkOIgreen}{HTML}{009E73}
\definecolor{fkOIpink}{HTML}{CC79A7}
\begin{tikzpicture}[baseline=(current bounding box.north)]
\begin{groupplot}[
    group style={group size=2 by 2,
        horizontal sep=0.95cm, vertical sep=0.9cm,
        xlabels at=edge bottom, xticklabels at=edge bottom,
        ylabels at=edge left},
    width=0.5\textwidth,
    height=4.4cm,
    xmode=log,
    xmin=2500, xmax=3.5e8,
    ymin=0,
    xtick={1e4,1e5,1e6,1e7,1e8},
    xlabel={$q$ (key-value pairs)},
    ylabel={Mops/s},
    title style={font=\footnotesize, yshift=-2pt},
    tick label style={font=\footnotesize},
    label style={font=\footnotesize},
    axis lines=left,
]
% ------------------------------------------------------------- insert
\nextgroupplot[title={insert}, ymax=460]
\addplot[fkOIblue, thick, mark=*, mark size=1pt,
    mark options={fill=fkOIblue, draw=fkOIblue}] coordinates {
(3549,411.133) (14193,439.580) (56771,343.218) (227083,259.518)
(908330,241.196) (3633317,217.278) (232532215,54.465)
};
\addplot[fkOIgreen, thick, mark=triangle*, mark size=1.4pt,
    mark options={fill=fkOIgreen, draw=fkOIgreen}] coordinates {
(3549,427.259) (14193,360.231) (56771,290.571) (227083,259.734)
(908330,225.332) (3633317,121.804) (232532215,60.729)
};
\addplot[fkOIverm, thick, mark=square*, mark size=1pt,
    mark options={fill=fkOIverm, draw=fkOIverm}] coordinates {
(3549,273.928) (14193,255.043) (56771,203.231) (227083,180.901)
(908330,173.286) (3633317,155.448) (232532215,42.721)
};
\addplot[fkOIpink, thick, densely dashed, mark=square, mark size=1.2pt,
    mark options={draw=fkOIpink, fill=none, solid}] coordinates {
(3549,266.823) (14193,236.167) (56771,185.756) (227083,166.570)
(908330,160.182) (3633317,143.111) (232532215,40.873)
};
% ---------------------------------------------------- positive lookup
\nextgroupplot[title={positive lookup}, ymax=860]
\addplot[fkOIblue, thick, mark=*, mark size=1pt,
    mark options={fill=fkOIblue, draw=fkOIblue}] coordinates {
(3549,803.729) (14193,790.576) (56771,401.865) (227083,315.110)
(908330,300.788) (3633317,179.617) (232532215,87.593)
};
\addplot[fkOIgreen, thick, mark=triangle*, mark size=1.4pt,
    mark options={fill=fkOIgreen, draw=fkOIgreen}] coordinates {
(3549,325.563) (14193,303.914) (56771,265.993) (227083,255.382)
(908330,240.506) (3633317,85.397) (232532215,46.442)
};
\addplot[fkOIverm, thick, mark=square*, mark size=1pt,
    mark options={fill=fkOIverm, draw=fkOIverm}] coordinates {
(3549,561.230) (14193,452.899) (56771,315.259) (227083,263.116)
(908330,244.726) (3633317,126.852) (232532215,54.492)
};
\addplot[fkOIpink, thick, densely dashed, mark=square, mark size=1.2pt,
    mark options={draw=fkOIpink, fill=none, solid}] coordinates {
(3549,438.078) (14193,380.329) (56771,269.644) (227083,226.994)
(908330,210.962) (3633317,99.089) (232532215,44.820)
};
% ---------------------------------------------------- negative lookup
\nextgroupplot[title={negative lookup}, ymax=330]
\addplot[fkOIblue, thick, mark=*, mark size=1pt,
    mark options={fill=fkOIblue, draw=fkOIblue}] coordinates {
(3549,283.310) (14193,294.994) (56771,282.159) (227083,216.953)
(908330,250.784) (3633317,180.914) (232532215,54.561)
};
\addplot[fkOIgreen, thick, mark=triangle*, mark size=1.4pt,
    mark options={fill=fkOIgreen, draw=fkOIgreen}] coordinates {
(3549,303.721) (14193,305.064) (56771,306.786) (227083,249.085)
(908330,248.466) (3633317,65.815) (232532215,38.462)
};
\addplot[fkOIverm, thick, mark=square*, mark size=1pt,
    mark options={fill=fkOIverm, draw=fkOIverm}] coordinates {
(3549,188.076) (14193,88.397) (56771,77.228) (227083,67.530)
(908330,63.493) (3633317,61.039) (232532215,21.290)
};
\addplot[fkOIpink, thick, densely dashed, mark=square, mark size=1.2pt,
    mark options={draw=fkOIpink, fill=none, solid}] coordinates {
(3549,222.370) (14193,81.922) (56771,70.095) (227083,61.133)
(908330,57.955) (3633317,56.267) (232532215,19.795)
};
% ------------------------------------------------------- mixed lookup
\nextgroupplot[
    title={mixed lookup (50\% hits)},
    ymax=630,
    legend style={font=\scriptsize, at={(0,-0.35)}, anchor=north,
        draw=none, fill=none,
        /tikz/every even column/.append style={column sep=0.5cm}},
    legend columns=2,
    legend cell align=left,
]
\addplot[fkOIblue, thick, mark=*, mark size=1pt,
    mark options={fill=fkOIblue, draw=fkOIblue}] coordinates {
(3549,603.464) (14193,176.146) (56771,126.497) (227083,108.512)
(908330,109.083) (3633317,91.278) (232532215,52.577)
};
\addlegendentry{multitable filtered, $s{=}16$ (byte keys)}
\addplot[fkOIgreen, thick, mark=triangle*, mark size=1.4pt,
    mark options={fill=fkOIgreen, draw=fkOIgreen}] coordinates {
(3549,203.289) (14193,117.980) (56771,98.436) (227083,90.098)
(908330,86.450) (3633317,51.987) (232532215,30.221)
};
\addlegendentry{multitable plain, $s{=}8$ (byte keys)}
\addplot[fkOIverm, thick, mark=square*, mark size=1pt,
    mark options={fill=fkOIverm, draw=fkOIverm}] coordinates {
(3549,469.462) (14193,117.726) (56771,87.355) (227083,79.928)
(908330,76.145) (3633317,68.334) (232532215,30.658)
};
\addlegendentry{hashbrown (\texttt{u32} keys)}
\addplot[fkOIpink, thick, densely dashed, mark=square, mark size=1.2pt,
    mark options={draw=fkOIpink, fill=none, solid}] coordinates {
(3549,441.034) (14193,105.635) (56771,78.996) (227083,71.119)
(908330,67.740) (3633317,61.525) (232532215,27.815)
};
\addlegendentry{hashbrown (byte keys)}
\end{groupplot}
\end{tikzpicture}
\caption{Throughput on an Apple M2 Pro for the $4{+}4$-byte map shape
at overfill factor \hyperref[par:overfill]{$r=1.01$}: the filtered
multitable at $s=16$ and the plain one at $s=8$ against hashbrown's
\texttt{HashMap} in both of its keyings, the multitable's own byte
arrays and native \texttt{u32}, at equal physical memory per size,
versus stored key-value pairs $q$. Every series realizes physical
load factor $0.770$, and hashbrown's two keyings are one table. The
key axis is logarithmic; throughput is linear, each panel on its own
range, higher is better.}
\label{fig:bench-kv4}
\end{figure}

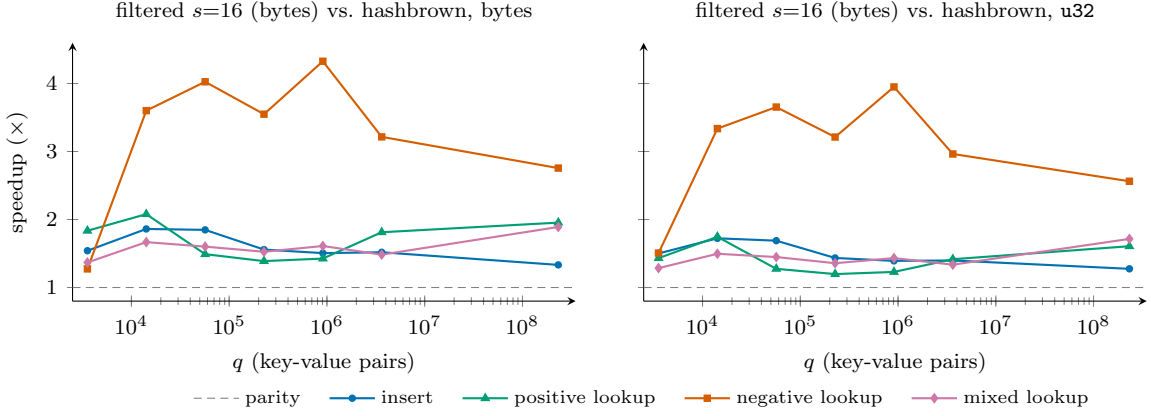
\begin{figure}[t]
\centering
\definecolor{fkrOIblue}{HTML}{0072B2}
\definecolor{fkrOIgreen}{HTML}{009E73}
\definecolor{fkrOIverm}{HTML}{D55E00}
\definecolor{fkrOIpurp}{HTML}{CC79A7}
\definecolor{fkrRefGray}{HTML}{808080}
\begin{tikzpicture}[baseline=(current bounding box.north)]
\begin{groupplot}[
    group style={group size=2 by 1,
        horizontal sep=0.9cm,
        ylabels at=edge left, yticklabels at=edge left},
    width=0.47\textwidth,
    height=5cm,
    xmode=log,
    xmin=2500, xmax=3.5e8,
    ymin=0.8, ymax=4.6,
    xtick={1e4,1e5,1e6,1e7,1e8},
    ytick={1,2,3,4},
    xlabel={$q$ (key-value pairs)},
    ylabel={speedup ($\times$)},
    title style={font=\footnotesize, yshift=-2pt},
    tick label style={font=\footnotesize},
    label style={font=\footnotesize},
    axis lines=left,
]
% -------------------------------------------- vs. hashbrown, byte keys
\nextgroupplot[title={filtered $s{=}16$ (bytes) vs.\ hashbrown, bytes}]
\addplot[fkrRefGray, densely dashed, thin, forget plot] coordinates {
(2500,1) (3.5e8,1)
};
\addplot[fkrOIblue, thick, mark=*, mark size=1pt,
    mark options={fill=fkrOIblue, draw=fkrOIblue}] coordinates {
(3549,1.5408) (14193,1.8613) (56771,1.8477) (227083,1.5580)
(908330,1.5058) (3633317,1.5183) (232532215,1.3326)
};
\addplot[fkrOIgreen, thick, mark=triangle*, mark size=1.4pt,
    mark options={fill=fkrOIgreen, draw=fkrOIgreen}] coordinates {
(3549,1.8347) (14193,2.0787) (56771,1.4904) (227083,1.3882)
(908330,1.4258) (3633317,1.8127) (232532215,1.9543)
};
\addplot[fkrOIverm, thick, mark=square*, mark size=1pt,
    mark options={fill=fkrOIverm, draw=fkrOIverm}] coordinates {
(3549,1.2740) (14193,3.6009) (56771,4.0254) (227083,3.5489)
(908330,4.3272) (3633317,3.2153) (232532215,2.7563)
};
\addplot[fkrOIpurp, thick, mark=diamond*, mark size=1.4pt,
    mark options={fill=fkrOIpurp, draw=fkrOIpurp}] coordinates {
(3549,1.3683) (14193,1.6675) (56771,1.6013) (227083,1.5258)
(908330,1.6103) (3633317,1.4836) (232532215,1.8903)
};
% ---------------------------------------------- vs. hashbrown, u32 keys
\nextgroupplot[
    title={filtered $s{=}16$ (bytes) vs.\ hashbrown, \texttt{u32}},
    legend style={font=\scriptsize, at={(0,-0.3)}, anchor=north,
        draw=none, fill=none, legend columns=5,
        /tikz/every even column/.append style={column sep=0.3cm}},
    legend cell align=left,
]
\addplot[fkrRefGray, densely dashed, thin] coordinates {
(2500,1) (3.5e8,1)
};
\addlegendentry{parity}
\addplot[fkrOIblue, thick, mark=*, mark size=1pt,
    mark options={fill=fkrOIblue, draw=fkrOIblue}] coordinates {
(3549,1.5009) (14193,1.7235) (56771,1.6888) (227083,1.4346)
(908330,1.3919) (3633317,1.3977) (232532215,1.2749)
};
\addlegendentry{insert}
\addplot[fkrOIgreen, thick, mark=triangle*, mark size=1.4pt,
    mark options={fill=fkrOIgreen, draw=fkrOIgreen}] coordinates {
(3549,1.4321) (14193,1.7456) (56771,1.2747) (227083,1.1976)
(908330,1.2291) (3633317,1.4160) (232532215,1.6074)
};
\addlegendentry{positive lookup}
\addplot[fkrOIverm, thick, mark=square*, mark size=1pt,
    mark options={fill=fkrOIverm, draw=fkrOIverm}] coordinates {
(3549,1.5064) (14193,3.3371) (56771,3.6536) (227083,3.2127)
(908330,3.9498) (3633317,2.9639) (232532215,2.5627)
};
\addlegendentry{negative lookup}
\addplot[fkrOIpurp, thick, mark=diamond*, mark size=1.4pt,
    mark options={fill=fkrOIpurp, draw=fkrOIpurp}] coordinates {
(3549,1.2854) (14193,1.4962) (56771,1.4481) (227083,1.3576)
(908330,1.4326) (3633317,1.3358) (232532215,1.7149)
};
\addlegendentry{mixed lookup}
\end{groupplot}
\end{tikzpicture}
\caption{Per-operation speedup of the filtered multitable ($s=16$)
over hashbrown's \texttt{HashMap} on the $4{+}4$-byte shape, at equal
physical memory on an Apple M2 Pro at
\hyperref[par:overfill]{$r=1.01$}, versus the number of stored
key-value pairs $q$: left, hashbrown keyed on the same raw bytes;
right, on native \texttt{u32}. Ratios above one favor the
multitable.}
\label{fig:bench-kv4-ratio}
\end{figure}

Figure~\ref{fig:bench-kv4-ratio} replots the $4{+}4$ shape as
per-operation ratios, and all $56$ cells of the two panels sit above
parity. Against the byte-keyed arm the ratios run from $1.27\times$
to $4.33\times$, mean $2.04\times$ over the shape's
$28$ cells; against the integer-keyed arm
they run from $1.20\times$ to $3.95\times$, mean $1.84\times$.
Figure~\ref{fig:bench-kv4} shows where the ratios come from:
integer-keyed hashbrown's negative-lookup throughput falls from
$188$ to $88$ Mops/s after the smallest size while the filtered
multitable holds between $217$ and $295$ through $9\times10^5$ keys, and
every series loses $40$ to $64\%$ of its positive-lookup throughput
between $9\times10^5$ and $3.6\times10^6$ keys, where the $4{+}4$
map leaves cache.

Over the $84$ equal-memory configurations ($4$ shapes, $7$ sizes,
insert, hit and miss) the filtered multitable is ahead in every cell,
by $1.23$-$4.47\times$ against the byte-keyed arm and
$1.08$-$3.95\times$ against the integer-keyed one, and the mixed
stream, over its own $28$ cells, runs $1.56\times$ and $1.43\times$.
No size aggregate goes to hashbrown: the weakest anchor, the
smallest tables, still averages $1.52\times$ and $1.35\times$ over
its $16$ cells of four shapes and four operations. Split by size,
the five smallest anchors per shape (under $10$\,MB for the
$4$-byte shapes, $18$\,MB for the $16$-byte set and already $35$\,MB for the $16{+}16$ map) against
the two largest, the arithmetic means are those of Table~\ref{tab:perf}:
\begin{center}
\footnotesize
\begin{tabular}{lcccccccc}
 & \multicolumn{4}{c}{vs.\ byte-keyed} &
   \multicolumn{4}{c}{vs.\ integer-keyed}\\
arithmetic mean & insert & positive & negative & combined
               & insert & positive & negative & combined\\
\hline
all $84$ cells & $1.57\times$ & $1.57\times$ & $3.24\times$ & $2.13\times$
               & $1.47\times$ & $1.34\times$ & $2.94\times$ & $1.92\times$\\
five smallest anchors ($60$) & $1.63\times$ & $1.57\times$ & $3.36\times$ & $2.19\times$
               & $1.52\times$ & $1.35\times$ & $3.03\times$ & $1.97\times$\\
two largest anchors ($24$) & $1.43\times$ & $1.58\times$ & $2.93\times$ & $1.98\times$
               & $1.36\times$ & $1.30\times$ & $2.69\times$ & $1.78\times$\\
\end{tabular}
\captionof{table}{Throughput ratio of the filtered multitable ($s=16$) over hashbrown's \texttt{HashMap} at equal physical memory on an Apple M2 Pro: arithmetic means over the $84$ equal-memory cells ($4$ shapes, $7$ sizes, insert, positive and negative lookup) and over the five smallest and two largest size anchors per shape; \emph{combined} averages the three operations; ratios above one favor the multitable.}\label{tab:perf}
\end{center}
The miss advantage peaks at the smaller sizes, where the byte
filters reject most misses without a key comparison, and narrows at
the two largest without closing; the geometric means of the same $84$
ratios are $1.96\times$ and $1.75\times$ (byte and integer keys).

The filtered multitable runs at $s=16$ for the hit-miss trade: against
$s=8$, the wider bucket cuts positive-lookup time to $0.88\times$ and
adds $5\%$ to negative-lookup time. The plain multitable runs at $s=8$,
being slower at $s=16$ on every shape's aggregate, by $1.19\times$
overall.

The filter byte pays off most where entries are wide: at the shared
$s=8$ the plain multitable's per-operation time runs $1.13\times$ the
filtered one's on the $4$-byte set to $2.35\times$ on the
$16{+}16$ map (geometric means over insert, hit, and miss), so with
$16$ bytes to compare, one filter byte per slot is cheap insurance
against failing comparisons. With $4$-byte keys, though, the plain
multitable beats hashbrown outright in most cells on insert, negative and
mixed lookups, above all on the $4$-byte set
(Appendix~\ref{app:benchmarks}); with $16$-byte keys it trails, and positive
lookups are its weak operation, $8$ of $28$ cells against the
byte-keyed arm and $3$ against the integer-keyed one.

The other three shapes, charted in Appendix~\ref{app:benchmarks},
hold the same ordering, and the load-factor sweep of
Figure~\ref{fig:bench-lfsweep} holds throughput fixed instead of
memory: the filtered multitable matches hashbrown's mixed-lookup
throughput with up to $12\%$ fewer bytes, and in cache the plain
multitable with $18\%$ fewer, $22\%$ more keys in the same memory.
The plain multitable's throughput has no cliff as the load approaches
$1$, giving up $6$ to $25\%$ from realized physical load $0.86$ to
$0.977$ (Appendix~\ref{app:benchmarks}); the filtered multitable,
whose slot load on this shape is $9/8$ of its physical load, stops
near $0.88$. Rounding level counts to
powers of two, the sizing variant of
Section~\ref{sec:implementation} tested in
Appendix~\ref{app:benchmarks}, is a wash where the first level
already sits at a power of two and costs about $14\%$ elsewhere at
$s=16$ (Figure~\ref{fig:bench-variants}). The campaign's scope is one
machine and one baseline.

\section{Variations}\label{sec:variations}

\paragraph{Large keys.}
When keys are large and need not be iterated over, the table can store
a fixed-size hash of each key instead of the key itself, with the
usual false-positive trade-off (everything else is unchanged).

\paragraph{Framed table.}
The basic construction jumps between levels that may be far apart in
memory. When a lookup should stay inside one region (one SSD page,
one DRAM row or huge page, one disk sector, one server), locality
can be enforced with one extra indirection:
\begin{enumerate}[leftmargin=*,nosep]
\item Fix the bucket size $s$, the frame capacity $S$ (the region's
byte size over the entry size), and a budget $\delta$.
\item Choose the number of frames $N$ as large as the target load
factor $\alpha = q/(NS)$ allows, subject to some $q_{\max}$ with
$s<q_{\max}<S-s$ capping every frame's key count with probability $1-\delta$
($q_{\max}$ bounded via Lemma~\ref{lem:counts} and Markov, as
before).
\item If no $N$ meets the target load, enlarge the frame.
\item Set $\alpha_{\max}:=q_{\max}/(S-s)$ and build every frame as a
standalone multitable for $q_{\max}$ keys at load factor
$\alpha_{\max}$: by Theorem~\ref{thm:efficient}(a) the frame's
capacity is at most $q_{\max}/\alpha_{\max}+s=S$, so each frame fits
its region. The guarantees compose by a union bound: with probability
at least $1-\delta$ no frame receives more than $q_{\max}$ keys, and
each of the $N$ frames then builds with probability at least
$1-(L-1)\delta$ (Theorem~\ref{thm:efficient}(b), whose proof bounds
every level at any feed up to its plan), so the whole table succeeds
with probability at least $1-\delta-N(L-1)\delta$, $L$ being the
depth of one frame.
\end{enumerate}
A lookup maps the key's hash to a frame, then runs the ordinary
multitable lookup inside it: positive or negative, it touches a single
region, which is what a paged or sharded medium needs, and frames
parallelize trivially. The same arithmetic with single-bucket frames
yields a one-level table. In both roles a small concession helps:
instead of demanding that nothing overflow, we tolerate an overflow
of up to $s$ keys and catch it in one spare bucket, small enough to
pin in fast memory.

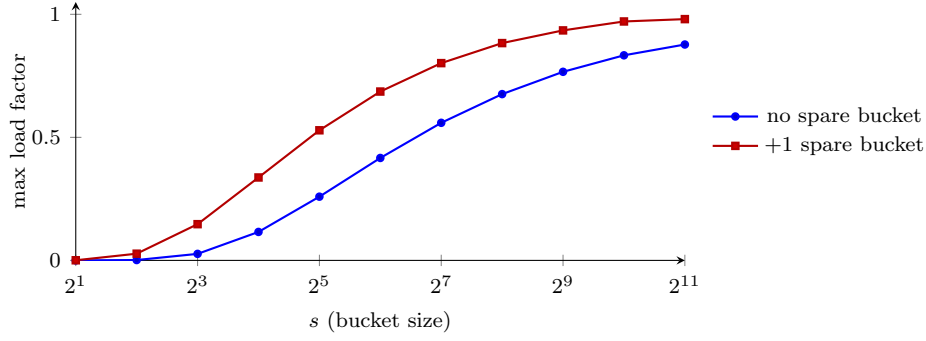
\begin{figure}[t]
\centering
\begin{tikzpicture}[baseline=(current bounding box.north)]
\begin{axis}[
    width=0.55\textwidth,
    height=5cm,
    xmode=log,
    log basis x=2,
    xmin=2, xmax=2048,
    enlarge x limits=0.04,
    xtick={2,8,32,128,512,2048},
    xlabel={$s$ (bucket size)},
    ylabel={max load factor},
    ylabel style={yshift=-0.15cm},
    % legend to the right of the panel, vertically centered, to save vertical space
    legend style={font=\footnotesize, at={(1.03,0.5)}, anchor=west, draw=none, fill=none, legend columns=1},
    tick label style={font=\footnotesize},
    label style={font=\footnotesize},
    axis lines=left,
    ymin=0, ymax=1.05,
]
\addplot+[mark=*, mark size=1.2pt, blue, thick] coordinates {
(2,0.000004) (4,0.001447) (8,0.026509) (16,0.115672) (32,0.258942) (64,0.416233) (128,0.559050) (256,0.675676) (512,0.766284) (1024,0.833333) (2048,0.877193)
};
\addlegendentry{no spare bucket}
\addplot+[mark=square*, mark size=1.2pt, red!70!black, thick] coordinates {
(2,0.000506) (4,0.027231) (8,0.147194) (16,0.336931) (32,0.528663) (64,0.686106) (128,0.801603) (256,0.883002) (512,0.934579) (1024,0.970874) (2048,0.980392)
};
\addlegendentry{+1 spare bucket}
\end{axis}
\end{tikzpicture}
\caption{Maximum load factor of a single level versus $s$,
$q=102400$, $\delta=2^{-20}$: the smallest $n$ with
$\PR[O\ge1]\le\delta$ (zero overflow) versus the smallest $n$ with
$\PR[O\ge s{+}1]\le\delta$ plus one spare bucket catching the at
most $s$ overflowing keys, the overflow quantiles computed by the
dispatch of Section~\ref{sec:quantiles}.}
\label{fig:framed}
\end{figure}
Figure~\ref{fig:framed} quantifies the concession: the spare
bucket lifts the maximum load factor at every $s$ shown, and the load factor
climbs steeply with $s$, reaching $0.98$ at $s=2048$ for
$q\approx10^5$, which is why framed configurations favor large
buckets.

\section{Deletion churn}\label{sec:churn}

The guarantees of Section~\ref{sec:construction} are one-shot. The
simplest way to admit deletions is
SwissTable's~\cite{abseil2018swisstables,hashbrown2026crate}: a
delete leaves a tombstone, so its slot stays counted against
capacity, and when the spare capacity runs out the table is rebuilt.
Hashbrown, the Rust port, rehashes in place, with no extra memory,
only if the live keys fill at most half its capacity; otherwise it
allocates a table of twice the size and moves every key. The
multitable restores itself more cheaply due to its clear-cut structure. A \emph{hoist} pass
re-inserts the keys of levels $2$ to $L$ into level $1$, which gives
a one-shot placement again (Theorem~\ref{thm:efficient} applies).
The pass runs at any load, so nothing forces a doubling above load
$0.5$; it touches only the keys beyond level $1$, a minority ($43\%$
in the $10^5$-key example), where a rehash touches every key; and it
needs no additional memory.

How often does the pass run? Under \emph{continuous churn} (worst case) the
table holds $m\le q$ keys, every step deletes a uniformly random
stored key and inserts a fresh one, and a \emph{turnover} is $m$
steps. A slot freed at level $i$ is refilled only by a new key that
reaches that level and hashes to its bucket, so holes accumulate in
the top levels, yet level $1$ still turns away about twice the share
of inserts it overflowed during the fill, and that flow runs into
deep levels sized for the one-shot residue. Each bucket is then full
only part of the time, hence no scheme that leaves stored keys in
place sustains load factor one (Appendix~\ref{app:churn}). We propose three
options tailored to different use-cases.

\begin{enumerate}[leftmargin=*,nosep]
\item \emph{Hoist on first overflow.} The in-place restore above,
usable alone with no resizing. The pass recurs and moves entries,
so it suits light or batch churn.
\item \emph{Grow as the tail fills.} Prepend a level before the first overflow
(Section~\ref{sec:grow}); no stored key moves. Combines with the 1st option.
\item \emph{Load-factor overload.} Size for $q(1+x)$ keys at a
raised target $\alpha'>\alpha$ (the target of
Definition~\ref{def:efficient}; $x=\alpha'/\alpha-1$)
and hold only up to $q$ entries: the same slots within a percent, more of them deep; at the floor $\alpha'_{\min}$, a shape scored safe survives
  continuous churn.
\end{enumerate}

Table~\ref{tab:churn} compares them; Appendix~\ref{app:churn} has
the model, the measurements and the choice among the three.

\begin{center}
\footnotesize
\begin{tabular}{lllll}
option & stability & extra space & extra work & suits\\
\hline
hoist on overflow & lost during a pass & none & one pass per overflow & light or batch churn\\
grow as the tail fills & kept & one level per growth & allocation & medium to high churn\\
load-factor overload & kept & same slots within $1\%$ & $\approx 7\%$ more hit probes & continuous churn\\
\end{tabular}
\captionof{table}{The three churn options.}\label{tab:churn}
\end{center}

\section{Conclusion}\label{sec:conclusion}

Multitable makes the multicollision distribution a sizing tool rather
than an obstacle: binomial tails set each level's bucket count, and
the cascade absorbs the variance that forces other designs to
overprovision. The result is a table whose slot load factor is a
parameter reaching one, with metadata costing a few parts in ten
thousand, whose capacity is chosen freely rather than rounded to a
power of two, and whose resizing and eviction are local operations.
The same tail arithmetic yields a multicollision bound that stays
informative where the classical subset bound is vacuous. Against
hashbrown, the Rust port of SwissTable behind Rust's standard
\texttt{HashMap} and the strongest baseline we know of, the filtered
multitable delivers on average about twice the throughput at equal
memory, $2.1\times$ and $1.9\times$ in arithmetic mean over $84$
configurations depending on how hashbrown is keyed, about three times
on negative lookups and $1.3$ to $1.6\times$ on inserts and positive
lookups. Hash tables are among the most studied and most tuned data
structures, so a factor of two over the table behind Rust's standard
library is a large improvement, measured on one machine against one
baseline. One seam
remains: the analysis assumes a random oracle, with independent hash
values at every level, while the implementation derives all levels
from a single hash per lookup, a shortcut that is compensated and
verified empirically rather than proved. A formal treatment of that
reuse is the natural next step.

\clearpage

\bibliographystyle{abbrv}
\bibliography{refs}

\appendix

\section{Benchmark charts}\label{app:benchmarks}

This appendix collects the remaining charts of the Apple M2 Pro
campaign of Section~\ref{sec:performance}: the $16{+}16$-byte map
and the two set shapes, then the load-factor sweep and the sizing
variants. All of them run at overfill factor
\hyperref[par:overfill]{$r=1.01$}. The three shape figures follow
the equal-memory protocol of Figure~\ref{fig:bench-kv4}, the
multitable and hashbrown sized to the same footprint in bytes at
every point; the sweep varies the multitable's footprint against a
fixed hashbrown reference, and the variants compare multitable
configurations at equal key counts. Each caption states the
configuration its chart needs. Speedups quoted in the text below, such as $1.87\times$ and $1.73\times$, are arithmetic means of hashbrown time over multitable time, against the byte-keyed and the integer-keyed arm in that order.

\begin{figure}[htbp]
\centering
\definecolor{fkOIblue}{HTML}{0072B2}
\definecolor{fkOIverm}{HTML}{D55E00}
\definecolor{fkOIgreen}{HTML}{009E73}
\definecolor{fkOIpink}{HTML}{CC79A7}
\begin{tikzpicture}[baseline=(current bounding box.north)]
\begin{groupplot}[
    group style={group size=2 by 2,
        horizontal sep=0.95cm, vertical sep=0.9cm,
        xlabels at=edge bottom, xticklabels at=edge bottom,
        ylabels at=edge left},
    width=0.5\textwidth,
    height=4.4cm,
    xmode=log,
    xmin=2500, xmax=3.5e8,
    ymin=0,
    xtick={1e4,1e5,1e6,1e7,1e8},
    xlabel={$q$ (key-value pairs)},
    ylabel={Mops/s},
    title style={font=\footnotesize, yshift=-2pt},
    tick label style={font=\footnotesize},
    label style={font=\footnotesize},
    axis lines=left,
]
% ------------------------------------------------------------- insert
\nextgroupplot[title={insert}, ymax=415]
\addplot[fkOIblue, thick, mark=*, mark size=1.1pt,
    mark options={fill=fkOIblue, draw=fkOIblue}] coordinates {
(3548,379.550) (14193,387.597) (56771,307.465) (227082,237.648)
(908329,185.123) (3633316,117.295) (116266107,53.304)
};
\addplot[fkOIgreen, semithick, mark=triangle*, mark size=1.4pt,
    mark options={fill=fkOIgreen, draw=fkOIgreen}] coordinates {
(3548,276.778) (14193,211.104) (56771,173.807) (227082,145.127)
(908329,81.728) (3633316,39.681) (116266107,28.322)
};
\addplot[fkOIverm, semithick, mark=square*, mark size=1pt,
    mark options={fill=fkOIverm, draw=fkOIverm}] coordinates {
(3548,261.650) (14193,240.050) (56771,190.883) (227082,168.722)
(908329,105.126) (3633316,79.748) (116266107,38.785)
};
\addplot[fkOIpink, semithick, densely dashed, mark=square, mark size=1.2pt,
    mark options={draw=fkOIpink, fill=none, solid}] coordinates {
(3548,249.963) (14193,228.071) (56771,176.473) (227082,159.144)
(908329,105.930) (3633316,78.790) (116266107,37.690)
};
% ---------------------------------------------------- positive lookup
\nextgroupplot[title={positive lookup}, ymax=615]
\addplot[fkOIblue, thick, mark=*, mark size=1.1pt,
    mark options={fill=fkOIblue, draw=fkOIblue}] coordinates {
(3548,581.835) (14193,566.861) (56771,362.017) (227082,263.908)
(908329,153.607) (3633316,97.994) (116266107,58.390)
};
\addplot[fkOIgreen, semithick, mark=triangle*, mark size=1.4pt,
    mark options={fill=fkOIgreen, draw=fkOIgreen}] coordinates {
(3548,228.770) (14193,177.088) (56771,153.822) (227082,142.765)
(908329,59.798) (3633316,33.909) (116266107,24.919)
};
\addplot[fkOIverm, semithick, mark=square*, mark size=1pt,
    mark options={fill=fkOIverm, draw=fkOIverm}] coordinates {
(3548,464.339) (14193,372.287) (56771,269.251) (227082,230.012)
(908329,119.626) (3633316,87.259) (116266107,46.022)
};
\addplot[fkOIpink, semithick, densely dashed, mark=square, mark size=1.2pt,
    mark options={draw=fkOIpink, fill=none, solid}] coordinates {
(3548,400.336) (14193,339.489) (56771,238.914) (227082,203.314)
(908329,102.708) (3633316,74.717) (116266107,40.557)
};
% ---------------------------------------------------- negative lookup
\nextgroupplot[
    title={negative lookup},
    ymax=295,
    legend style={font=\scriptsize, at={(1,-0.35)}, anchor=north,
        draw=none, fill=none,
        /tikz/every even column/.append style={column sep=0.5cm}},
    legend columns=2,
    legend cell align=left,
]
\addplot[fkOIblue, thick, mark=*, mark size=1.1pt,
    mark options={fill=fkOIblue, draw=fkOIblue}] coordinates {
(3548,277.339) (14193,259.767) (56771,253.575) (227082,236.395)
(908329,188.182) (3633316,119.262) (116266107,48.137)
};
\addlegendentry{multitable filtered, $s{=}16$ (byte keys)}
\addplot[fkOIgreen, semithick, mark=triangle*, mark size=1.4pt,
    mark options={fill=fkOIgreen, draw=fkOIgreen}] coordinates {
(3548,141.927) (14193,134.519) (56771,120.011) (227082,103.700)
(908329,38.984) (3633316,20.296) (116266107,10.694)
};
\addlegendentry{multitable plain, $s{=}8$ (byte keys)}
\addplot[fkOIverm, semithick, mark=square*, mark size=1pt,
    mark options={fill=fkOIverm, draw=fkOIverm}] coordinates {
(3548,230.017) (14193,82.898) (56771,73.060) (227082,63.855)
(908329,60.051) (3633316,58.304) (116266107,21.477)
};
\addlegendentry{hashbrown (\texttt{u128} keys)}
\addplot[fkOIpink, semithick, densely dashed, mark=square, mark size=1.2pt,
    mark options={draw=fkOIpink, fill=none, solid}] coordinates {
(3548,147.427) (14193,79.428) (56771,67.843) (227082,59.744)
(908329,55.739) (3633316,54.313) (116266107,19.807)
};
\addlegendentry{hashbrown (byte keys)}
% ------------------------------------------------------- mixed lookup
\nextgroupplot[title={mixed lookup (50\% hits)}, ymax=560]
\addplot[fkOIblue, thick, mark=*, mark size=1.1pt,
    mark options={fill=fkOIblue, draw=fkOIblue}] coordinates {
(3548,528.067) (14193,169.898) (56771,114.672) (227082,101.693)
(908329,87.583) (3633316,69.914) (116266107,45.857)
};
\addplot[fkOIgreen, semithick, mark=triangle*, mark size=1.4pt,
    mark options={fill=fkOIgreen, draw=fkOIgreen}] coordinates {
(3548,118.603) (14193,71.465) (56771,65.022) (227082,59.943)
(908329,36.148) (3633316,21.188) (116266107,12.156)
};
\addplot[fkOIverm, semithick, mark=square*, mark size=1pt,
    mark options={fill=fkOIverm, draw=fkOIverm}] coordinates {
(3548,435.711) (14193,110.374) (56771,81.655) (227082,72.200)
(908329,63.298) (3633316,56.386) (116266107,28.454)
};
\addplot[fkOIpink, semithick, densely dashed, mark=square, mark size=1.2pt,
    mark options={draw=fkOIpink, fill=none, solid}] coordinates {
(3548,412.082) (14193,99.405) (56771,76.165) (227082,68.565)
(908329,62.364) (3633316,55.994) (116266107,26.662)
};
\end{groupplot}
\end{tikzpicture}
\caption{The $16{+}16$-byte map shape, as in
Figure~\ref{fig:bench-kv4}: the filtered multitable at $s=16$ and the
plain one at $s=8$ (\hyperref[par:overfill]{$r=1.01$}) against hashbrown on the same raw
$16$-byte arrays and on native
\texttt{u128} integers, at physical load factor
$\approx0.840$. This shape's ladder tops out at $1.16\cdot10^{8}$
keys.}
\label{fig:bench-kv16}
\end{figure}
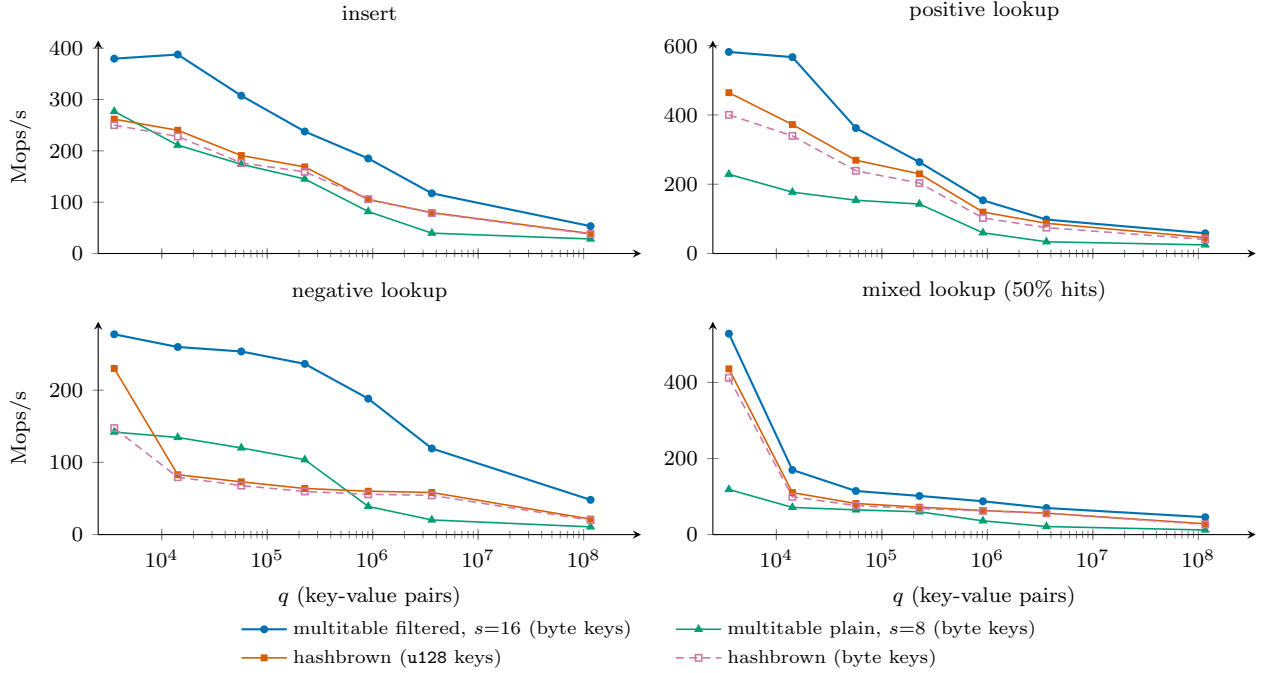

Figure~\ref{fig:bench-kv16} widens key and value to $16$ bytes each.
The plain multitable's only wins over hashbrown sit at the smallest table
on insert and across a middle stretch of negative lookups, from
about $1.4\times10^4$ to $2.3\times10^5$ keys; every larger size, and
every other operation, favors hashbrown instead. The filtered
multitable's margin narrows to $1.87\times$ against the byte-keyed arm and
$1.73\times$ against the integer-keyed one over the shape's $28$
cells, mixed stream included, and this is the shape on which the
plain multitable suffers most: with $16$ bytes to compare per slot, it
takes only $4$ of the $28$ cells from either hashbrown arm.

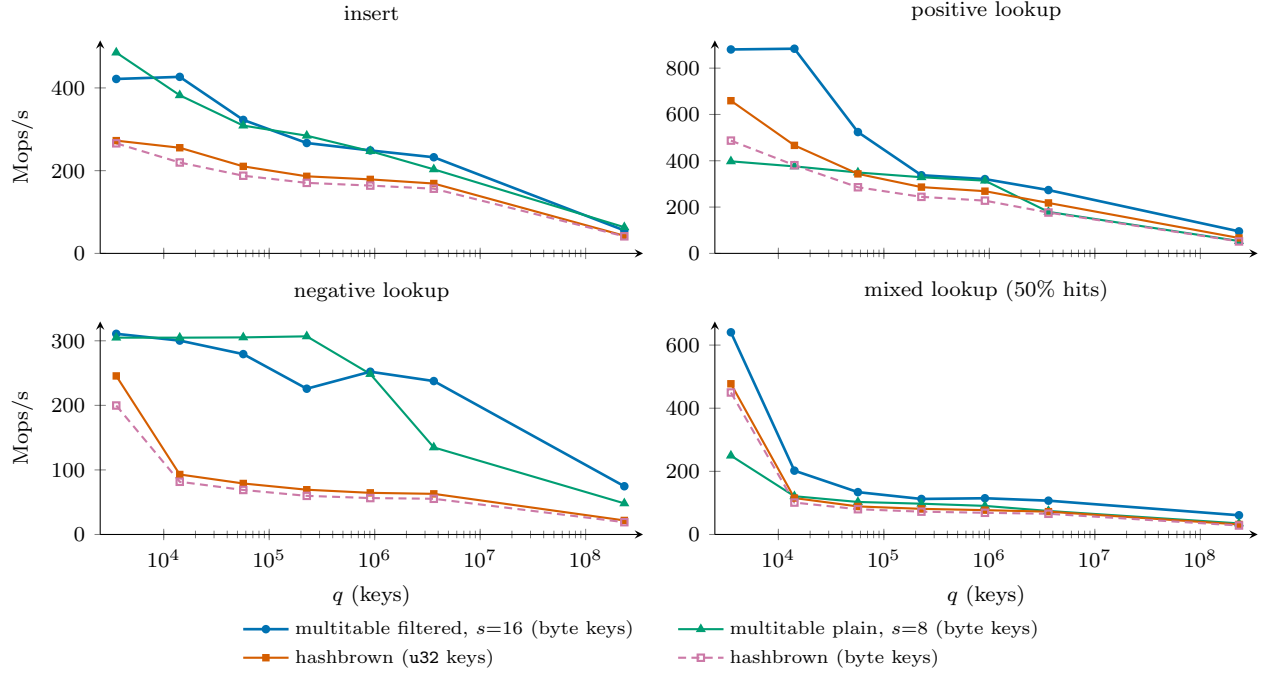
\begin{figure}[htbp]
\centering
\definecolor{fsOIblue}{HTML}{0072B2}
\definecolor{fsOIverm}{HTML}{D55E00}
\definecolor{fsOIgreen}{HTML}{009E73}
\definecolor{fsOIpink}{HTML}{CC79A7}
\begin{tikzpicture}[baseline=(current bounding box.north)]
\begin{groupplot}[
    group style={group size=2 by 2,
        horizontal sep=0.95cm, vertical sep=0.9cm,
        xlabels at=edge bottom, xticklabels at=edge bottom,
        ylabels at=edge left},
    width=0.5\textwidth,
    height=4.4cm,
    xmode=log,
    xmin=2500, xmax=3.5e8,
    ymin=0,
    xtick={1e4,1e5,1e6,1e7,1e8},
    xlabel={$q$ (keys)},
    ylabel={Mops/s},
    title style={font=\footnotesize, yshift=-2pt},
    tick label style={font=\footnotesize},
    label style={font=\footnotesize},
    axis lines=left,
]
% ------------------------------------------------------------- insert
\nextgroupplot[title={insert}, ymax=515]
\addplot[fsOIblue, line width=1pt, mark=*, mark size=1.1pt,
    mark options={fill=fsOIblue, draw=fsOIblue}] coordinates {
(3550,421.638) (14194,426.730) (56772,322.674) (227084,266.837)
(908330,248.899) (3633317,232.396) (232532215,55.531)
};
\addplot[fsOIgreen, thick, mark=triangle*, mark size=1.4pt,
    mark options={fill=fsOIgreen, draw=fsOIgreen}] coordinates {
(3550,485.460) (14194,382.234) (56772,308.966) (227084,284.511)
(908330,247.231) (3633317,203.219) (232532215,63.281)
};
\addplot[fsOIverm, thick, mark=square*, mark size=1pt,
    mark options={fill=fsOIverm, draw=fsOIverm}] coordinates {
(3550,272.576) (14194,255.376) (56772,210.358) (227084,186.317)
(908330,178.897) (3633317,168.890) (232532215,42.657)
};
\addplot[fsOIpink, thick, densely dashed, mark=square, mark size=1.2pt,
    mark options={draw=fsOIpink, fill=none, solid}] coordinates {
(3550,266.078) (14194,219.727) (56772,187.758) (227084,170.509)
(908330,163.905) (3633317,156.392) (232532215,41.263)
};
% ---------------------------------------------------- positive lookup
\nextgroupplot[title={positive lookup}, ymax=920]
\addplot[fsOIblue, line width=1pt, mark=*, mark size=1.1pt,
    mark options={fill=fsOIblue, draw=fsOIblue}] coordinates {
(3550,880.669) (14194,883.861) (56772,523.917) (227084,337.564)
(908330,320.842) (3633317,273.493) (232532215,95.152)
};
\addplot[fsOIgreen, thick, mark=triangle*, mark size=1.4pt,
    mark options={fill=fsOIgreen, draw=fsOIgreen}] coordinates {
(3550,397.519) (14194,375.728) (56772,349.369) (227084,328.872)
(908330,314.130) (3633317,179.308) (232532215,53.364)
};
\addplot[fsOIverm, thick, mark=square*, mark size=1pt,
    mark options={fill=fsOIverm, draw=fsOIverm}] coordinates {
(3550,659.370) (14194,466.483) (56772,343.430) (227084,286.229)
(908330,268.637) (3633317,217.851) (232532215,66.531)
};
\addplot[fsOIpink, thick, densely dashed, mark=square, mark size=1.2pt,
    mark options={draw=fsOIpink, fill=none, solid}] coordinates {
(3550,487.045) (14194,380.909) (56772,285.600) (227084,244.194)
(908330,227.754) (3633317,176.410) (232532215,51.637)
};
% ---------------------------------------------------- negative lookup
\nextgroupplot[
    title={negative lookup},
    ymax=330,
    legend style={font=\scriptsize, at={(1,-0.35)}, anchor=north,
        draw=none, fill=none,
        /tikz/every even column/.append style={column sep=0.5cm}},
    legend columns=2,
    legend cell align=left,
]
\addplot[fsOIblue, line width=1pt, mark=*, mark size=1.1pt,
    mark options={fill=fsOIblue, draw=fsOIblue}] coordinates {
(3550,310.800) (14194,300.300) (56772,279.376) (227084,225.774)
(908330,252.099) (3633317,237.626) (232532215,74.699)
};
\addlegendentry{multitable filtered, $s{=}16$ (byte keys)}
\addplot[fsOIgreen, thick, mark=triangle*, mark size=1.4pt,
    mark options={fill=fsOIgreen, draw=fsOIgreen}] coordinates {
(3550,304.878) (14194,304.832) (56772,305.288) (227084,306.909)
(908330,248.583) (3633317,134.882) (232532215,48.190)
};
\addlegendentry{multitable plain, $s{=}8$ (byte keys)}
\addplot[fsOIverm, thick, mark=square*, mark size=1pt,
    mark options={fill=fsOIverm, draw=fsOIverm}] coordinates {
(3550,245.592) (14194,92.917) (56772,79.003) (227084,69.330)
(908330,64.498) (3633317,62.940) (232532215,21.718)
};
\addlegendentry{hashbrown (\texttt{u32} keys)}
\addplot[fsOIpink, thick, densely dashed, mark=square, mark size=1.2pt,
    mark options={draw=fsOIpink, fill=none, solid}] coordinates {
(3550,199.764) (14194,81.641) (56772,68.961) (227084,59.828)
(908330,56.415) (3633317,55.299) (232532215,19.049)
};
\addlegendentry{hashbrown (byte keys)}
% ------------------------------------------------------- mixed lookup
\nextgroupplot[title={mixed lookup (50\% hits)}, ymax=675]
\addplot[fsOIblue, line width=1pt, mark=*, mark size=1.1pt,
    mark options={fill=fsOIblue, draw=fsOIblue}] coordinates {
(3550,640.533) (14194,202.163) (56772,134.124) (227084,112.507)
(908330,114.772) (3633317,107.077) (232532215,60.880)
};
\addplot[fsOIgreen, thick, mark=triangle*, mark size=1.4pt,
    mark options={fill=fsOIgreen, draw=fsOIgreen}] coordinates {
(3550,249.700) (14194,121.161) (56772,102.993) (227084,97.482)
(908330,90.544) (3633317,74.440) (232532215,35.092)
};
\addplot[fsOIverm, thick, mark=square*, mark size=1pt,
    mark options={fill=fsOIverm, draw=fsOIverm}] coordinates {
(3550,477.418) (14194,115.389) (56772,88.812) (227084,81.197)
(908330,77.016) (3633317,72.265) (232532215,31.734)
};
\addplot[fsOIpink, thick, densely dashed, mark=square, mark size=1.2pt,
    mark options={draw=fsOIpink, fill=none, solid}] coordinates {
(3550,449.802) (14194,101.068) (56772,80.003) (227084,72.283)
(908330,68.703) (3633317,65.528) (232532215,28.396)
};
\end{groupplot}
\end{tikzpicture}
\caption{Set workload, $4$-byte keys and no values, as in
Figure~\ref{fig:bench-kv4}: the filtered ($s=16$) and plain ($s=8$)
multitables against hashbrown on the same raw key bytes and on the native \texttt{u32}, at physical load factor
$\approx0.693$,
\hyperref[par:overfill]{$r=1.01$}.}
\label{fig:bench-set4}
\end{figure}

Figure~\ref{fig:bench-set4} shows the multitable's best shape. On
misses the plain multitable tracks the filtered one closely up to about
$9\times10^5$ keys, ahead of it at three of those five sizes,
by under $10\%$ at two and by $36\%$ at $2.3\times10^5$ keys, before
falling well behind once the table passes
$3\times10^6$ keys. With no
values to fetch, the filtered multitable averages $2.18\times$ and
$1.91\times$ over its $28$ cells, and the plain multitable beats
hashbrown in $25$ of $28$ cells against the
byte-keyed arm and $23$ against the integer-keyed one: $1.60\times$
and $1.48\times$ on insert, $3.46\times$ and $3.00\times$ on negative
lookups, $1.15\times$ and $1.03\times$ on the mixed stream, and
$1.11\times$ and $0.91\times$ on positive lookups, the one operation
on which it falls behind the integer-keyed arm.

\begin{figure}[htbp]
\centering
\definecolor{fsOIblue}{HTML}{0072B2}
\definecolor{fsOIverm}{HTML}{D55E00}
\definecolor{fsOIgreen}{HTML}{009E73}
\definecolor{fsOIpink}{HTML}{CC79A7}
\begin{tikzpicture}[baseline=(current bounding box.north)]
\begin{groupplot}[
    group style={group size=2 by 2,
        horizontal sep=0.95cm, vertical sep=0.9cm,
        xlabels at=edge bottom, xticklabels at=edge bottom,
        ylabels at=edge left},
    width=0.5\textwidth,
    height=4.4cm,
    xmode=log,
    xmin=2500, xmax=3.5e8,
    ymin=0,
    xtick={1e4,1e5,1e6,1e7,1e8},
    xlabel={$q$ (keys)},
    ylabel={Mops/s},
    title style={font=\footnotesize, yshift=-2pt},
    tick label style={font=\footnotesize},
    label style={font=\footnotesize},
    axis lines=left,
]
% ------------------------------------------------------------- insert
\nextgroupplot[title={insert}, ymax=435]
\addplot[fsOIblue, thick, mark=*, mark size=1.2pt,
    mark options={fill=fsOIblue, draw=fsOIblue}] coordinates {
(3549,392.265) (14193,403.730) (56771,294.672) (227083,248.398)
(908329,230.579) (3633316,182.113) (232532214,53.421)
};
\addplot[fsOIgreen, semithick, mark=triangle*, mark size=1.4pt,
    mark options={fill=fsOIgreen, draw=fsOIgreen}] coordinates {
(3549,303.961) (14193,253.646) (56771,199.573) (227083,178.021)
(908329,139.128) (3633316,53.310) (232532214,29.398)
};
\addplot[fsOIverm, semithick, mark=square*, mark size=1pt,
    mark options={fill=fsOIverm, draw=fsOIverm}] coordinates {
(3549,271.334) (14193,259.041) (56771,198.051) (227083,177.952)
(908329,171.019) (3633316,127.943) (232532214,41.934)
};
\addplot[fsOIpink, semithick, densely dashed, mark=square, mark size=1.2pt,
    mark options={draw=fsOIpink, fill=none, solid}] coordinates {
(3549,265.343) (14193,223.894) (56771,184.509) (227083,167.367)
(908329,160.870) (3633316,121.564) (232532214,40.578)
};
% ---------------------------------------------------- positive lookup
\nextgroupplot[title={positive lookup}, ymax=695]
\addplot[fsOIblue, thick, mark=*, mark size=1.2pt,
    mark options={fill=fsOIblue, draw=fsOIblue}] coordinates {
(3549,671.682) (14193,648.466) (56771,338.949) (227083,277.955)
(908329,243.002) (3633316,118.828) (232532214,65.537)
};
\addplot[fsOIgreen, semithick, mark=triangle*, mark size=1.4pt,
    mark options={fill=fsOIgreen, draw=fsOIgreen}] coordinates {
(3549,289.385) (14193,213.858) (56771,174.712) (227083,164.804)
(908329,119.811) (3633316,40.390) (232532214,26.727)
};
\addplot[fsOIverm, semithick, mark=square*, mark size=1pt,
    mark options={fill=fsOIverm, draw=fsOIverm}] coordinates {
(3549,494.633) (14193,383.289) (56771,290.655) (227083,251.370)
(908329,211.180) (3633316,110.395) (232532214,54.868)
};
\addplot[fsOIpink, semithick, densely dashed, mark=square, mark size=1.2pt,
    mark options={draw=fsOIpink, fill=none, solid}] coordinates {
(3549,477.281) (14193,358.603) (56771,264.494) (227083,226.270)
(908329,186.261) (3633316,91.496) (232532214,46.294)
};
% ---------------------------------------------------- negative lookup
\nextgroupplot[
    title={negative lookup},
    ymax=305,
    legend style={font=\scriptsize, at={(1,-0.35)}, anchor=north,
        draw=none, fill=none,
        /tikz/every even column/.append style={column sep=0.5cm},
        /pgfplots/legend image code/.code={
            \draw[mark repeat=2, mark phase=2, ##1] plot coordinates {
                (0cm,0cm) (0.16cm,0cm) (0.32cm,0cm)};}},
    legend columns=2,
    legend cell align=left,
]
\addplot[fsOIblue, thick, mark=*, mark size=1.2pt,
    mark options={fill=fsOIblue, draw=fsOIblue}] coordinates {
(3549,281.357) (14193,267.659) (56771,256.674) (227083,239.653)
(908329,224.200) (3633316,135.133) (232532214,45.541)
};
\addlegendentry{multitable filtered, $s{=}16$ (byte keys)}
\addplot[fsOIgreen, semithick, mark=triangle*, mark size=1.4pt,
    mark options={fill=fsOIgreen, draw=fsOIgreen}] coordinates {
(3549,180.226) (14193,142.869) (56771,132.405) (227083,120.106)
(908329,75.681) (3633316,24.661) (232532214,11.653)
};
\addlegendentry{multitable plain, $s{=}8$ (byte keys)}
\addplot[fsOIverm, semithick, mark=square*, mark size=1pt,
    mark options={fill=fsOIverm, draw=fsOIverm}] coordinates {
(3549,232.845) (14193,88.485) (56771,73.584) (227083,63.964)
(908329,60.306) (3633316,60.059) (232532214,20.105)
};
\addlegendentry{hashbrown (\texttt{u128} keys)}
\addplot[fsOIpink, semithick, densely dashed, mark=square, mark size=1.2pt,
    mark options={draw=fsOIpink, fill=none, solid}] coordinates {
(3549,173.199) (14193,80.323) (56771,68.511) (227083,59.980)
(908329,56.608) (3633316,55.716) (232532214,20.516)
};
\addlegendentry{hashbrown (byte keys)}
% ------------------------------------------------------- mixed lookup
\nextgroupplot[title={mixed lookup (50\% hits)}, ymax=595]
\addplot[fsOIblue, thick, mark=*, mark size=1.2pt,
    mark options={fill=fsOIblue, draw=fsOIblue}] coordinates {
(3549,566.476) (14193,145.677) (56771,110.557) (227083,102.442)
(908329,97.984) (3633316,75.950) (232532214,46.734)
};
\addplot[fsOIgreen, semithick, mark=triangle*, mark size=1.4pt,
    mark options={fill=fsOIgreen, draw=fsOIgreen}] coordinates {
(3549,170.198) (14193,75.956) (56771,66.943) (227083,64.266)
(908329,54.391) (3633316,26.874) (232532214,13.294)
};
\addplot[fsOIverm, semithick, mark=square*, mark size=1pt,
    mark options={fill=fsOIverm, draw=fsOIverm}] coordinates {
(3549,474.518) (14193,113.390) (56771,83.711) (227083,75.062)
(908329,70.190) (3633316,62.585) (232532214,29.375)
};
\addplot[fsOIpink, semithick, densely dashed, mark=square, mark size=1.2pt,
    mark options={draw=fsOIpink, fill=none, solid}] coordinates {
(3549,448.873) (14193,103.948) (56771,75.670) (227083,68.922)
(908329,65.745) (3633316,57.319) (232532214,27.323)
};
\end{groupplot}
\end{tikzpicture}
\caption{Set workload, $16$-byte keys, otherwise as in
Figure~\ref{fig:bench-set4}: the filtered ($s=16$) and plain ($s=8$)
multitables against hashbrown on the same raw key bytes and on the native
\texttt{u128}, at physical load factor $\approx0.815$.}
\label{fig:bench-set16}
\end{figure}
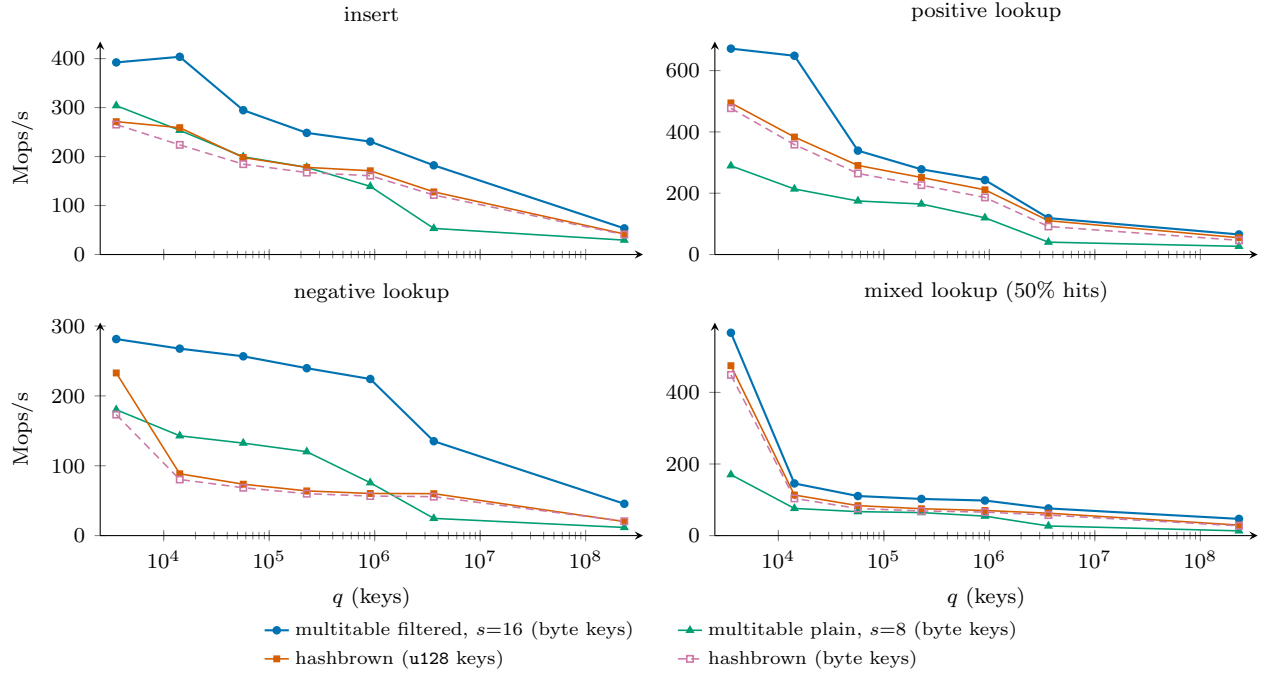

The $16$-byte set (Figure~\ref{fig:bench-set16}) keeps the keys
without values and lands beside the $16{+}16$ map. On insert the
plain multitable leads the byte-keyed arm at the four smallest sizes and
is level with the integer-keyed one at three of them, $12\%$ ahead
at the smallest, then falls behind both from $9\times10^5$ keys on.
The filtered
multitable leads by $1.85\times$ and $1.71\times$ over its $28$ cells,
while the plain multitable takes $9$ and $6$ of them (a seventh insert
cell against the integer-keyed arm is tied within $0.1\%$), a few
more than on the map but still a small minority, so the filter byte pays for
itself at $16$-byte keys whether or not values are stored.

\begin{figure}[htbp]
\centering
\input{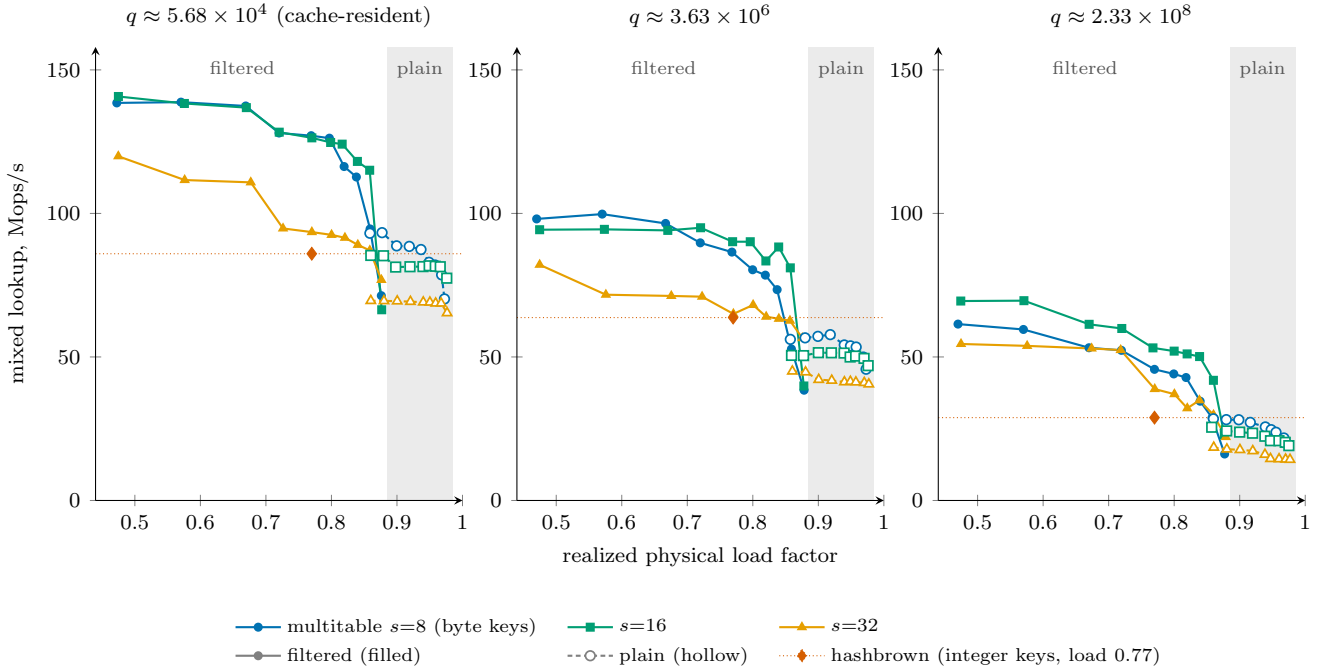}
\caption{Speed-versus-size dial for the $4{+}4$-byte map on an Apple
M2 Pro at \hyperref[par:overfill]{$r=1.01$}: mixed-lookup throughput
against realized physical load factor, the three panels sharing one
throughput scale. The shaded band marks loads only the plain multitable
reaches, and the two families overlap near $0.86$. In the right
panel the plain $s=8$ line ends at requested $0.97$.}
\label{fig:bench-lfsweep}
\end{figure}

Figure~\ref{fig:bench-lfsweep} turns load factor into a dial:
mixed-lookup throughput against \emph{realized} physical load
factor at bucket widths $s=8,16,32$ against integer-keyed hashbrown's
sizing. At hashbrown's memory (load $0.77$) the filtered multitable runs
$1.36$-$1.84\times$ its throughput at $s=8,16$ and $1.02$-$1.35\times$
at $s=32$, ahead in every cell; at hashbrown's speed
those arms cross the reference at realized $0.85$-$0.87$,
$9.1$-$11.6\%$ fewer bytes for the same throughput, $s=32$ crossing
earlier at the middle size, at realized $0.83$. Only plain $s=8$ in
cache, with no filter bytes at all, holds the reference past that,
through realized $0.937$ at $18\%$ fewer bytes; the rest trail by tier,
$0.76$-$0.90$ of hashbrown's throughput in cache and $0.49$-$0.75$ in
RAM at realized $0.97$ ($21\%$ fewer bytes for the same keys). That
stretch stays cheap, $6.1$-$25.2\%$ of throughput given up from
realized $0.86$ to $0.977$ with no cliff as load approaches $1$: a
plain arm enters it $3$-$6$ levels deep and leaves at $5$-$11$, the
$s=8$ arms deepening most and, in cache and at the middle size,
losing most.

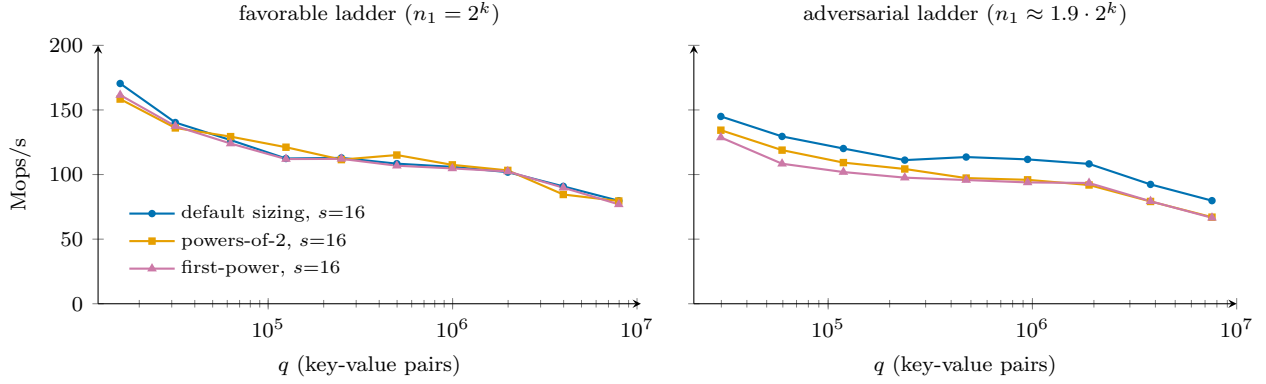
\begin{figure}[htbp]
\centering
\definecolor{fvOIblue}{HTML}{0072B2}
\definecolor{fvOIorange}{HTML}{E69F00}
\definecolor{fvOIpurp}{HTML}{CC79A7}
\begin{tikzpicture}[baseline=(current bounding box.north)]
\begin{groupplot}[
    group style={group size=2 by 1,
        horizontal sep=0.7cm,
        ylabels at=edge left, yticklabels at=edge left},
    width=0.5\textwidth,
    height=5.0cm,
    xmode=log,
    ymin=0, ymax=200,
    xtick={1e4,1e5,1e6,1e7},
    ytick={0,50,100,150,200},
    xlabel={$q$ (key-value pairs)},
    ylabel={Mops/s},
    title style={font=\footnotesize, yshift=-2pt},
    tick label style={font=\footnotesize},
    label style={font=\footnotesize},
    axis lines=left,
]
% -------------------------------------------------- favorable ladder
\nextgroupplot[
    title={favorable ladder ($n_1 = 2^k$)},
    xmin=11960, xmax=10500000,
    legend style={font=\scriptsize, at={(0.04,0.06)}, anchor=south west,
        draw=none, fill=none},
    legend cell align=left,
]
\addplot[fvOIblue, thick, mark=*, mark size=1pt,
    mark options={fill=fvOIblue, draw=fvOIblue}] coordinates {
(15758,170.470) (31376,140.270) (62571,126.680) (124907,112.420)
(249494,112.900) (498556,108.410) (996519,105.960) (1992220,101.850)
(3983297,90.853) (7964997,79.710)
};
\addlegendentry{default sizing, $s{=}16$}
\addplot[fvOIorange, thick, mark=square*, mark size=1pt,
    mark options={fill=fvOIorange, draw=fvOIorange}] coordinates {
(15758,158.420) (31376,136.080) (62571,129.350) (124907,121.150)
(249494,111.550) (498556,115.100) (996519,107.500) (1992220,103.210)
(3983297,84.561) (7964997,79.525)
};
\addlegendentry{powers-of-2, $s{=}16$}
\addplot[fvOIpurp, thick, mark=triangle*, mark size=1.3pt,
    mark options={fill=fvOIpurp, draw=fvOIpurp}] coordinates {
(15758,161.540) (31376,137.510) (62571,124.010) (124907,111.990)
(249494,112.140) (498556,106.850) (996519,104.750) (1992220,102.410)
(3983297,89.840) (7964997,76.931)
};
\addlegendentry{first-power, $s{=}16$}
% ------------------------------------------------ adversarial ladder
\nextgroupplot[
    title={adversarial ladder ($n_1 \approx 1.9 \cdot 2^k$)},
    xmin=22000, xmax=10000000,
]
\addplot[fvOIblue, thick, mark=*, mark size=1pt,
    mark options={fill=fvOIblue, draw=fvOIblue}] coordinates {
(29836,144.970) (59482,129.500) (118729,120.150) (237162,111.170)
(473902,113.540) (947227,111.730) (1893654,108.250) (3786192,92.365)
(7570826,79.791)
};
\addplot[fvOIorange, thick, mark=square*, mark size=1pt,
    mark options={fill=fvOIorange, draw=fvOIorange}] coordinates {
(29836,134.330) (59482,118.900) (118729,109.250) (237162,104.310)
(473902,97.265) (947227,95.912) (1893654,91.913) (3786192,79.173)
(7570826,67.023)
};
\addplot[fvOIpurp, thick, mark=triangle*, mark size=1.3pt,
    mark options={fill=fvOIpurp, draw=fvOIpurp}] coordinates {
(29836,128.710) (59482,108.370) (118729,101.940) (237162,97.603)
(473902,95.755) (947227,93.907) (1893654,93.468) (3786192,79.338)
(7570826,66.535)
};
\end{groupplot}
\end{tikzpicture}
\caption{Sizing variants (Section~\ref{sec:implementation}) on an
Apple M2 Pro: mixed-lookup throughput of the filtered multitable at bucket
size $s=16$ on the $4{+}4$-byte map shape, at equal key counts and
physical load factor $\approx0.77$, \hyperref[par:overfill]{$r=1.01$}. Rounding to powers
of two leaves the first level untouched on the favorable ladder,
rounding only the deeper levels, and halves the first level on the
adversarial one. Multitable-only, with no hashbrown arm.}
\label{fig:bench-variants}
\end{figure}

Figure~\ref{fig:bench-variants} measures the addressing cost of
sizing variants (Section~\ref{sec:implementation}): default sizing
with Lemire's remap against rounding every level's bucket count to a
power of two, or only the first level. On a
favorable ladder (first level exactly $2^k$, so rounding touches
only the deeper levels) full rounding is a wash, $+0.6\%$ median
mixed-lookup throughput, positive at five of the ten sizes; on an
adversarial one, where rounding halves the first level and balloons
the second, it cedes $14.2\%$ median, negative at all nine.
First-power is no better, $-1.3\%$ favorable and $-15.2\%$
adversarial. Neither variant makes a case for avoiding the
multiply-shift at $s=16$: where the rounded counts are near what the
sizing would choose anyway the multiplication costs nothing
measurable, and where they are not, rounding costs $14\%$ of the
throughput. The default-sizing arm here is the filtered
configuration of Section~\ref{sec:performance}, the one that leads
both hashbrown arms in all $84$ cells, so the two charts share a
bucket size.

%\clearpage

\section{Asymptotics}\label{app:asymptotics}

For a fixed configuration the cascade's depth grows only
logarithmically in the number of keys, and a positive lookup
costs a constant expected number of probes. This appendix
states and proves both results.

\begin{proposition}[Logarithmic depth]\label{prop:depth}
Fix an integer $s\ge1$, $\alpha\in(0,1)$ and $\delta\in(0,1)$, and set
\begin{equation}\label{eq:depth-constants}
\beta:=\min\Bigl\{\tfrac12,\
  \tfrac{s}{8\bigl(2+\ln\frac{1}{1-\alpha}\bigr)}\Bigr\},
\qquad
T:=\max\Bigl\{\tfrac{2s}{\beta},\
   \tfrac{2s\ln(1/\delta)}{(1-\alpha)\,\beta}\Bigr\},
\qquad
c:=1-\tfrac{\alpha\beta}{2}\in(0,1).
\end{equation}
Then the construction of Definition~\ref{def:efficient} satisfies
$o_{i+1}\le c\,o_i$ whenever $o_i\ge T$, and hence, for
$q\ge T$,
\begin{equation}\label{eq:depth-bound}
L\;\le\;\frac{\ln(q/T)}{\ln(1/c)}+\frac{T}{s}+2
\;=\;O(\log q),
\end{equation}
with the constants depending only on $(s,\alpha,\delta)$.
\end{proposition}

The proof exhibits, for every level with
input $o\ge T$, a feasible bucket count $n^*=\lfloor\beta o/s\rfloor$;
maximality of $n_i$ then forces the overflow to contract by the factor
$c$, and levels with input below $T$ lose at least $s$ keys each.
Feasibility runs through the slack tilt $B_4$ in place of the
$\sqrt{o\ln(1/\delta)}$ deviation of $B_1$: evaluated at the fixed
tilt $\theta=2/s$, the $\ln(1/\delta)$ cost of $w^*$ is additive.
With $B_1$ the same argument would need that deviation,
$\sqrt{(o/2)\ln(1/\delta)}$, to fit in half the slack
$(1-\alpha)sn^*\ge(1-\alpha)\beta o/2$, that is,
$o\ge T_1:=8\ln(1/\delta)/\bigl((1-\alpha)^2\beta^2\bigr)$.
The dependence on $1-\alpha$ and on $\beta$ thus drops from
quadratic to linear, $T_1/T\approx22$ at
$(s,\alpha,\delta)=(8,0.9,2^{-10})$,
in exchange for a factor of $s$ that is a loss only for
$s>8/(1-\alpha)$.

\phantomsection\label{pf:depth}
\begin{proof}[Proof of Proposition~\ref{prop:depth}]
Let $o:=o_i\ge T$ and put $n^*:=\lfloor\beta o/s\rfloor$. Since
$o\ge2s/\beta$ we have $\beta o/s\ge2$, so $n^*\ge2$ and, as
$\lfloor x\rfloor\ge x/2$ for $x\ge2$, $n^*\ge\beta o/(2s)$.
Also $n^*\le\beta o/s$, so the mean
$\mu:=o/n^*$ satisfies $\mu\ge s/\beta\ge2s$.

\emph{Step 1: $n^*$ is feasible in \eqref{eq:efficient-ni}.}
Write $X\sim\Bin(o,1/n^*)$, so $\E X=\mu$, and set
$F:=F(s-1;o,1/n^*)=\PR[X\le s-1]$. With $\varepsilon:=1-s/\mu$ we have
$\varepsilon\ge1-\beta\ge\tfrac12$ and $\varepsilon<1$, so the
multiplicative Chernoff lower-tail bound
\cite[Thm.~4.5]{mitzenmacher2005probability} applies:
\begin{equation}\label{eq:depth-chernoff}
F\le\PR\bigl[X\le(1-\varepsilon)\mu\bigr]
\le e^{-\varepsilon^2\mu/2}
\le e^{-\mu/8}
\le e^{-s/(8\beta)}
\le (1-\alpha)e^{-2},
\end{equation}
where the last step is the definition of $\beta$. Next we evaluate
the slack quantile of \eqref{eq:chernoff-cert} at the tilt
$\theta:=2/s$. Pointwise
$e^{\theta\pos{s-X}}\le1+(e^{\theta s}-1)\ind{X\le s-1}$: on
$\{X\ge s\}$ both sides equal $1$, and on $\{X\le s-1\}$ the left
side is at most $e^{\theta s}$. Taking expectations and using
$\ln(1+x)\le x$,
\begin{equation}\label{eq:depth-wstar}
w^*(n^*)\;\le\;\frac{n^*\ln\bigl(1+(e^{2}-1)F\bigr)+\ln(1/\delta)}{2/s}
\;\le\;\Bigl(n^*(e^{2}-1)F+\ln(1/\delta)\Bigr)\frac{s}{2}.
\end{equation}
After the factor $s/2$, each parenthesized term contributes at most
half of the slack allowance $(1-\alpha)sn^*$. By
\eqref{eq:depth-chernoff},
$(e^{2}-1)F\le(1-e^{-2})(1-\alpha)\le1-\alpha$. By
$o\ge2s\ln(1/\delta)/((1-\alpha)\beta)$ and $n^*\ge\beta o/(2s)$,
\begin{equation}\label{eq:depth-dev}
n^*\;\ge\;\frac{\beta o}{2s}\;\ge\;\frac{\ln(1/\delta)}{1-\alpha},
\end{equation}
so $\ln(1/\delta)\le(1-\alpha)n^*$. Hence
$w^*(n^*)\le(1-\alpha)sn^*$, and with \eqref{eq:Bdef} and
\eqref{eq:B34},
\begin{equation}\label{eq:depth-feasible}
B(o,n^*,\delta)\;\le\;B_4(o,n^*,\delta)\;=\;o-s\,n^*+w^*(n^*)
\;\le\;o-\alpha\,s\,n^*,
\end{equation}
so $n^*$ satisfies the constraint in \eqref{eq:efficient-ni}.

\emph{Step 2: contraction.} Since $n_i$ is the maximal feasible
choice, $n_i\ge n^*$, and feasibility of $n_i$ gives
\begin{equation}\label{eq:depth-contract}
o_{i+1}\le B(o_i,n_i,\delta)\le o_i-\alpha s\,n_i
\le o_i-\alpha s\,n^*\le o_i-\tfrac{\alpha\beta}{2}\,o_i=c\,o_i .
\end{equation}

\emph{Step 3: counting.} The $o_i$ are nonincreasing, so the levels
with $o_i\ge T$ form a prefix; by \eqref{eq:depth-contract} its length
is at most $\ln(q/T)/\ln(1/c)+1$. Every level with $s<o_i<T$ has
$o_{i+1}\le o_i-s$ (the cap in \eqref{eq:Bdef}), so there are at most
$T/s$ of them, and one final level remains.
\end{proof}

\begin{remark}[Honest reading of the constants]\label{rem:depth-honest}
The bound \eqref{eq:depth-bound} is asymptotic in $q$ for \emph{fixed}
$(s,\alpha,\delta)$: the constants degrade as $\alpha\to1$, and the
proposition says nothing about the perfect table's $\alpha=1$ regime.
$T$ is very conservative: at $(8,0.9,2^{-10})$ it evaluates to about
$4772$, the bound \eqref{eq:depth-bound} at $q=10^5$ to $L\le626$,
and the cascade actually built at that configuration
(Remark~\ref{rem:free-tail}) has $14$ levels, the
inputs contracting by $0.19$--$0.27$ per level until
$o_i\approx10^3$. The single-bucket tail that follows is precisely
what Corollary~\ref{cor:shortcircuit} truncates.
\end{remark}

\begin{corollary}[Constant positive-lookup cost]\label{cor:probes}
Under the hypotheses of Proposition~\ref{prop:depth} with $q\ge T$,
condition on the event of Theorem~\ref{thm:efficient}(b). Then a key
chosen uniformly among the $q$ stored keys is found in $O(1)$ expected
probes, the constant depending only on $(s,\alpha,\delta)$.
\end{corollary}

\begin{proof}
Let $a_i$ be the number of keys reaching level $i$, so $a_1=q$,
$a_{L+1}=0$, and level $i$ stores $a_i-a_{i+1}$ keys at $i$ probes
each; by Abel summation the expected cost is
$\frac1q\sum_{i=1}^{L}i\,(a_i-a_{i+1})=\frac1q\sum_{i=1}^{L}a_i$. On
the conditioning event every level's overflow is at most the planned
$o_{i+1}$, so $a_i\le o_i$ for all $i$ by induction. As in the proof
of Proposition~\ref{prop:depth}, the levels with $o_i\ge T$ form a
prefix on which $o_i\le c^{\,i-1}q$, and at most $T/s+1$ levels
follow, each with $o_i<T$, so
\begin{equation}\label{eq:probes-bound}
\frac1q\sum_{i=1}^{L}a_i
\;\le\;\frac1q\sum_{i=1}^{L}o_i
\;\le\;\frac{1}{1-c}+\Bigl(\frac{T}{s}+1\Bigr)\frac{T}{q}
\;\le\;\frac{1}{1-c}+\frac{T}{s}+1.
\end{equation}
The last inequality uses $q\ge T$, and the resulting constant depends
only on $(s,\alpha,\delta)$ through \eqref{eq:depth-constants}.
\end{proof}

\section{An improved multicollision bound}\label{app:multicollision}

Let $E_s$ be the event that some bucket receives at least $s$ of $q$
keys hashed uniformly into $n$ buckets. Suzuki, Tonien, Kurosawa and
Toyota~\cite{suzuki2008birthday} disprove the
folklore belief that $q=n^{(s-1)/s}$ samples produce an $s$-collision
with constant probability: at that $q$ they show $\PR[E_s]\le1/s!$,
and a nearly matching lower bound $1/s!-1/(2(s!)^2)$ when $s$ is
small enough to be negligible against $n^{(s-1)/s}$. Hence roughly
$(s!)^{1/s}\,n^{(s-1)/s}$ samples are needed for even odds. The engine
behind bounds of this shape is the union bound over all $s$-subsets of
keys, $\PR[E_s]\le\binom{q}{s}n^{1-s}$. Our
Lemma~\ref{lem:counts} yields a bound that is never
worse and strictly dominates in the loaded regime.

\begin{proposition}[Improved multicollision bound]\label{prop:subset}
For all $1\le s\le q$ and $n\ge1$,
\begin{equation}\label{eq:subset-chain}
\PR[E_s]\;\le\; n\,\PR\!\left[\Bin\!\left(q,\tfrac1n\right)\ge s\right]
\;\le\;\binom{q}{s}\,n^{1-s}.
\end{equation}
\end{proposition}

\phantomsection\label{pf:subset}
\begin{proof}[Proof of Proposition~\ref{prop:subset}]
Let $N_s:=\sum_b\ind{X_b\ge s}$, with $X_b$ the bucket loads of
Proposition~\ref{prop:occupancy}, so $E_s=\{N_s\ge1\}$. Markov's inequality
and Lemma~\ref{lem:counts} give the first inequality:
$\PR[N_s\ge1]\le\E[N_s]=n\,\PR[\Bin(q,\tfrac1n)\ge s]$.
For the second, note the pointwise bound
$\ind{x\ge s}\le\binom{x}{s}$ for integers $x\ge0$ (the right side is
$0$ for $x<s$ and at least $1$ otherwise). Hence, with
$X\sim\Bin(q,\tfrac1n)$,
\begin{equation}\label{eq:factorial-moment}
\PR[X\ge s]\;\le\;\E\binom{X}{s}
=\sum_{x=s}^{q}\binom{x}{s}\binom{q}{x}n^{-x}\bigl(1-\tfrac1n\bigr)^{q-x}
=\binom{q}{s}n^{-s}\sum_{y=0}^{q-s}\binom{q-s}{y}n^{-y}\bigl(1-\tfrac1n\bigr)^{q-s-y}
=\binom{q}{s}n^{-s},
\end{equation}
using $\binom{x}{s}\binom{q}{x}=\binom{q}{s}\binom{q-s}{x-s}$ and the
binomial theorem. In words: $\binom{X}{s}$ counts the $s$-subsets of
keys that landed together in the bucket, each of the $\binom{q}{s}$
candidate subsets does so with probability $n^{-s}$, and
\eqref{eq:factorial-moment} is that count in expectation.
Multiplying by $n$ gives
$n\,\PR[X\ge s]\le\binom{q}{s}n^{1-s}$, the second inequality of
\eqref{eq:subset-chain}.
\end{proof}

\begin{remark}[When the improvement matters]\label{rem:subset-regimes}
The proof shows the subset bound is the
$s$-th factorial-moment relaxation \eqref{eq:factorial-moment} of our
binomial tail: the two agree to leading order as $n\to\infty$ with
$q,s$ fixed, so nothing is gained while $q\ll n$, and the gap opens
when the table is loaded. At $n=10$, $q=100$ the subset bound gives
$54$ for $s=20$ (vacuous) against our $0.020$, and $0.80$ against
$4.0\cdot10^{-4}$ at $s=24$, a factor of $2000$.
Figure~\ref{fig:suzuki} traces the gap across $q$ at $n=10$,
$s=20$. Both bounds
in \eqref{eq:subset-chain} are also expected counts of overloaded
buckets, so the comparison transfers to every use of
Lemma~\ref{lem:counts} in this paper.
\end{remark}

\begin{remark}[Poisson form]\label{rem:poisson-form}
For huge $q$ the binomial tail in \eqref{eq:subset-chain} is awkward
to evaluate directly. It can be replaced by the tail of
$\Pois(q/n)$ as an approximation: the two laws differ in total
variation by at most $\min(q/n^{2},\,1/n)$
\cite{lecam1960poisson,barbour1992poisson}.
\end{remark}

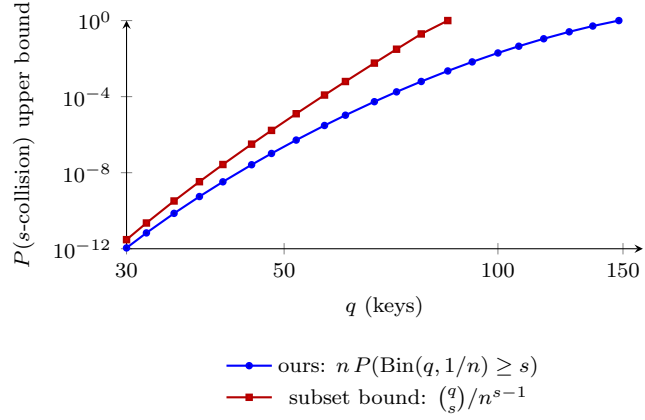
\begin{figure}[t]
\centering
\begin{tikzpicture}[baseline=(current bounding box.north)]
\begin{axis}[
    width=0.48\textwidth,
    height=4.6cm,
    xmode=log,
    ymode=log,
    % the q-range spans less than a decade, so default log ticks come out
    % as fractional powers (10^{1.6}, ...); fix readable integer ticks
    xtick={30,50,100,150},
    xticklabels={$30$,$50$,$100$,$150$},
    % explicit xmax: keeps the 150 tick inside the range and leaves room
    % so the last marker (q=148, blue) is not clipped at the edge
    xmax=160,
    % explicit ymin keeps the 1e-12 tick
    ymin=1e-12,
    ytick={1e0,1e-4,1e-8,1e-12},
    xlabel={$q$ (keys)},
    ylabel={$P(s\text{-collision})$ upper bound},
    legend style={font=\footnotesize, at={(0.5,-0.42)}, anchor=north, draw=none, fill=none, legend columns=1},
    tick label style={font=\footnotesize},
    label style={font=\footnotesize},
    axis lines=left,
]
\addplot+[mark=*, mark size=1pt, blue, thick] coordinates {
(30,1.105654e-12) (32,6.805577e-12) (35,7.257980e-11) (38,5.564498e-10) (41,3.308337e-09) (45,2.617650e-08) (48,1.025819e-07) (52,5.192593e-07) (57,3.039111e-06) (61,1.054273e-05) (67,5.418009e-05) (72,1.782092e-04) (78,6.269465e-04) (85,2.237807e-03) (92,6.715136e-03) (100,1.978561e-02) (107,4.479298e-02) (116,1.105294e-01) (126,2.562895e-01) (136,5.156355e-01) (148,1.000000e+00)
};
\addlegendentry{ours: $n\,P(\mathrm{Bin}(q,1/n)\geq s)$}
\addplot+[mark=square*, mark size=1pt, red!70!black, thick] coordinates {
(30,3.004501e-12) (32,2.257928e-11) (35,3.247943e-10) (38,3.357800e-09) (41,2.691289e-08) (45,3.169871e-07) (48,1.673568e-06) (52,1.259946e-05) (57,1.210270e-04) (61,6.236647e-04) (67,5.796380e-03) (72,3.120491e-02) (78,1.980225e-01) (85,1.000000e+00)
};
\addlegendentry{subset bound: $\binom{q}{s}/n^{s-1}$}
\end{axis}
\end{tikzpicture}
\caption{Our bound versus the $s$-subset bound, $s=20$, $n=10$,
log-log: the subset bound is vacuous by $q=85$, ours informative
through $q=147$.}
\label{fig:suzuki}
\end{figure}

\section{Deferred proofs}\label{app:proofs}

\phantomsection\label{pf:truncated}
\begin{proof}[Proof of Lemma~\ref{lem:truncated}]
Write $X=\sum_{j=1}^{q}I_j$ with $I_j$ the indicator that key $j$
lands in the bucket. By symmetry,
$\E[X\ind{X\ge s+1}]
=q\,\E[I_1\ind{X\ge s+1}]
=q\,\PR[I_1=1]\,\PR[X\ge s+1\mid I_1=1]$.
Given $I_1=1$ the occupancy is $1+\Bin(q-1,\tfrac1n)$ by
independence of the keys, so the conditional probability is
$\PR[\Bin(q-1,\tfrac1n)\ge s]$, and $\PR[I_1=1]=1/n$.
\end{proof}

\phantomsection\label{pf:meanoverflow}
\begin{proof}[Proof of Theorem~\ref{thm:meanoverflow}]
By linearity and Proposition~\ref{prop:occupancy},
$\E[O]=n\,\E\pos{X-s}$ with $X\sim\Bin(q,1/n)$. Since
$\pos{X-s}=X\ind{X\ge s+1}-s\,\ind{X\ge s+1}$,
Lemma~\ref{lem:truncated} gives
\begin{equation}\label{eq:mean-overflow-split}
\E\pos{X-s}
= \E\bigl[X\ind{X\ge s+1}\bigr] - s\,\PR[X\ge s+1]
= \tfrac{q}{n}\,\PR\!\left[\Bin(q-1,\tfrac1n)\ge s\right]
  - s\,\PR[X\ge s+1].
\end{equation}
Multiply by $n$.
\end{proof}

\phantomsection\label{pf:capmono}
\begin{proof}[Proof of Lemma~\ref{lem:capmono}]
(i) For $a,b\ge0$ one checks case by case that
$\pos{a-s}+\pos{b-s}\le\pos{a+b-s}$; by induction over buckets,
$O=\sum_b\pos{X_b-s}\le\pos{\sum_b X_b - s}=\pos{q-s}$. For $n=1$ the
single bucket holds $\min(q,s)$ keys and $O=\pos{q-s}$ exactly.
(ii) Couple the two experiments by placing the same $q$ balls and
revealing only the first $q'$ of them: adding one ball increases one
coordinate $X_b$ by $1$, and $O$ is nondecreasing in each coordinate,
so $O^{(q',n)}\le O^{(q'+1,n)}\le\dots\le O^{(q,n)}$ pointwise under
this coupling.
\end{proof}

\phantomsection\label{pf:na}
\begin{proof}[Proof of Lemma~\ref{lem:na}]
The multinomial occupancy vector $(X_1,\dots,X_n)$ is negatively
associated~\cite{joagdev1983negative}; balls-into-bins is the standard
example. Monotone functions acting on disjoint sets of coordinates
preserve negative association~\cite{dubhashi2009concentration}. For
\eqref{eq:stored-keys}: each bucket stores $\min(X_b,s)$
of the $q$ keys and the remaining $\pos{X_b-s}$ overflow, so
$O=\sum_b\pos{X_b-s}=q-\sum_b\min(X_b,s)$.
\end{proof}

\phantomsection\label{pf:tilt}
\begin{proof}[Proof of Lemma~\ref{lem:tilt}]
By Lemma~\ref{lem:na} the occupancy vector is negatively associated
and monotone functions of disjoint coordinates preserve this, so for
every $\theta>0$ the increasing summands $\pos{X_b-s}$ of $O$ satisfy
$\E[e^{\theta O}]\le\prod_b\E\bigl[e^{\theta\pos{X_b-s}}\bigr]
=\E\bigl[e^{\theta\pos{X-s}}\bigr]^{\,n}$.
Markov's inequality applied to $e^{\theta O}$ then gives
$\PR[O>o]\le e^{-\theta o}\,\E[e^{\theta O}]\le\delta$ whenever $o$
is at least the first quotient in \eqref{eq:chernoff-cert}. Hence
$\PR[O>o^*(n)+\varepsilon]\le\delta$ for every $\varepsilon>0$, and
letting $\varepsilon\downarrow0$ gives the claim. The bound for $W$
is identical with the nonincreasing summands $\pos{s-X_b}$.
\end{proof}

\phantomsection\label{pf:levelbound}
\begin{proof}[Proof of Proposition~\ref{prop:levelbound}]
$B$ is a deterministic quantity equal to one of the five candidates,
so it suffices that each candidate $b$ satisfies $\PR[O>b]\le\delta$.
For $B_1$ this is Lemma~\ref{lem:mcdiarmid} with
$t=\sqrt{(q/2)\ln(1/\delta)}$. For $B_2$ this is Markov's
inequality applied to $O\ge0$. For $B_3$ and $B_4$ this is
Lemma~\ref{lem:tilt}, the latter because $O>B_4$ is the event
$W>w^*(n)$ under the identity $O=q-sn+W$ of
\eqref{eq:stored-keys}. For $\pos{q-s}$ the
probability is $0$ by Lemma~\ref{lem:capmono}(i). Finally, for $n=1$
the occupancy is $X=q$ almost surely, so the infima in
\eqref{eq:chernoff-cert} evaluate to $o^*(1)=\pos{q-s}$ and
$w^*(1)=\pos{s-q}$, giving $B_3=B_4=\pos{q-s}$. By
Lemma~\ref{lem:capmono}(i) the overflow of a single bucket equals
$\pos{q-s}$, so $\E[O]=\pos{q-s}$ and the two remaining candidates,
$B_1=\pos{q-s}+\sqrt{(q/2)\ln(1/\delta)}$ and
$B_2=\pos{q-s}/\delta$, are at least $\pos{q-s}$ as well. The
minimum is therefore the cap, which the overflow attains.
\end{proof}

\phantomsection\label{pf:efficient}
\begin{proof}[Proof of Theorem~\ref{thm:efficient}]
(a) For $i<L$ the constraint in \eqref{eq:efficient-ni} and
$o_{i+1}\le B(o_i,n_i,\delta)$ give
$s\,n_i\le(o_i-o_{i+1})/\alpha$; summing the telescoping differences,
$s\sum_{i<L}n_i\le(o_1-o_L)/\alpha=(q-o_L)/\alpha$, and the final
level adds $s$.

(b) Let $O_i$ be the overflow of level $i$ and let
$G_i:=\{O_i\le o_{i+1}\}$ for $i<L$. Condition on any realization of
the randomness of levels $1,\dots,i-1$ lying in
$G_1\cap\dots\cap G_{i-1}$; under it, the number $o'$ of keys reaching
level $i$ is a fixed integer with $o'\le o_i$, and the level-$i$
placements are fresh and uniform. By Lemma~\ref{lem:capmono}(ii)
(stochastic monotonicity in the number of keys) and
Proposition~\ref{prop:levelbound},
\begin{equation}\label{eq:efficient-onelevel}
\PR\bigl[O_i>o_{i+1}\,\big|\,\cdots\bigr]
\;\le\;\PR\bigl[O^{(o_i,n_i)}>o_{i+1}\bigr]
\;=\;\PR\bigl[O^{(o_i,n_i)}>B(o_i,n_i,\delta)\bigr]
\;\le\;\delta,
\end{equation}
where $O^{(o_i,n_i)}$ denotes the overflow of level $i$ fed with exactly
$o_i$ keys, and the middle equality uses that $O^{(o_i,n_i)}$ is an
integer, so exceeding $\lfloor B\rfloor$ is the same as exceeding $B$.
Multiplying the conditional bounds exactly as in
\eqref{eq:perfect-union}, $\PR[\bigcap_{i<L}G_i]\ge(1-\delta)^{L-1}$. On
$\bigcap_{i<L}G_i$ the final level receives at most $o_L\le s$ keys
into a bucket with $s$ slots, so nothing cascades past it and every
key is stored.

(c) Divide $q$ by \eqref{eq:capacity-bound}, drop $o_L\ge0$, and use
$1/(1+x)\ge1-x$ together with $\alpha\le 1$.
\end{proof}

\section{Deletion churn: model and measurements}\label{app:churn}

Numbers here and in Section~\ref{sec:churn} come from our simulation
of the cascades of Section~\ref{sec:construction} and from the
implementation's churn tests: fill a table with $N$ random keys,
then run $20$ turnovers; a run halts at the first terminal overflow.
Unless stated otherwise, every figure in this appendix assumes the
worst case of a delete followed by an insert at every step at constant $m$; batched or
mixed workloads are milder for the reason given under
\emph{Choosing a strategy}.

\paragraph{The churn model.}
Under churn at $m$ keys, a bucket at level $i$ receives its share of
the flow reaching its level, each key staying a random time of mean
$m$ steps, and turns keys away while full: a loss system with $s$
servers. At offered load $a$ (arrivals per mean lifetime) the
occupancy is Poisson truncated at $s$ and the fraction of arrivals
turned away, the blocking probability of an $s$-server loss system,
is
\begin{equation}\label{eq:loss}
P_s(a) \;=\; \frac{a^s/s!}{\sum_{k=0}^{s} a^k/k!}\,.
\end{equation}
Two facts about churn are exact. Write $f_i$ for the fraction of
full buckets at level $i$ at a given step. First, because the
bucket indices at the $L$ levels are independent, the insert of
that step reaches level $i$ with probability $\rho_i=\prod_{j<i}f_j$
and overflows the last level with probability $\prod_{j\le L}f_j$.
Second, once the
history of the $f_j$ is fixed, each stored key sits at level $i$ or
deeper independently, with the $\rho_i$ in force at its arrival; the
count there is $\Bin(m,\rho_i)$ under a constant flow, and any
excess variance at a deep level is a fluctuation of the flow feeding
it, not of its buckets. The mean field sets these fluctuations to
zero: $\rho_i$ becomes its stationary mean, the fraction of inserts
that reach level $i$, and $y_i=m\rho_i/n_i$ is the load offered to
one of the $n_i$ buckets of level $i$, in arrivals per mean
lifetime. Each bucket is then the loss system above, blocking the
fraction \eqref{eq:loss}, so
\begin{equation}\label{eq:churn-mf}
\rho_{i+1}=\rho_iP_s(y_i),\qquad \rho_1=1.
\end{equation}
We quantify failure with the \emph{failure rate} $D$, the probability
per step that the insert overflows the last level, which is a
terminal overflow. By
the first fact it is the average of $\prod_{j\le L}f_j$ over steps,
and in the model it is the flow out of level $L$, $\rho_{L+1}$ in
the mean field. The values of $D$ quoted in this appendix, including
those behind Table~\ref{tab:churn-floor}, come from
\eqref{eq:churn-mf} with the fluctuation correction described under
\emph{Fluctuations and validation}, which matters only for shapes
with thin tails (few buckets in the last levels). We call a cascade safe when
$D<10^{-12}$ per step. At any finite load a bucket is full only part
of the time, so a table with no slot to spare cannot be sustained.
The perfect table
of Section~\ref{sec:perfect} with every slot occupied is the
extreme case: it accepts the first replacement insert only if it
hashes into the freed slot's bucket. The freed slot sits at level
$j$ with probability $s\,n_j/q$, and the new key, blocked by the
full buckets above, reaches that bucket with probability $1/n_j$;
summing over the $L$ levels, the insert succeeds with probability
$Ls/q$ and fails with probability $D=1-Ls/q$, $0.975$ for the $q=7680$, $s=16$, $L=12$ cascade of
Remark~\ref{rem:perfect-example} ($0.964$ in $500$ trials). The
$2^{20}$ shape $(68907,6399,355,1)$ sized directly for slot load
$0.866$ at $s=16$ fails after about $6.7$ thousand steps, $0.006$
of a turnover.

\paragraph{Choosing a strategy.}
Which option of Section~\ref{sec:churn} to use depends on how many
deletes and inserts are expected relative to the table size, and on
whether they come in batches or continuously. Batch churn, a block
of deletes followed by a block of inserts, is the milder regime:
the inserts meet a level $1$ that the deletes have partly emptied,
which blocks fewer of them, so less flow reaches the deep levels
$2$ to $L$, the tail in the sense of Section~\ref{sec:depth} (level
$1$ being the head). So we
take the failure rate $D$ of the model above as an upper bound for
batch churn,
not an estimate; in our batch scans (one seed per shape) no shape
predicted safe failed, while several predicted unsafe survived.
Light or batch churn is well served by hoisting alone, with no resizing and no extra memory; heavy continuous churn calls for load-factor overload, whose shape absorbs the churn with no pass at all; growth suits medium to high churn or an unknown volume, since no stored key ever moves.

\paragraph{Growing and hoisting.}
Growing can also rehash into a larger table by an increment of our
choosing, as SwissTable does at its load threshold but later on
average, since the rehash waits for the tail to fill; that variant
loses stability during the rehash. The hoist pass re-inserts the keys of
levels $2$ to $L$ into level $1$ in an order independent of their
hashes, giving a one-shot placement of the current key set, and
retries the failed insert. In the $10^5$-key example (the cascade
$(7334,4962,1401,65,1)$ of Section~\ref{sec:quantiles}, $s=8$,
realized load $0.908$) it re-places
the $43\%$ of keys below level $1$ at about $2.1$ probes each and
moves a quarter of them; at $m=q$ the next overflow follows after
about $0.025q$ steps ($36$ probes per step amortized), at $0.95q$
and $0.9q$ after $0.16q$ and about $0.5q$ steps (under $5$ and under
$1$ probe per step). It recurs because a freshly packed level $1$
has about $47\%$ of its buckets full and blocks that share of
newcomers, against $16\%$ at the churn equilibrium: repacking
triples the inflow into the tail.

\paragraph{Load-factor overload.}
Why does a table sized for its load fail under churn, and why does
sizing it for a higher load help? The sizing recipe of
Section~\ref{sec:quantiles}, which the implementation follows,
sizes level $1$ for the target load and hands down the residue that
target leaves; each deeper level hands down its own overflow
estimate. A table sized directly for $\alpha$ therefore puts nearly
every bucket in level $1$ and closes with a two- or three-level
tail. Under churn that tail receives far more
flow than it was sized for.
In the $2^{20}$ shape $(68907,6399,355,1)$ sized directly for its
load, level $1$ blocks $P_{16}(15.2)=15.1\%$ of inserts under
churn where the one-shot residue it was sized for is $7.9\%$; each
level below is overloaded in turn, the last bucket is offered a
load of $5.9\times10^{4}$, and the model predicts
$D=5.6\times10^{-2}$ per step. The $21$ shapes sized directly for
their targets, $0.84$ down to $0.67$, all failed within $0.03$
turnovers at $2^{20}$ keys and $0.15$ at $20$ thousand.

The remedy keeps the slot count and moves buckets down. We run the
recipe for $q(1+x)$ keys at a target $\alpha'$ above the intended
load $\alpha$ and hold $q$ keys. Here $\alpha'$ is the slot-load
target of Definition~\ref{def:efficient}, which the recipe realizes
to within $0.3\%$, so the table gets $q(1+x)/\alpha'$ slots, and
holding it to the $q/\alpha$ slots of a direct build fixes
$1+x=\alpha'/\alpha$. Hence $x$ is not a free margin but follows
from $\alpha'$ (at $\alpha=0.80$, $\alpha'=0.934$ gives $x=0.17$),
and $\alpha'$ alone sets how the slots are distributed, through the
residue mechanism above. At $2^{20}$ keys, $s=16$ and $x=0.18$, a
percent above the equal-slot value, it yields
$(n_i)=(75837,9389,1023,1)$ at $\alpha'=0.90$
but $(69511,11909,1311,98,1)$ at $\alpha'=0.934$, about $27\%$ more
buckets at levels $2$ and $3$ plus a further level; the latter has
$1325280$ slots against $1310720$ for exact load $0.80$, one percent
more, and a realized load of $0.791$. The mean offered loads explain
why the second shape survives: they decay geometrically down the
cascade, $y_i=15.1,13.0,10.0,3.0$, so each level absorbs most of
what reaches it and almost nothing arrives at the last bucket (the
model predicts $D=3.3\times10^{-23}$). The margin is a knife edge: with
exactly $q/\alpha$ slots the target $0.93$ is unsafe (every shape in
the window has a predicted failure rate of $2\times10^{-9}$ or more), and the
lowest target at which every shape with slots in
$[q/\alpha,\,1.01\,q/\alpha]$ is safe is $\alpha'=0.94$, giving
$(67866,12434,1485,134,1)$ with exactly $1310720$ slots. Below a
floor $\alpha'_{\min}$ the deep levels get too few buckets, the load
on the last level explodes as in the directly sized shape, and the
table fails within a fraction of a turnover.

Table~\ref{tab:churn-floor} gives this floor per $\alpha$ and $s$
for $q=2^{20}$ keys, found as the $0.94$ above, so it already allows for the knife edge.
For each pair we
step $\alpha'$ through $0.70,0.71,\dots,0.99$. At each value we
sweep $x$: starting from the smallest key count $q(1+x)$ whose shape
reaches $q/\alpha$ slots, we raise it in steps of $0.0003\,q$ until
the slot count passes $1.01\,q/\alpha$, which yields $26$ to $50$
distinct shapes, and we score each with the model at $m=q$ keys.
The floor $\alpha'_{\min}$ is the lowest grid value at which every
one of these shapes is safe, and at $2^{20}$ keys every grid value
above it is safe as well. To use the
table, a builder takes $\alpha'=\alpha'_{\min}$ (for $\alpha$
between columns, the column above), runs the recipe for $q(1+x)$
keys with $x=\alpha'/\alpha-1$, or up to a percent more so that the
shape has at least $q/\alpha$ slots, and scores the resulting shape,
raising $\alpha'$ by $0.01$ while it scores unsafe; if no target up
to $0.99$ passes, or the cell reads ``none'', the builder lowers the
intended load or picks another option of Section~\ref{sec:churn}.
Scoring is needed for two reasons. Smaller
tables ($\le 2^{16}$ keys) move some floors by a step or two either
way or leave no safe target, and show pockets above the floor where
shapes sit within a factor $20$ of the safety line (at $s=32$ and
intended load $0.70$, at $0.85$ to $0.89$ and again at $0.97$). And safety is
not monotone in $\alpha'$ in general, since the construction can
leave one bucket behind a large level, so a raised target needs
checking too.
\begin{center}
\footnotesize
\begin{tabular}{lcccccccccc}
intended load $\alpha$ & $0.90$ & $0.88$ & $0.86$ & $0.84$ & $0.82$ & $0.80$ & $0.77$ & $0.74$ & $0.70$ & $0.67$\\
\hline
$\alpha'_{\min}$, $s=8$ & none & none & none & $0.99$ & $0.98$ & $0.97$ & $0.95^\dagger$ & $0.92^\dagger$ & $0.89^\dagger$ & $0.86^\dagger$\\
$\alpha'_{\min}$, $s=16$ & none & $0.99$ & $0.98$ & $0.97^\dagger$ & $0.96$ & $0.94^\dagger$ & $0.92^\dagger$ & $0.89^\dagger$ & $0.85^\dagger$ & $0.81$\\
$\alpha'_{\min}$, $s=32$ & $0.99$ & $0.98$ & $0.97$ & $0.95$ & $0.93$ & $0.91$ & $0.88$ & $0.85$ & $0.81$ & $0.77$\\
\end{tabular}
\captionof{table}{Lowest churn-safe slot-load target
$\alpha'_{\min}$ per intended load $\alpha$ and bucket size $s$ at
$2^{20}$ keys: every shape with $q/\alpha$ to $1.01\,q/\alpha$
slots has $D<10^{-12}$; the overload is $x=\alpha'/\alpha-1$.
``None'' means no safe target up to $0.99$. The constructor takes a
physical load $\varphi'$: $\alpha'=9\varphi'/8$ in the filtered arm
with $4$-byte keys and values, $\alpha'\approx\varphi'$ in the plain
arm. Daggers mark measured cells; the others are
predictions.}\label{tab:churn-floor}
\end{center}
In the churn tests, the $17$ overload shapes behind the table's
daggers ($s=16$ at intended loads $0.84$, $0.80$, $0.77$, $0.74$,
$0.70$; $s=8$ at $0.77$, $0.74$, $0.70$, and at $0.67$ with $20$
thousand keys only) ran at an
$\alpha'$ within one step of the tabulated floor, below it in five
cells and at or above it in four, with about $1\%$ more slots than
$q/\alpha$. They saw no failure in $4000$ seeds at $20$ thousand keys
($1.6\times10^9$ steps per shape, bound $1.87\times10^{-9}$ per
step) or in $100$ seeds at $2^{20}$ ($2.1\times10^9$ steps, bound
$1.43\times10^{-9}$); mean levels were within $0.001$ of the model,
and the single-bucket last level never held more than $2$ keys. At
target $0.80$, $s=16$, $2^{20}$ keys, a live key sits at mean level
$1.160$ under churn, and a hit costs one probe per level, so
positive lookups average $1.16$ probes, $7\%$ more than the $1.08$
of the shape $(68907,6399,355,1)$ sized directly for slot load
$0.866$ after a one-shot fill (mean level $1.082$). Space is
unchanged: an overload shape holds the same slots as a direct build
for the same intended load, within about a percent.

\paragraph{Fluctuations and validation.}
The mean field \eqref{eq:churn-mf} reproduces the mean level of a
stored key to within $0.001$
but underestimates the failures of shapes with thin tails by as much as
thirteen orders of magnitude: at $m=20000$ the
test shape $(1335,191,29,6,1)$ is predicted to fail
$4.1\times10^{-20}$ times per step where $3.9\times10^{-7}$ were
measured. The gap arises because the load offered to a deep level
is itself random and $aP_s(a)$ is convex below
$s$: a fluctuating load overflows more than its mean would. Two
mechanisms feed the randomness. First, the overflow of a loss
system is burstier than its input: writing the \emph{peakedness} of
a count for its variance-to-mean ratio ($1$ for Poisson),
Wilkinson~\cite{wilkinson1956toll} gives the overflow of a
Poisson-fed $s$-server loss system at load $a$ the peakedness
$z^{W}(a)=1-M+a/(s+1-a+M)$ with $M=aP_s(a)$.
Second, a relative fluctuation of the offered load is amplified by
the elasticity $\eta_i=s-y_i(1-P_s(y_i))$ of the
blocking probability, about $3$ at level $1$ and about $15$ at a
nearly empty last level. We follow the squared coefficient of
variation $v_i$ of the lifetime-smoothed flow into level $i$ with
\begin{equation}\label{eq:churn-var}
v_{i+1}=\frac{z^{W}(y_i)-1}{m\rho_{i+1}}+\Bigl(1+\tfrac{\eta_i}{\sqrt2}\Bigr)^{2}v_i ,
\end{equation}
where the factor $1/\sqrt2$ records that a fluctuation already
smoothed over one key lifetime keeps half its variance under a
second smoothing. Level $1$ is pinned at $m$ keys, so the flow into
level $2$ is under-dispersed; a linear-response estimate gives its
peakedness $z_2=1/(1+b_1)$ with
$b_1=y_1P_s(y_1)\,\eta_1/\operatorname{Var}[\Pois(y_1)\text{ truncated at }s]$,
and the recursion starts at $v_2=(z_2-1)/(m\rho_2)$. Treating the
load of level $i$ as Gamma distributed with mean $y_i$ and squared
coefficient of variation $v_i$, and replacing
$\rho_iP_s(y_i)$ in \eqref{eq:churn-mf} by the
Gamma average of $aP_s(a)$, closes the system, and $D$ is the flow
out of level $L$ of the closed system. The two facts of
\emph{The churn model} and the loss formula \eqref{eq:loss} are exact; \eqref{eq:churn-var}, the
level-$1$ feedback and the Gamma closure are heuristics with the
approximations just stated. In the churn tests at $20$ thousand keys,
six test shapes with thin tails (one to seven buckets in the last
level) failed; the predicted $D$ is $2.1$ to $7.6$ times the measured
rate for each, and the predicted peakedness of the tail counts
tracks the measured rise from $0.4$ to $4.7$ down the cascade. All
$17$ overload shapes are predicted safe, the closest at
$1.6\times10^{-13}$. The mean field alone places $\alpha'_{\min}$
within $0.01$ of the full model, so the floor is a mean-field
effect and the fluctuation terms decide only the marginal shapes
with thin tails. The model is pessimistic by the factor just
quoted, it does not model the fill transient, and the Gamma tail is
an assumption that carries $D$ when the last level is nearly empty;
its combination with the depth reduction of Section~\ref{sec:depth}
is open.

\section{Comparison with funnel hashing}\label{app:funnel}

Funnel hashing is the second of two constructions in Farach-Colton,
Krapivin and Kuszmaul~\cite{farachcolton2025funnel}; the first,
elastic hashing, halves its arrays too but probes each one uniformly
and non-greedily, and is the more distant relative.
The scheme is stated for an array of $n$ slots receiving
$n-\lfloor\delta n\rfloor$ keys, so their $\delta$ is the vacancy
and not a failure probability; their failure probabilities are
$1/\mathrm{poly}(n)$ and $n^{-\omega(1)}$ throughout. We state the
comparison in our notation: their $n$ is our capacity
$C=s\sum_i n_i$, their $\delta$ is our $1-\alpha$, their bucket size
$\beta$ is our $s$, and their count of arrays before the special
array is our $L-1$. With that translation the skeletons coincide.
Their array is split into $L-1=\lceil4\log(1/(1-\alpha))+10\rceil$
levels of buckets of $s=\lceil2\log(1/(1-\alpha))\rceil$ slots,
with $n_{i+1}/n_i\approx\tfrac34$, and a key is inserted by hashing
to one bucket of the first level, scanning it for a free slot, and
on failure to one bucket of the second, and so on, a lookup
following the same path: this is the cascade of
Section~\ref{sec:model}. The similarity ends at the tail: keys that
fail in every level enter a special array of between
$\lceil(1-\alpha)C/2\rceil$ and $\lfloor3(1-\alpha)C/4\rfloor$
slots, itself split into a uniform-probing half that gives up after
$\log\log C$ probes and a two-choice half with buckets of
$2\log\log C$ slots, whereas the multitable ends in a single
bucket. Moreover, their scheme doesn't include deletions while multitable manages them, including while keeping stability.

\paragraph{Parameters and depth.}
Funnel hashing has two free parameters, $C$ and $\alpha$; bucket
size, level count and shrink ratio follow from the proof, with its
constants: at $1-\alpha=0.1$ the formulas give $s=7$ and $L-1=24$
levels before the special array, at $1-\alpha=2^{-5}$ they give
$s=10$ and $L-1=30$, and at $1-\alpha=2^{-10}$ they give $s=20$
and $L-1=50$ (logarithms to base $2$). The multitable takes
the bucket size $s$ from the hardware's access unit and the workload's
hit-miss trade (Sections~\ref{sec:intuition}
and~\ref{sec:implementation}), the load factor $\alpha$ and the
per-level budget $\delta$ from the application, and the overfill
factor $r$ from the workload mix (Section~\ref{sec:depth}), and
sizes every level from the tail of its own occupancy law, taking the
largest bucket count at which the $(1-\delta)$-quantile bound still
meets the target load \eqref{eq:efficient-ni}, in place of a fixed
ratio between consecutive levels. Progressive overfill, the
truncation of Corollary~\ref{cor:shortcircuit} and the
near-exact quantile of Section~\ref{sec:quantiles} shorten the
cascade further. Sized by the near-exact quantile for $q=10^5$ keys at
$\delta=2^{-10}$ and $r=1.1$, with the funnel's bucket size, the
multitable needs $5$ levels at $1-\alpha=0.1$ and $s=7$ (realized
load $0.906$), against the funnel's $24$ plus the special array,
and $6$ levels at $1-\alpha=2^{-5}$ and $s=10$ (realized load
$0.971$), against $30$ plus the special array. The perfect table of
Remark~\ref{rem:perfect-example} stores $64000$ keys at load exactly
$1$ in $9$ levels, a load outside the funnel theorem's range, which
needs $\alpha<1$. The depths scale differently. The funnel
level count depends on $1-\alpha$ alone, because the special array
absorbs whatever remains after a fixed number of levels, whereas the
multitable cascade, which must shrink to a single bucket, has
$O(\log q)$ levels for a fixed configuration
(Proposition~\ref{prop:depth}); the comparison above is one of
practical sizes, not of asymptotics.

\paragraph{Analysis.}
The analyses take different routes to similar quantities. Their
Lemma~5 bounds the slots left unfilled in a level after twice its
size in insertion attempts, by a Chernoff bound on the attempts each
bucket receives and McDiarmid's inequality on the count of unfilled
buckets, and their Lemma~6 deduces by counting that fewer than
$(1-\alpha)C/8$ keys reach the special array; the oversampling
factor $2$ and the unfilled fraction $(1-\alpha)/64$ are fixed by
the argument.
The multitable starts from the exact occupancy law $\Bin(q,1/n)$ of
a bucket, derives the exact mean overflow
(Theorem~\ref{thm:meanoverflow}), bounds its upper tail four ways
(Proposition~\ref{prop:levelbound}), and sizes each level at the
quantile bound so obtained. The resulting guarantee is finite and
explicit, $(1-\delta)^{L-1}$ per build
(Theorems~\ref{thm:perfect} and~\ref{thm:efficient}), where theirs
is asymptotic in $C$. Two products of this route have no counterpart
in funnel hashing: the near-exact quantile $\widehat{Q}$ of
Section~\ref{sec:quantiles}, within $2$ keys of the simulated
quantile where the proven envelope runs more than twice as high at
low load (Figure~\ref{fig:approx}), and the multicollision bound of
Appendix~\ref{app:multicollision}, which sharpens the classical
$s$-subset bound by orders of magnitude once keys outnumber buckets.

\paragraph{Purpose and results.}
Funnel hashing settles the worst-case expected probe complexity of
greedy open addressing without reordering at
$\Theta(\log^2(1/(1-\alpha)))$, with matching lower bounds, disproving
a conjecture of Yao, and reports no implementation or measurement.
The multitable's figure of merit is the physical load factor,
payload bytes over allocated bytes, which charges the metadata that
slot-load accounting ignores; the construction is implemented
(Section~\ref{sec:implementation}) and, at equal physical memory,
measured ahead of hashbrown in all $84$ benchmarked configurations
(Section~\ref{sec:performance}). The
framed variant of Section~\ref{sec:variations}, which keeps every
lookup inside one region of a paged or sharded medium, has no
analogue in the funnel construction either.

\end{document}